\documentclass[10pt]{amsart}
\usepackage{amsrefs}

\usepackage{etex}
\usepackage{amsfonts}
\usepackage{amsmath}
\usepackage[latin1]{inputenc}
\usepackage{setspace}
\usepackage{mathrsfs}
\usepackage{graphicx}
\graphicspath{ {./pics/} }
\usepackage{caption}
\usepackage{subcaption}
\usepackage{stmaryrd}
\usepackage[english]{babel}
\usepackage{enumitem}
\usepackage{amssymb}
\usepackage{amsthm}
\usepackage{bbding}
\usepackage[all]{xy}
\usepackage[a4paper]{geometry} 

\usepackage{setspace}
\usepackage[bottom]{footmisc} 				

\usepackage{tikz-cd}
\usepackage{pstricks}
\usepackage{pst-slpe}
\usepackage{multido}							
\usepackage{pstricks-add}
\usepackage{pst-plot}

\theoremstyle{plain}
\newtheorem{theorem}{Theorem}[section]
\newtheorem{lemma}[theorem]{Lemma}
\newtheorem{proposition}[theorem]{Proposition}
\newtheorem{corollary}[theorem]{Corollary}
\newtheorem{definition}[theorem]{Definition}

\theoremstyle{definition}
\newtheorem*{definition*}{Definition}

\newtheorem{example}[theorem]{Example}
\newtheorem{remark}[theorem]{Remark}

\theoremstyle{remark}

\newcommand{\Z}{\mathbb{Z}}
\newcommand{\C}{\mathbb{C}}

\newcommand{\End}{\operatorname{End}}

\definecolor{line}{rgb}{0,0,1} 
\definecolor{diskin}{gray}{0.8} 

\def\ltileb{
\psset{unit=0.8cm}
\begin{pspicture}[shift=-0.35](0,0)(1,1)
		\psframe[fillstyle=solid,fillcolor=diskin](0,0)(1,1)
		\psarc[linecolor=blue,linewidth=1.2pt](0,1){0.5}{270}{0}	
		\psarc[linecolor=blue,linewidth=1.2pt ](1,0){0.5}{90}{180}	
\end{pspicture} 
}

\def\ltilea{
\psset{unit=0.8cm}
\begin{pspicture}[shift=-0.35](0,0)(1,1)
		\psframe[fillstyle=solid,fillcolor=diskin](0,0)(1,1)
		\psarc[linecolor=blue,linewidth=1.2pt](0,0){0.5}{0}{90}	
		\psarc[linecolor=blue,linewidth=1.2pt](1,1){0.5}{180}{270}	
\end{pspicture} 
}
\def\cross{{
\psset{unit=0.25cm}
\begin{pspicture}[shift=-0.5](-1,-1)(1,1)
			\pscircle[fillstyle=solid,fillcolor=diskin,linewidth=0.8pt,linestyle=dotted](0,0){1}
			\degrees[8]
			\psline[linecolor= line,linewidth=1pt](1;1)(1;5)
			\psline[linecolor= blue,linewidth=1pt](1;3)(0.25;3)
			\psline[linecolor= blue,linewidth=1pt](1;7)(0.25;7)	
		\end{pspicture}
		\psset{unit=0.5cm}}
}

\def\id{{	
\psset{unit=0.25cm}
\begin{pspicture}[shift=-0.5](-1,-1)(1,1)
			\pscircle[fillstyle=solid,fillcolor=diskin,linewidth=0.8pt,linestyle=dotted](0,0){1}
			\degrees[8]
			\pscurve[linecolor= line,linewidth=1.2pt](1;5)(0.9;5)(0.4;4)(0.9;3)(1;3)
			\pscurve[linecolor= line,linewidth=1.2pt](1;1)(0.9;1)(0.4;0)(0.9;7)(1;7)
		\end{pspicture} 
		\psset{unit=0.5cm}
		}}

\numberwithin{equation}{section}

\title{Towers of solutions of qKZ equations and their applications to loop
models}
\date{\today}

\author{K. Al Qasimi}
\address{KdV Institute for Mathematics, University of Amsterdam, Science Park 105-107, 1098 XG Amsterdam, The Netherlands.}
\email{Kayed.AlQasimi@gmail.com}

\author{B. Nienhuis}
\address{Institute of Physics, University of Amsterdam, Science Park 904, 1098 XG Amsterdam, 
The Netherlands, currently: Instituut-Lorentz for Theoretical Physics, P.O. Box 9506, NL-2300 RA Leiden, The Netherlands.}
\email{Nienhuis@iLorentz.org}

\author{J.V. Stokman}
\address{KdV Institute for Mathematics, University of Amsterdam, Science Park 105-107, 1098 XG Amsterdam, The Netherlands.}
\email{j.v.stokman@uva.nl}

\begin{document}
\begin{abstract}
Cherednik's type A quantum affine Knizhnik-Zamolodchikov (qKZ) equations form  a consistent system of linear $q$-difference equations for $V_n$-valued meromorphic functions on a complex $n$-torus, with 
$V_n$ a module over the $\textup{GL}_n$-type extended affine Hecke algebra $\mathcal{H}_n$.
The family $(\mathcal{H}_n)_{n\geq 0}$ of extended affine Hecke algebras forms a tower of algebras, with the associated algebra morphisms $\mathcal{H}_n\rightarrow\mathcal{H}_{n+1}$ in the Hecke algebra descending of arc insertion at the affine braid group level.  In this paper we consider 
qKZ towers $(f^{(n)})_{n\geq 0}$ of solutions, which consist of 
twisted-symmetric polynomial solutions $f^{(n)}$ ($n\geq 0$) of the qKZ equations that are compatible
with the tower structure on $(\mathcal{H}_n)_{n\geq 0}$. The compatibility is encoded by so-called braid recursion relations: $f^{(n+1)}(z_1,\ldots,z_{n},0)$ is required to coincide up to a quasi-constant factor with the push-forward of $f^{(n)}(z_1,\ldots,z_{n})$ by an intertwiner $\mu_{n}: V_{n}\rightarrow V_{n+1}$ of $\mathcal{H}_{n}$-modules, where $V_{n+1}$ is considered as an $\mathcal{H}_{n}$-module through the tower structure on $(\mathcal{H}_n)_{n\geq 0}$. 

We associate to the dense loop model on the half-infinite cylinder with nonzero loop weights a qKZ tower $(f^{(n)})_{n\geq 0}$ of solutions. The solutions
$f^{(n)}$ are constructed from specialised dual non-symmetric Macdonald polynomials with specialised parameters using the Cherednik-Matsuo correspondence. In the special case that the extended affine Hecke algebra parameter is a third root of unity, $f^{(n)}$ coincides with the (suitably normalized) ground state of the inhomogeneous dense O(1) loop model on the half-infinite cylinder with circumference $n$.
\end{abstract}
\maketitle

\setcounter{tocdepth}{2}

\section{Introduction}

Quantum Knizhnik-Zamolodchikov (qKZ) equations are consistent systems of linear $q$-dif\-fe\-rence equations that naturally arise in the context of representation theory of quantum affine algebras \cite{Frenkel:1992aa} and affine Hecke algebras \cite{Cherednik:1992aa}. They appear as
consistency equations for form factors and correlation functions of various integrable models (see e.g., \cite{Smirnov:1992aa,Jimbo:1995aa} for the first examples). In this paper we focus on Cherednik's qKZ equations associated to the
$\textup{GL}_n$-type extended affine Hecke algebra $\mathcal{H}_n$. This case relates to integrable one-dimensional lattice models with quasi-periodic boundary conditions, with the integrability governed by the extended affine Hecke algebra $\mathcal{H}_n$. Important examples, also in the context of the present paper, are the XXZ spin-$\frac{1}{2}$ chain and the dense loop model.

The collection $(\mathcal{H}_n)_{n\geq 0}$ of extended affine Hecke algebras forms a tower of algebras with respect to algebra morphisms $\mathcal{H}_n\rightarrow\mathcal{H}_{n+1}$ that arise as descendants of arc insertion morphisms $\mathcal{B}_n\rightarrow\mathcal{B}_{n+1}$ for the groups $\mathcal{B}_n$ of affine $n$-braids, cf. \cite{Al-Qasimi:2017aa, Al-Harbat:2015aa, Gainutdinov:2016aa}. In this paper, we study families $(f^{(n)})_{n\geq 0}$ of solutions $f^{(n)}$ of qKZ equations taking values in $\mathcal{H}_n$-modules $V_n$ that are naturally compatible to the tower structure. 

It leads us to introducing the notion of a tower $(f^{(n)})_{n\geq 0}$ of solutions of qKZ equations. The constituents $f^{(n)}$ of the tower are polynomials in $n$ complex variables $z_1,\ldots,z_n$, taking values in a finite dimensional $\mathcal{H}_n$-module $V_n$. 
They are twisted-symmetric solutions of Cherednik's qKZ equations interrelated by the so-called {\it braid recursion relations},
meaning that $f^{(n+1)}(z_1,\ldots,z_{n},0)$ coincides with the push-forward of $f^{(n)}(z_1,\ldots,z_{n})$ by an $\mathcal{H}_n$-intertwiner $\mu_{n}: V_{n}\rightarrow V_{n+1}$ up to a quasi-constant factor, where $V_{n+1}$ is regarded as an $\mathcal{H}_{n}$-module through the tower structure of $(\mathcal{H}_n)_{n\geq 0}$.
In the terminology of \cite{Al-Qasimi:2017aa}, the collection $\{(V_n,\mu_n)\}_{n\geq 0}$ of $\mathcal{H}_n$-modules $V_n$ and $\mathcal{H}_n$-intertwiners $\mu_n: V_n\rightarrow V_{n+1}$ is a tower of extended affine Hecke algebra modules. {}From this perspective, towers of solutions of qKZ equations are naturally associated with towers of extended affine Hecke algebra modules. The braid recursion relations are then determined by the module tower up to the quasi-constant factors.

In \cite{Al-Qasimi:2017aa}, the first and third authors constructed a family of module towers, called {\it link pattern towers}, which depends on a twist parameter $v$. The link pattern tower actually descends to a tower of extended affine Temperley-Lieb algebra modules. The representations $V_n$ are realized on spaces of link patterns on the punctured disc, which alternatively can be interpreted as the quantum state spaces for the dense O($\tau$) loop models on the half-infinite cylinder (with $n$ the circumference of the cylinder). The intertwiners $\mu_n$ in the link pattern tower are constructed skein theoretically (for even $n$ this goes back to 
\cite{Di-Francesco:2006aa}), and are in fact closely related to arc insertion morphisms in a relative version of the Roger-Yang \cite{Roger:2014aa} skein module in the presence of a pole (see \cite[Rem. 8.11]{Al-Qasimi:2017aa}). In this paper, we construct towers $(f^{(n)})_{n\geq 0}$ of solutions of qKZ equations relative to the link pattern tower with twist parameter one, and describe the corresponding quasi-constant factors in the braid recursion relations explicitly. We consider two cases. 

We show that the (suitably normalized) ground states $f^{(n)}$ of the inhomogeneous dense O(1) loop model on the half-infinite cylinder with circumference $n$ form a tower of solutions relative to the link pattern tower. In this case, the associated affine Hecke algebra parameter is a third root of unity. 
This generalizes results from \cite{Di-Francesco:2006aa}, where the braid recursion relations relating $f^{(2k+1)}(z_1,\ldots,z_{2k},0)$ to $f^{(2k)}(z_1,\ldots,z_{2k})$ were derived under the implicit additional assumption that a unique normalized ground state for the inhomogeneous dense O(1) loop model exists when one of the rapidities is set equal to zero (the latter is not guaranteed, since the transfer operator is no longer stochastic when one of the rapidities is set equal to zero).  In an upcoming paper \cite{Al-Qasimi:2019aa} the full set of braid recursion relations for the ground states is used to derive explicit formulas for various observables of the dense O(1) loop model on the infinite cylinder.

We generalize this example by constructing a tower of solutions $(f^{(n)})_{n\geq 0}$ 
for twist parameter one and for all values of the affine Hecke algebra parameter for which the loop
weights of the associated dense loop model are nonzero. In this case, the constituents $f^{(n)}$ are constructed using the Cherednik-Matsuo correspondence \cite{Kasatani:2007ab,Stokman:2011aa}. The Cherednik-Matsuo correspondence, relating solutions of qKZ equations to common eigenfunctions of Cherednik's commuting $Y$-operators, can be applied in the present context since the link pattern modules are principal series modules, as we shall show in Theorem \ref{principalTHM}. It leads to the construction of the constituents $f^{(n)}$ of the tower
in terms of nonsymmetric Macdonald polynomials. Subtle issues arise here since the two parameters of the associated double affine Hecke algebra satisfy an algebraic relation that breaks down the semisimplicity of the $Y$-operators. We resort to Kasatani's \cite{Kasatani:2005aa} work to deal with these issues. See \cite{Kasatani:2007aa} for an alternative approach to construct polynomial twisted-symmetric solutions $f^{(n)}$ of the qKZ equations 
using Kazhdan-Lusztig bases.

In both towers the constituent $f^{(n)}$ is a nonzero twisted-symmetric homogeneous polynomial solution of the qKZ equations of total degree $\frac{1}{2}n(n-1)$. In fact, this property characterizes $f^{(n)}$ up to a nonzero scalar multiple, a result that plays a crucial role in establishing the explicit braid recursion relations. In particular it allows us to prove the braid recursion relations for the suitably normalized ground states of the inhomogeneous dense O(1) loop models without addressing the issue of the existence of a unique normalized ground state when the rapidities are outside the stochastic regime.

The content of the paper is as follows. 
In Section 2 we recall the definitions of extended affine Hecke algebras and qKZ equations, and introduce the notion of a qKZ tower of solutions.
In Section 3 we recall from \cite{Al-Qasimi:2017aa} the definition of the link pattern tower.
In Section \ref{subsec-linkqkz} we determine necessary conditions for the existence of nonzero twisted-symmetric homogeneous polynomial solutions $f^{(n)}$ of total degree $\frac{1}{2}n(n-1)$ of the qKZ equations with values in the link pattern modules. We show that the existence implies that 
$(f^{(n)})_{n\geq 0}$ forms a tower of solutions relative to the link pattern tower, and we explicitly write down the corresponding braid recursion relations.
The construction of the tower of solutions when the Hecke algebra parameter is a third root of unity is discussed in Section 5. The general case is discussed in Section 6.
We derive a dual version of the braid recursion relations in Section \ref{dual}.
Lastly, in Appendix \ref{uniqueproof} we discuss uniqueness properties for various classes of twisted-symmetric solutions to qKZ equations,
some of which were considered before in \cite{Di-Francesco:2006aa, Kasatani:2007aa, De-Gier:2010aa}.

\subsection{Acknowledgments}
We thank Eric Opdam for a comment leading to the precise conditions on the affine Hecke algebra 
parameter for which the main theorem of the paper (Theorem \ref{mainTHM}) holds true.
The work by Kayed Al Qasimi is supported by the Ministry of Education of the United Arab Emirates  under scholarship number 201366644. 
Diagrams were coded using PSTricks.

\section{Towers of solutions of qKZ equations} 
In this section we begin by recalling the extended affine Hecke algebra, the qKZ equations and introduce what we call a qKZ tower of solutions.
The extended affine Hecke algebra can be defined using two different presentations.
We make use of both presentations as one is more convenient for defining qKZ equations, while the other is more suitable for relating the algebra to the extended affine Temperley-Lieb algebra.

\subsection{Extended affine Hecke algebras}

Let $t^{\frac{1}{4}}\in\mathbb{C}^*$.
\begin{definition}\label{defH}
Let $n\geq 3$. The extended affine Hecke algebra $\mathcal{H}_n=\mathcal{H}_n(t^{\frac{1}{2}})$ of type $\widehat{A}_{n-1}$ is the complex associative algebra with generators $T_i$ ($i\in\mathbb{Z}/n\mathbb{Z}$) and $\rho,\rho^{-1}$ and defining relations
\begin{equation}\label{hecke2}
	\begin{split}		
	 &(T_i - t^{-\frac{1}{2}}) ( T_i +t^{\frac{1}{2}}) =0, \\
	 &T_iT_j =T_jT_i \hspace{3cm}(i-j\not\equiv \pm 1),\\
	&T_i T_{i+1}T_i = T_{i+1} T_{i}T_{i+1},\\
	&\rho T_{i} = T_{i+1} \rho,\\
	&\rho\rho^{-1} =1 = \rho^{-1}\rho,
	\end{split}
\end{equation}
where the indices are taken modulo $n$. For $n=2$ the extended affine Hecke algebra 
$\mathcal{H}_2=\mathcal{H}_2(t^{\frac{1}{2}})$ is the algebra generated by $T_0,T_1,\rho^{\pm 1}$ with defining relations \eqref{hecke2} but with the third relation omitted. For $n=1$ we set $\mathcal{H}_1:=\mathbb{C}[\rho,\rho^{-1}]$ to be the algebra of Laurent polynomials in one variable $\rho$, and for $n=0$ we set $\mathcal{H}_0:=\mathbb{C}[X]$, the polynomial algebra in the variable $X$.
\end{definition}
Note that $T_i$ is invertible with inverse $T_i^{-1}=T_i -t^{-\frac{1}{2}} + t^{\frac{1}{2}}$.
For $n\geq 1$ the element $\rho^n\in\mathcal{H}_n$ is central. 

For $n\geq 2$ the {\it affine Hecke algebra} $\mathcal{H}_n^a=\mathcal{H}_n^a(t^{\frac{1}{2}})$ of type $\widehat{A}_{n-1}$ is the subalgebra of $\mathcal{H}_n$ generated by $T_i$ ($i\in\mathbb{Z}/n\mathbb{Z}$). For $n\geq 3$ the first three relations of \eqref{hecke2} are the defining relations of $\mathcal{H}_n^a$ in terms of these generators (for $n=2$ the first two relations are the defining relations).  Furthermore, $\mathcal{H}_n$ is isomorphic to the crossed
product algebra $\mathbb{Z}\ltimes\mathcal{H}_n^a$,  where $m \in\mathbb{Z}$ acts on 
$\mathcal{H}^a_n$
by the algebra automorphism $T_i\mapsto T_{i+m}$ (with the indices modulo $n$). Equivalently,
$m\in\mathbb{Z}$ acts by restricting the inner automorphism $h\mapsto \rho^mh\rho^{-m}$ of $\mathcal{H}_n$ to $\mathcal{H}_n^a$. 
For $n\geq 2$ the (finite) Hecke algebra of type $A_{n-1}$ is the subalgebra $\mathcal{H}^0_n$ of $\mathcal{H}^a_n$ generated by $T_1,\ldots, T_{n-1}$. The defining relations of $\mathcal{H}_n^0$
in terms of the generators $T_1,\ldots,T_{n-1}$ are given again by the first three relations of
\eqref{hecke2}, restricted to those indices that they make sense.

Bernstein and Zelevinsky \cite{Lusztig:1989aa} obtained the following alternative presentation of 
the extended affine Hecke algebra (see also \cite{Haines:2002aa} for a detailed discussion).

\begin{theorem}
Let $n\geq 2$ and define $Y_j\in\mathcal{H}_n$ for $j=1,\ldots,n$ by
\[
Y_j:=T_{j-1}^{-1}T_{j-2}^{-1}\cdots T_1^{-1}\rho T_{n-1}\cdots T_{j+1}T_j.
\]
Then, $\mathcal{H}_n$ is generated by $T_1,\ldots,T_{n-1},Y_1^{\pm 1},\ldots,Y_n^{\pm 1}$.
The defining relations of $\mathcal{H}_n$ in terms of these generators are given by
\begin{align} \label{hecke1}
	&(T_i - t^{-\frac{1}{2}}) ( T_i +t^{\frac{1}{2}}) =0\qquad &&(1\leq i<n),\nonumber \\
	&T_i T_{i+1}T_i = T_{i+1} T_{i}T_{i+1}\qquad  &&(1\leq i<n-1),\nonumber \\
	&T_iT_j =T_jT_i &&(1\leq i,j<n: |i-j|>1),\nonumber\\
	&T_iY_{i+1}T_i = Y_{i} &&(1\leq i<n),\\
	&T_iY_j=Y_jT_i &&(1\leq i<n, 1\leq j\leq n: j\neq i,i+1),\nonumber \\
	&Y_i Y_j =Y_jY_i && (1\leq i,j\leq n),\nonumber \\
	&Y_iY^{-1}_i =1 = Y^{-1}_iY_i && (1\leq i\leq n). \nonumber
\end{align}
\end{theorem}
Note that $\rho\in\mathcal{H}_n$ can be expressed as
\[
\rho=T_1T_2\cdots T_{n-1}Y_n
\]
with respect to the Bernstein-Zelevinsky presentation of $\mathcal{H}_n$. Let $\mathcal{A}_n$
be the commutative subalgebra of $\mathcal{H}_n$ generated by $Y_1^{\pm 1},\ldots,Y_n^{\pm 1}$.

More can be said about the structure of $\mathcal{H}_n$ in terms of the Bernstein-Zelevinsky presentation (see \cite{Lusztig:1989aa} and \cite{Haines:2002aa}).
Let $f\in \mathbb{C}[\mathbf{z}^{\pm 1}]:=\mathbb{C}[z_1^{\pm 1},\ldots,z_n^{\pm 1}]$ be a Laurent polynomial in $n$ variables $z_1,\ldots,z_n$. Let $f=\sum_{\alpha\in\mathbb{Z}^n}c_\alpha \mathbf{z}^\alpha$ ($c_\alpha\in\mathbb{C}$) be its expansion in monomials
$\mathbf{z}^\alpha:=z_1^{\alpha_1}\cdots  z_n^{\alpha_n}$. Then, we write
$f(Y):=\sum_{\alpha\in\mathbb{Z}^n}c_\alpha Y^\alpha\in\mathcal{A}_n$, where $Y^\alpha:=Y_1^{\alpha_1}\cdots Y_n^{\alpha_n}$. The map
$f\mapsto f(Y)$ defines an isomorphism $\mathbb{C}[\mathbf{z}^{\pm 1}]\overset{\sim}{\longrightarrow} \mathcal{A}_n$ of commutative algebras. In addition, the multiplication map
\[
\mathcal{H}_n^0\otimes \mathcal{A}_n\rightarrow \mathcal{H}_n,\qquad
h\otimes f(Y)\mapsto hf(Y),
\]
is a linear isomorphism.

In \cite[\S 8]{Al-Qasimi:2017aa} it was shown that there exists a unique unit preserving algebra map
$\nu_n: \mathcal{H}_n\rightarrow \mathcal{H}_{n+1}$ satisfying for $n\geq 2$,
\begin{equation}\label{nun}
\begin{split}
&\nu_n(T_i)=T_i,\qquad i=1,\ldots,n-1,\\
&\nu_n(T_0)=T_nT_{0}T_n^{-1},\\
&\nu_n(\rho)=t^{-\frac{1}{4}}\rho T_n^{-1},
\end{split}
\end{equation}
satisfying $\nu_1(\rho)=t^{-\frac{1}{4}}\rho T_1^{-1}$ for $n=1$, 
and satisfying $\nu_0(X)=t^{\frac{1}{4}}\rho+t^{-\frac{1}{4}}\rho^{-1}$ for $n=0$. The $\nu_n$ was obtained in  \cite[\S 8]{Al-Qasimi:2017aa} as the Hecke algebra descent of an algebra homomorphism
$\mathbb{C}[\mathcal{B}_n]\rightarrow\mathbb{C}[\mathcal{B}_{n+1}]$, with $\mathcal{B}_n$ the extended affine braid group on $n$ strands, defined topologically by inserting an extra braid going underneath all the other braids it meets. 
At the end of this section, we require the algebra maps $\nu_n$ in constructing towers of $\mathcal{H}_n$-modules and qKZ towers of solutions.

\subsection{qKZ equations}
We consider Cherednik's  \cite{Cherednik:1992aa,Cherednik:1994aa} qKZ equations of type $\textup{GL}_n$.
We will follow closely \cite{Stokman:2011aa}, and we will restrict attention to twisted-symmetric solutions of qKZ equations. The notations $(m,k,\xi)$ in \cite[\S 4.3]{Stokman:2011aa} correspond to our $(n,-t^{\frac{1}{2}},\rho)$. The qKZ equations depend on an additional parameter $q$, which we for the moment take to be an arbitrary nonzero complex number.

Recall that for $n\geq 1$ and $t^{\frac{1}{2}}=1$, the extended affine Hecke algebra $\mathcal{H}_n(1)$ is isomorphic to the group algebra $\mathbb{C}[W_n]$ of the the extended affine symmetric group $W_n\simeq S_n\ltimes\mathbb{Z}^n$. Writing $s_i$ ($i\in\mathbb{Z}/n\mathbb{Z}$) and $\rho$ for the (Coxeter type) generators of $W_n$, acting on $\mathbb{C}[\mathbf{z}^{\pm 1}]$ and 
$\mathbb{C}(\mathbf{z}):=\mathbb{C}(z_1,\ldots,z_n)$ by 
\begin{equation}\label{actionWn}
\begin{split}
(s_if)(\mathbf{z})&:=f(\ldots,z_{i+1},z_i,\ldots)\qquad\qquad\quad (1\leq i<n),\\
(s_0f)(\mathbf{z})&:=f(qz_n,z_2,\ldots,z_{n-1},q^{-1}z_1),\\
(\rho f)(\mathbf{z})&:=f(z_2,\ldots,z_n,q^{-1}z_1),
\end{split}
\end{equation}
cf. Definition \ref{defH}. 
Note that the $W_n$-action on $\mathbb{C}[\mathbf{z}^{\pm 1}]$ is by graded algebra automorphisms, with the grading defined by the total degree. In addition, $W_n$ preserves the polynomial algebra
$\mathbb{C}[\mathbf{z}]:=\mathbb{C}[z_1,\ldots,z_n]$.

Define for $n\geq 1$ and $i\in\mathbb{Z}/n\mathbb{Z}$,
$$\widetilde{R}_i(x) := \frac{xT^{-1}_i - T_i}{t^{\frac{1}{2}}-t^{-\frac{1}{2}} x}, $$
which we view as rational $\mathcal{H}_n(t^{\frac{1}{2}})$-valued function in $x$. The key point
in the construction of qKZ equations is the fact that for any $\mathcal{H}_n(t^{\frac{1}{2}})$-module $V_n$ with representation map $\sigma_n: \mathcal{H}_n(t^{\frac{1}{2}})\rightarrow\textup{End}(V_n)$ and for $q\in\mathbb{C}^*$, the formulas
\begin{equation}\label{actionnabla}
\begin{split}
\bigl(\nabla(s_i)f\bigr)(\mathbf{z})&:=\sigma_n(\widetilde{R}_i(z_{i+1}/z_i))(s_if)(\mathbf{z})\qquad
1\leq i<n,\\
\bigl(\nabla(s_0)f\bigr)(\mathbf{z})&:=\sigma(\widetilde{R}_0(z_1/qz_n))(s_0f)(\mathbf{z}),\\
\bigl(\nabla(\rho)f\bigr)(\mathbf{z})&:=\sigma(\rho)(\rho f)(\mathbf{z}),
\end{split}
\end{equation}
define a left $W_n$-action on the space $V_n(\mathbf{z}):=\mathbb{C}(\mathbf{z})\otimes V_n$
of $V_n$-valued rational functions in $z_1,\ldots,z_n$, where the $W_n$-action in the right-hand
side is the action on the variables as given by \eqref{actionWn}. For $n=0$, we simply take $\nabla=\sigma_0$ acting on $V_0$. The fact that \eqref{actionnabla} defines a $W_n$-action is a consequence of the following identities for 
the $R$-operators $\widetilde{R}_i(x)$,
\begin{equation}\label{ybe}
\begin{split}
	\widetilde{R}_i(x)\widetilde{R}_{i+1}(xy) \widetilde{R}_i(y) &= \widetilde{R}_{i+1}(y)\widetilde{R}_i(xy)\widetilde{R}_{i+1}(x), \\
	\widetilde{R}_i(x)\widetilde{R}_j(y) &= \widetilde{R}_j(y)\widetilde{R}_i(x)\qquad\qquad i-j\not\equiv \pm 1, \\
	\widetilde{R}_i(x)\widetilde{R}_i(x^{-1})&=1, \\
	\rho \widetilde{R}_i(x) &= \widetilde{R}_{i+1}(x) \rho
	\end{split}
\end{equation}
with the indices taken modulo $n$.
The first equation is the Yang-Baxter equation \cite[Vol. 5]{b:2006aa} in braid form.

Note that in \eqref{actionWn} and \eqref{actionnabla} the action of $s_0$ is determined by the
action of $s_i$ ($1\leq i<n$) and of $\rho$, and hence does not have to be specified. We will often 
omit the explicit formula for the action of $s_0$ in the remainder of the paper.  Following
\cite{Stokman:2011aa} we call  the subspace $V_n(\mathbf{z})^{\nabla(W_n)}$ of $\nabla(W_n)$-invariant elements in $V_n(\mathbf{z})$ the space of twisted-symmetric
solutions of the qKZ equations on $V_n$. 
We need a more refined class of qKZ solutions, defined as follows.

\begin{definition}\label{Soldef}
Let $q\in\mathbb{C}^*$ and $c\in\mathbb{C}$. Fix a $\mathcal{H}_n(t^{\frac{1}{2}})$-module $V_n$ with representation map $\sigma_n: \mathcal{H}_n(t^{\frac{1}{2}}) \rightarrow \End(V_n)$. For $n\geq 2$ write $\textup{Sol}_n(V_n;q,c)\subseteq V_n[\mathbf{z}]$ for the $V_n$-valued polynomials $f\in V_n[\mathbf{z}]$ in the variables $z_1,\ldots,z_n$ satisfying
\begin{equation}\label{qkz}
	\begin{split}
	\sigma_{n}(\widetilde{R}_i(z_{i+1}/z_i))f(\ldots,z_{i+1},z_i,\ldots)&=f(\mathbf{z})
	\qquad (1\leq i<n),\\
	\sigma_n(\rho) f(z_2 ,\ldots,z_{n},q^{-1} z_1)&=cf(\mathbf{z}) .
\end{split}
\end{equation}
For $n=1$ we write $\textup{Sol}_1(V_1;q,c)$ for the $V_1$-valued polynomials $f\in V_1[z]$
in the single variable $z$ satisfying the $q$-difference equation $\sigma_1(\rho)f(q^{-1}z)=c\,f(z)$.
Finally, for $n=0$ write $\textup{Sol}_0(V_0;q,c)\subseteq V_0$ for the eigenspace of $\sigma_0(X)\in\textup{End}(V_0)$
with eigenvalue $c$.
\end{definition}
If $n\geq 1$ and $\textup{Sol}_n(V_n;q,c)\not=\{0\}$, then necessarily $c\in\mathbb{C}^*$. In this
case 
\[
\textup{Sol}_n(V_n;q,c)=V_n^{(c)}(\mathbf{z})^{\nabla(W_n)}\cap
V_n^{(c)}[\mathbf{z}],
\]
with $V_n^{(c)}$ denoting the vector space $V_n$ endowed with the twisted action $\sigma_n^{c}: \mathcal{H}_n\rightarrow\textup{End}(V_n)$ defined by
$\sigma_n^{c}(T_i):=\sigma_n(T_i)$ for $i\in\mathbb{Z}/n\mathbb{Z}$ and
$\sigma_n^{c}(\rho):=c^{-1}\sigma_n(\rho)$.  We call $c$ a {\it twist parameter}.

For $n\geq 2$ let  $\pi_{n}^{t^{\frac{1}{2}},q}: \mathcal{H}_n(t^{\frac{1}{2}}) \rightarrow \End (\C[\mathbf{z}^{\pm 1}])$  be Cherednik's \cite{Cherednik:2005aa} basic representation, defined by
\begin{align*}
	\pi_{n}^{t^{\frac{1}{2}},q} (T_i) &:= -t^{\frac{1}{2}}+\left( \frac{ t^{\frac{1}{2}}z_i - t^{-\frac{1}{2}} z_{i+1}}{z_{i+1} - z_i } \right) (s_i -1)\qquad (1\leq i<n),\\ 
	\pi_n^{t^{\frac{1}{2}},q}(T_0) &:= -t^{\frac{1}{2}}+\left( \frac{ t^{\frac{1}{2}}qz_n - t^{-\frac{1}{2}}  z_{1}}{z_{1} - qz_n } \right) (s_0 -1),\\
	\pi_n^{t^{\frac{1}{2}},q}(\rho) &:=\rho
\end{align*}
(see \cite[Thm. 3.1]{Stokman:2011aa} with $(m,k_i,\xi)$ replaced by $(n,-t^{\frac{1}{2}},\rho)$ and specializing to type $A$ as in \cite[\S 4.3]{Stokman:2011aa}). For $n=1$ we define the basic representation $\pi_1^{t^{\frac{1}{2}},q}: \mathcal{H}_1(t^{\frac{1}{2}})\rightarrow\textup{End}
\bigl(\mathbb{C}[\mathbf{z}^{\pm 1}]\bigr)$ by $\pi_1^{t^{\frac{1}{2}},q}(\rho):=\rho$.
Note that $\mathbb{C}[\mathbf{z}]$ is a 
$\pi_n^{t^{\frac{1}{2}},q}(\mathcal{H}_n)$-submodule of $\mathbb{C}[\mathbf{z}^{\pm 1}]$.

By \cite[Prop. 3.10]{Stokman:2011aa} (see also
\cite[\S 4.1]{Pasquier:2006aa} and \cite{Kasatani:2007ab})
we have for $n\geq 1$ and $c\in\mathbb{C}^*$
the following alternative description of $\textup{Sol}_n(V_n;q,c)$,
\[
\textup{Sol}_n(V_n;q,c)=\left\{f\in V_n[\mathbf{z}] \,\,\, | \,\,\, \pi_n^{t^{-\frac{1}{2}},q}(h)f=
\sigma^{c}_n(J(h))f\qquad \forall\, h\in\mathcal{H}_n(t^{-\frac{1}{2}})\right\},
\]
where $J: \mathcal{H}_n(t^{-\frac{1}{2}})\rightarrow\mathcal{H}_n(t^{\frac{1}{2}})$
is the unique anti-algebra isomorphism satisfying $J(T_i):=T_i^{-1}$ ($i\in\mathbb{Z}/n\mathbb{Z}$)
and $J(\rho):=\rho^{-1}$. Here the basic representation $\pi_n^{t^{-\frac{1}{2}},q}$ acts on the first tensor component of $V_n[\mathbf{z}]=\mathbb{C}[\mathbf{z}]\otimes V_n$. More concretely,
	\begin{align}
		\text{Sol}_n(V_n;q,c)=
		\left \{f \in V_n[\textbf{z}] \,\,\,\, \middle | \,\,\,\,\,
		\begin{matrix}\pi^{t^{-\frac{1}{2}},q}_{n} (T_i) f = \sigma_{n} (T_i^{-1}) f\quad
		(1\leq i<n)  \\ \pi_{n}^{t^{-\frac{1}{2}},q}(\rho)f =c\sigma_{n}(\rho^{-1})f
		\end{matrix}
		 \right  \}
	\end{align}
where one needs to be well aware that the action on the variables through the basic representation
is with respect to the extended affine Hecke algebra $\mathcal{H}_n(t^{-\frac{1}{2}})$ and the action on  $V_n$ through $\sigma_n$ is with respect to the extended affine Hecke algebra
$\mathcal{H}_n(t^{\frac{1}{2}})$.

Before we can conclude this section with the introduction of the notion of a qKZ tower of solutions we need to establish some notation.
Let $A$ be a complex associative algebra and write $\mathcal{C}_A$ for the category of left $A$-modules. 
Write $\textup{Hom}_A(M,N)$ for the space of morphisms $M\rightarrow N$ in $\mathcal{C}_A$, which we will call intertwiners.
Suppose that $\eta: A\rightarrow B$ is a (unit preserving) morphism of $\mathbb{C}$-algebras, then we write
$\textup{Ind}^{\eta}: \mathcal{C}_A\rightarrow\mathcal{C}_B$ and $\textup{Res}^{\eta}: \mathcal{C}_B\rightarrow\mathcal{C}_A$ for the corresponding induction and restriction functor. 
Concretely,
if $M$ is a left $A$-module then
\[
\textup{Ind}^{\eta}(M):=B\otimes_AM
\]
with $B$ viewed as a right $A$-module by $b\cdot a:=b\eta(a)$ for $b\in B$ and $a\in A$.
If $N$ is a left $B$-module then $\textup{Res}^{\eta}(N)$ is the complex vector space $N$, viewed
as an $A$-module by $a\cdot n:=\eta(a)n$ for $a\in A$ and $n\in N$.

For a left $\mathcal{H}_{n+1}$-module $V_{n+1}$ we use the
shorthand notation $V^{\nu_n}_{n+1}$ for the left $\mathcal{H}_n$-module $\textup{Res}^{\nu_n} (V_{n+1})$. The following lemma introduces the concept of 
the module lift of a qKZ solution.

\begin{lemma} \label{liftedsol}
Let $n\geq 0$. Let $V_n$ be a left $\mathcal{H}_n(t^{\frac{1}{2}})$-module and $V_{n+1}$ a left $\mathcal{H}_{n+1}(t^{\frac{1}{2}})$-module, with representation maps $\sigma_n$ and $\sigma_{n+1}$ respectively.
Let $\mu_n \in \textup{Hom}_{\mathcal{H}_n}(V_n,V^{\nu_n}_{n+1})$ be an intertwiner.
Extend $\mu_n$  to a $\C[\textbf{z}]$-linear map 
$V_n[\mathbf{z}]\rightarrow V_{n+1}^{\nu_n}[\mathbf{z}]$, which we still denote by $\mu_n$. Then, its restriction to 
$\textup{Sol}_n(V_n;q,c_n)$ is a linear map
\[
\mu_n: \textup{Sol}_n(V_n;q,c_n)\rightarrow \textup{Sol}_n(V_{n+1}^{\nu_n};q,c_n).
\]
\end{lemma}
\begin{proof}
This is immediate from the intertwining property
\begin{equation} \label{qkzinter}
\mu_n\circ \sigma_{n}(h) = (\sigma_{n+1}\nu_n)(h) \circ  \mu_n  \quad \quad \forall h \in \mathcal{H}_n.
 \end{equation}
Indeed, if $f\in\textup{Sol}_n(V;q,c_n)$ then it follows for $n\geq 1$ from \eqref{qkzinter} that
\begin{align*}
		(\sigma_{n+1}\nu_n)(\widetilde{R}_i(z_{i+1}/z_i))\mu_n(f(\ldots,z_{i+1},z_i,\ldots) )
		&= \mu_n\big(\sigma_n(\widetilde{R}_i(z_{i+1}/z_i))f(\ldots,z_{i+1},z_i,\ldots)\big)\\
		&=\mu_n(f(\mathbf{z}))
\end{align*}
for $1\leq i<n$ and
\begin{align*}
(\sigma_{n+1}\nu_n)(\rho)f(z_2,\ldots,z_n,q^{-1}z_1)&=\mu_n\big(\sigma_{n}(\rho)f(z_2,\ldots,
z_n,q^{-1}z_1)\big)\\
&=c_n\mu_n(f(\mathbf{z})),
\end{align*}
hence $\mu_n(f)\in\textup{Sol}_n(V_{n+1}^{\nu_n};q,c_n)$.
For $n=0$ and $f\in\textup{Sol}_0(V_0;q,c_0)$, i.e. $f\in V_0$ satisfying $\sigma_0(X)f=c_0f$, we have
\[
(\sigma_1\nu_0)(X)\mu_0(f)=\mu_0(\sigma_0(X)f)=c_0\mu_0(f),
\]
hence $\mu_0(f)\in\textup{Sol}_0(V_1^{\nu_0};q,c_0)$.
\end{proof}
By the intertwiner $\mu_n$ a qKZ solution $f^{(n)}(\mathbf{z}) \in \textup{Sol}_n(V_n;q,c_n) $ gets lifted to a solution in $\textup{Sol}_n(V^{\nu_n}_{n+1} ;q,c_n)$, taking values in the $\mathcal{H}_{n+1}$-module  $V_{n+1}$. 
Along with this {\it{upward}} module lift there is also a {\it{downward}} descent of a solution, which reduces the number of variables. It is defined as follows.

Recall the algebra map 
$\nu_{n} :\mathcal{H}_{n} \rightarrow \mathcal{H}_{n+1}$ defined by \eqref{nun}.
\begin{lemma}\label{downlemma}
Let $n\geq 0$ and let $V_{n+1}$ be a left $\mathcal{H}_{n+1}(t^{\frac{1}{2}})$-module with associated representation map $\sigma_{n+1}$.
Then, for $n\geq 1$ and $f \in \textup{Sol}_{n+1}(V_{n+1};q,c_{n+1})$,
	$$f(z_1,\ldots, z_n,0) \in \textup{Sol}_n(V^{\nu_n}_{n+1} ;q,-t^{-\frac{3}{4}}c_{n+1}), $$
and for $n=0$ and $f\in\textup{Sol}_1(V_1;q,c_1)$,
$$f(0)\in\textup{Sol}_0(V_1^{\nu_0};q,t^{\frac{1}{4}}c_1+t^{-\frac{1}{4}}c_1^{-1}).$$
\end{lemma}%
\begin{proof}%
Let $n\geq 1$ and $f \in \textup{Sol}_{n+1}(V_{n+1};q,c_{n+1})$. Set
$g(z_1,\ldots,z_n):=f(z_1,\ldots,z_n,0)$.
For $1\leq i<n$ we have
\begin{align*}
		 (\sigma_{n+1}\nu_n)(\widetilde{R}_i(z_{i+1}/z_i)) g(\ldots,z_{i+1},z_i,\ldots) &=  
		 \sigma_{n+1}(\widetilde{R}_i(z_{i+1}/z_i))f(z_1,\ldots,z_{i+1},z_i,
		 \ldots,z_n,0)\\
		 &=f(z_1,\ldots,z_n,0)=g(z_1,\ldots,z_n).
\end{align*} 
Hence, to prove that $g\in\textup{Sol}_n(V_{n+1}^{\nu_n};q,-t^{-\frac{3}{4}}c_{n+1})$ it remains to show that 
\begin{equation}\label{todo1}
(\sigma_{n+1}\nu_n)(\rho)g(z_2,\ldots,z_n,q^{-1}z_1)=
-t^{-\frac{3}{4}}c_{n+1}g(\mathbf{z}).
\end{equation}
To prove \eqref{todo1}, first note that
\begin{align*}
\sigma_{n+1}(\rho\widetilde{R}_n(z_1/qz_{n+1}))f(z_2,\ldots,z_n,q^{-1}z_1,z_{n+1})&=
\sigma_{n+1}(\rho)f(z_2,\ldots,z_{n+1},q^{-1}z_1)\\
&=c_{n+1}f(z_1,\ldots,z_{n+1}).
\end{align*}
Setting $z_{n+1}=0$ and using that $\widetilde{R}_n(\infty):=\lim_{x\rightarrow\infty}\widetilde{R}_n(x)=
-t^{\frac{1}{2}}T_n^{-1}$, we get
$$
-t^{\frac{1}{2}}\sigma_{n+1}(\rho T_n^{-1})g(z_2,\ldots,z_n,q^{-1}z_1)=c_{n+1}g(z_1,\ldots,z_n).
$$
Then, \eqref{todo1} follows from the fact that $\nu_n(\rho)=t^{-\frac{1}{4}}\rho T_n^{-1}$.

For $n=0$ and $f\in\textup{Sol}_1(V_1;q,c_1)$ we have
$$
(\sigma_1\nu_0)(X)f(0)=\sigma_1(t^{\frac{1}{4}}\rho+t^{-\frac{1}{4}}\rho^{-1})f(0)=
(t^{\frac{1}{4}}c_1+t^{-\frac{1}{4}}c_1^{-1})f(0),
$$
hence $f(0)\in\textup{Sol}_0(V_1^{\nu_0};q,t^{\frac{1}{4}}c_1+t^{-\frac{1}{4}}c_1^{-1})$.
\end{proof}
By lifting solutions of qKZ equations by intertwiners $\mu_n$ and descending solutions of qKZ equations by setting variables equal to zero we can connect qKZ solutions of different rank.
This leads to 
 the definition of a qKZ tower of solutions. The starting point is the 
following definition of a tower of extended affine Hecke algebra modules (compare with  \cite{Al-Qasimi:2017aa}, where this notion was introduced for modules over extended affine Temperley-Lieb algebras, see also Section \ref{EATLsection}).
\begin{definition} \label{def-moduletower}
A tower 
\[
V_0\overset{\mu_0}{\longrightarrow}V_1\overset{\mu_1}{\longrightarrow}V_2
\overset{\mu_2}{\longrightarrow}V_3\overset{\mu_3}{\longrightarrow}\cdots
\]
of extended affine Hecke algebra modules is a sequence
$\{(V_n,\mu_n)\}_{n\in\mathbb{Z}_{\geq 0}}$ with $V_n$ a left $\mathcal{H}_n$-module
and $\mu_n\in\textup{Hom}_{\mathcal{H}_n}\bigl(V_n,V_{n+1}^{\nu_n}\bigr)$.
\end{definition}
To lift this notion of a tower to solutions of qKZ equations it is convenient to disregard quasi-periodic (with respect to the action of $\rho$) symmetric normalization factors $h$, i.e. polynomials $h\in\mathbb{C}[\mathbf{z}]^{S_n}$ satisfying $\rho h=\lambda h$
for some $\lambda\in\mathbb{C}^*$. We call such $h$ a {\it $\lambda$-recursion factor}, and 
$\lambda$ the {\it scale parameter}. We write $\mathcal{T}_{n,\lambda}\subset\mathbb{C}[\mathbf{z}]$
for the space of $\lambda$-recursion factors.
Note that $hf\in \textup{Sol}_n(V_n;q,\lambda c_n)$ if $f\in\textup{Sol}_n(V_n;q,c_n)$ and 
$h\in\mathcal{T}_{n,\lambda}$. By convention we define the space $\mathcal{T}_{0,\lambda}$ of $\lambda$-recursion factors for $n=0$ to be $\mathbb{C}$ if $\lambda=1$ and $\{0\}$ otherwise.

If $q$ is a root of unity, then we write $e\in\mathbb{Z}_{>0}$ for the smallest natural number such that $q^e=1$. We take $e=\infty$ if $q$ is not a root of unity.
 \begin{lemma}
 Let $n\geq 1$. Then, $\mathcal{T}_{n,\lambda}=\{0\}$ unless $\lambda=q^{-m}$ for some $0\leq m<e$.
 If $0\leq m<e$ then
 \[
 \mathcal{T}_{n,q^{-m}}=\mathbb{C}[z_1^e,\ldots,z_n^e]^{S_n}(z_1\cdots z_n)^m.
 \]
The latter formula should be read as $\mathcal{T}_{n,q^{-m}}=\textup{span}_{\mathbb{C}}\{(z_1\cdots z_n)^m\}$ if $e=\infty$.
 \end{lemma}
 \begin{proof}
 Let $\alpha\in\mathbb{Z}_{\geq 0}^n$. It suffices to show that 
 $\sum_{\beta\in S_n\alpha}\mathbf{z}^\beta\in\mathbb{C}[\mathbf{z}]^{S_n}$ is a $\lambda$-recursion factor if and only if there exists a $0\leq m<e$ such that $\lambda=q^{-m}$ and $\alpha_i\equiv m$ mod $e$ 
 for all $i$ (where the latter condition for $e=\infty$ is read as $\alpha_i=m$ for all $i$). 
 
 Note that
 \[
 \rho\Bigl(\sum_{\beta\in S_n\alpha}z^\beta\Bigr)=\sum_{\beta\in S_n\alpha}q^{-\beta_n}z_1^{\beta_n}
 z_2^{\beta_1}\cdots z_n^{\beta_{n-1}}=\sum_{\beta\in S_n\alpha}q^{-\beta_1}z^\beta,
 \]
 hence $\sum_{\beta\in S_n\alpha}z^\beta\in\mathcal{T}_{n,\lambda}$ if and only if $\lambda=q^{-\alpha_i}$ for all $i=1,\ldots,n$. This is equivalent to $\lambda=q^{-m}$ and $\alpha_i\equiv m$ mod $e$ for some $0\leq m<e$.
 \end{proof}
The following lemma shows that by rescaling a nonzero symmetric polynomial solution of the qKZ equations by an appropriate recursion factor, it will remain nonzero if one of its variables is set to zero.
\begin{lemma}
Let $n\geq 1$ and let $V_n$ be a left $\mathcal{H}_n$-module with representation map $\sigma_n$.
If $0\not=f\in \textup{Sol}_n(V_n;q,c_n)$ then there exists a unique $m\in\mathbb{Z}_{\geq 0}$
and $g\in\textup{Sol}_n(V_n;q,q^mc_n)$ such that $f(\mathbf{z})=(z_1\cdots z_n)^mg(\mathbf{z})$ and $g(z_1,\ldots,z_{n-1},0)\not\equiv 0$.
\end{lemma}
\begin{proof}
Recall that the existence of a nonzero $f\in\textup{Sol}_n(V_n;q,c_n)$ guarantees that $c_n\not=0$.
Suppose that $f(z_1,\ldots,z_{n-1},0)\equiv 0$. Using $\sigma_n(\rho)f(z_2,\ldots,z_n,q^{-1}z_1)=
c_nf(\mathbf{z})$ repeatedly we conclude that $f(\ldots,z_{i-1},0,z_{i+1},\ldots)\equiv 0$.
Hence, $f(\mathbf{z})$ is divisible by the $q^{-1}$-recursion factor $z_1\dots z_n$ in $V_n[\mathbf{z}]$. Now divide this factor
out and apply induction to the total degree of $f$.
\end{proof}

\begin{definition}[qKZ tower] 
Let $\{(V_n,\mu_n)\}_{n\in\mathbb{Z}_{\geq 0}}$ be a tower of extended affine Hecke algebra modules. We call
 $(f^{(n)})_{n\geq 0}$ an associated \emph{qKZ tower of solutions} with twisting parameters $c_n\in\mathbb{C}^*$ ($n\geq 1$) if there exist recursion factors 
 $h^{(n)}\in\mathcal{T}_{n,\lambda_n}$ 
 ($n\geq 0$) such that
\begin{enumerate}
\item[{\bf a)}] $0\not=f^{(n)}\in\textup{Sol}_n(V_n;q,c_n)$ for $n\geq 0$, with $c_0:=
t^{\frac{1}{4}}c_1+t^{-\frac{1}{4}}c_1^{-1}$.
\item[{\bf b)}] $f^{(n+1)}(z_1,\ldots,z_n,0)\not\equiv 0$ for all $n\geq 0$.
\item[{\bf c)}] For all $n\geq 0$ we have
 \begin{equation}\label{glue} 
  f^{(n+1)}(z_1,\ldots,z_n,0)=h^{(n)}(z_1,\ldots,z_n)\mu_n(f^{(n)}(z_1,\ldots,z_n)).
  \end{equation}
  We call \eqref{glue} the braid recursion relations for the qKZ tower $(f^{(n)})_{n\geq 0}$ of solutions.
\end{enumerate}
\end{definition}
Note that by Lemmas \ref{liftedsol} and \ref{downlemma}, we necessarily must have the
compatibility condition
\begin{equation}\label{compatibility}
-t^{-\frac{3}{4}}c_{n+1}=\lambda_nc_n\qquad n\geq 1
\end{equation}
between the twist and scale parameters in a qKZ tower of solutions (note that for $n=0$ we have $t^{\frac{1}{4}}c_1+t^{-\frac{1}{4}}c_1^{-1}=c_0$ by definition).

\section{Extended affine Temperley-Lieb algebra}\label{EATLsection}

The qKZ towers we construct are built using modules of the extended affine Temperley-Lieb algebra, which is a quotient of $\mathcal{H}_n$.
In this section we recall the definition of the extended affine Temperley-Lieb algebra and discuss
the relevant tower of extended affine Temperley-Lieb algebra modules, following \cite{Al-Qasimi:2017aa}.

The extended affine Temperley-Lieb algebras arise as the endomorphism algebras of the skein category of the annulus, see \cite{Al-Qasimi:2017aa} and references therein. 
We first give the definition of the extended affine Temperley-Lieb algebra in terms of generators and relations, and then discuss its relation to $\mathcal{H}_n$ and the qKZ equations.
For more details on the theory discussed in this section see \cite{Al-Qasimi:2017aa} and references within. 
\begin{definition} Let $n\geq 3$. The \emph{extended affine Temperley-Lieb algebra} $\mathcal{TL}_n=\mathcal{TL}_n(t^{\frac{1}{2}})$ is the complex associative algebra with generators
$e_i$ ($i\in\mathbb{Z}/n\mathbb{Z}$) and $\rho,\rho^{-1}$, and defining relations
\begin{equation}\label{relTL}
\begin{split}
	 &e_i^2 = \bigl(-t^{\frac{1}{2}} - t^{-\frac{1}{2}}\bigr) e_i,   \\
	 &e_ie_j = e_je_i\qquad\qquad\quad \hbox{ if }\,\, i-j\not\equiv\pm 1,\\
	 &e_ie_{i\pm 1}e_i = e_i,\\
	  	 	&\rho e_i =  e_{i+1} \rho, \\
		 & \rho\rho^{-1}=1=\rho^{-1}\rho,\\
		 & \bigl(\rho e_1\bigr)^{n-1}  = \rho^n (\rho e_1), 		 
\end{split}
\end{equation}
where the indices are taken modulo $n$. For $n=2$ the extended affine Temperley-Lieb algebra
$\mathcal{TL}_2=\mathcal{TL}_2(t^{\frac{1}{2}})$ is the algebra generated by $e_0,e_1,\rho^{\pm 1}$
with the defining relations \eqref{relTL} but with the third relation omitted. For $n=1$
we set $\mathcal{TL}_1=\mathcal{H}_1=\mathbb{C}[\rho,\rho^{-1}]$, and for $n=0$ we set
$\mathcal{TL}_0=\mathcal{H}_0=\mathbb{C}[X]$.
\end{definition}
The affine Temperley-Lieb algebra is the subalgebra $\mathcal{TL}^a_n$ of $\mathcal{TL}_n$ generated by $e_i$ ($i\in\mathbb{Z}/n\mathbb{Z}$). The first three relations in \eqref{relTL} are the defining relations in
terms of these generators (the first relation is the defining relation when $n=2$).
The (finite) Temperley-Lieb algebra is the subalgebra $\mathcal{TL}^0_n$ of $\mathcal{TL}^a_n$ generated by $e_1,\ldots,e_{n-1}$. The first three relations in \eqref{relTL} for the relevant indices
are then the defining relations. Note that the dependence on the parameter $t^{\frac{1}{2}}$
of $\mathcal{TL}_n$ is actually a dependence on $t^{\frac{1}{2}}+t^{-\frac{1}{2}}$.

It is well known that for $n\geq 2$ the assignments 
\[
T_i\mapsto e_i+t^{-\frac{1}{2}},\qquad \rho\mapsto \rho
\]
for $i\in\mathbb{Z}/n\mathbb{Z}$ extend to a surjective algebra homomorphism $\psi_n: \mathcal{H}_n(t^{\frac{1}{2}})\twoheadrightarrow
\mathcal{TL}_n(t^{\frac{1}{2}})$ see  e.g., \cite[Prop. 7.2]{Al-Qasimi:2017aa} and references therein.
For $n=1$ and $n=0$ we take $\psi_n:  \mathcal{H}_n\rightarrow\mathcal{TL}_n$ to be the identity
map.

Via the  map $\psi_n$ the R-operators $R_i(x):=\psi_n(\widetilde{R}_i(x))$ ($i\in\mathbb{Z}/n\mathbb{Z}$) on the extended affine Temperley-Lieb level are
	\begin{equation}
	R_i(x) = a(x)e_i+b(x)
	\end{equation}
	as rational $\mathcal{TL}_n$-valued function in $x$, with 
	$a(x)=a(x;t^{\frac{1}{2}})$ and $b(x)=b(x;t^{\frac{1}{2}})$ given by
	\begin{equation}\label{weights}
	a(x):=\frac{x -1}{t^{\frac{1}{2}}-t^{-\frac{1}{2}} x},\qquad
	b(x):=\frac{xt^{\frac{1}{2}} - t^{-\frac{1}{2}}}{t^{\frac{1}{2}}-t^{-\frac{1}{2}} x}.
		\end{equation}
		Note that the $R_i(x)$ ($i\in\mathbb{Z}/n\mathbb{Z}$) satisfy the Yang-Baxter type equations \eqref{ybe} in $\mathcal{TL}_n$. The weights $a(x)$ and $b(x)$ will play an important role in the next section, where they appear as the Boltzmann weights of the dense loop model. 

We can now define the following analog of the qKZ solution space $\textup{Sol}_n(V_n;q,c)$
(Definition \ref{Soldef}) for left $\mathcal{TL}_n$-modules $V_n$. For $n\geq 2$ it is the space of $V_n$-valued polynomials $f\in V_n[\mathbf{z}]$ in the variables $z_1,\ldots,z_n$ satisfying
\begin{equation}
\begin{split}
\sigma_n(R_i(z_{i+1}/z_i))f(\ldots,z_{i+1},z_i,\ldots)&=f(\mathbf{z})\qquad (1\leq i<n),\\
\sigma_n(\rho)f(z_2,\ldots,z_n,q^{-1}z_1)&=cf(\mathbf{z}),
\end{split}
\end{equation}
where $\sigma_n$ is the representation map of the $\mathcal{TL}_n$-module $V_n$. For $n=1$ it
is the space of $V_1$-valued polynomials $f$ in the single variable $z$ satisfying 
$\sigma_1(\rho)f(q^{-1}z)=cf(z)$.
For $n=0$ it is the eigenspace of $\sigma_0(X)$ with eigenvalue $c$. By a slight abuse of notation we will denote this space of solutions again by $\textup{Sol}_n(V_n;q,c)$.  No confusion can arise,
since $\textup{Sol}_n(V_n;q,c)$ for the left $\mathcal{TL}_n$-module $V_n$ coincides with $\textup{Sol}_n(\widetilde{V}_n;q,c)$, where $\widetilde{V}_n$ is the $\mathcal{H}_n$-module obtained by endowing
$V_n$ with the lifted $\mathcal{H}_n$-module structure with representation map 
$\sigma_n\circ\psi_n$.

{}From \cite[Prop. 6.3]{Al-Qasimi:2017aa} we have an algebra homomorphism 
$\mathcal{I}_n:\mathcal{TL}_n(t^{\frac{1}{2}})\rightarrow\mathcal{TL}_{n+1}(t^{\frac{1}{2}})$
for $n\geq 0$ defined by $\mathcal{I}_0(X)=t^{\frac{1}{4}}\rho+t^{-\frac{1}{4}}\rho^{-1}$ and
\begin{equation*}
\begin{split}
&\mathcal{I}_n(e_i)=e_i,\qquad\qquad 1\leq i<n,\\
&\mathcal{I}_n(\rho)=\rho(t^{-\frac{1}{4}}e_n+t^{\frac{1}{4}})
\end{split}
\end{equation*}
for $n\geq 1$. In particular, $\mathcal{I}_n(\rho^{-1})=(t^{\frac{1}{4}}e_n+t^{-\frac{1}{4}})\rho^{-1}$.
Note that we have a commutative diagram
\vspace{0.4cm}
\begin{equation}\label{commdiagram}
\setlength{\arraycolsep}{1cm}
\begin{array}{cc}
\Rnode{a}{\mathcal{H}_n} & \Rnode{b}{\mathcal{H}_{n+1}}\\[1.5cm]
 \Rnode{c}{\mathcal{TL}_n} &\Rnode{d}{\mathcal{TL}_{n+1}}\\[0.4cm]
\end{array}
\psset{nodesep=5pt,arrows=->}
\ncLine{a}{b}\Aput{\nu_n}
\ncLine{a}{c}\Bput{\psi_n}
\ncLine{b}{d}\Aput{\psi_{n+1}}
\ncLine{c}{d}\Aput{\mathcal{I}_n}
\end{equation}
Following \cite[Def. 7.1]{Al-Qasimi:2017aa}, we say that  $\{(V_n,\mu_n)\}_{n\in \mathbb{Z}_{\geq 0}}$
is a tower of extended affine Temperley-Lieb modules if $V_n$ is a left $\mathcal{TL}_n$-module
and $\mu_n\in\textup{Hom}_{\mathcal{TL}_n}\bigl(V_n,V_{n+1}^{\mathcal{I}_n}\bigr)$ for all
$n\geq 0$. We sometimes write the tower as 
\[
V_0\overset{\mu_0}{\longrightarrow}V_1\overset{\mu_1}{\longrightarrow}V_2
\overset{\mu_2}{\longrightarrow}V_3\overset{\mu_3}{\longrightarrow}\cdots
\]

Note that \eqref{commdiagram}
implies that an intertwiner $\mu_n\in\textup{Hom}_{\mathcal{TL}_n}(V_n,V_{n+1}^{\mathcal{I}_n})$
is also an intertwiner $\widetilde{V}_n\rightarrow\widetilde{V}_{n+1}^{\nu_n}$
of the associated $\mathcal{H}_n$-modules. Hence, the tower $\{(V_n,\mu_n)\}_{n\geq 0}$
of extended affine Temperley-Lieb algebra modules gives rise to the tower 
$\{(\widetilde{V}_n,\mu_n)\}_{n\geq 0}$ of extend affine Hecke algebra modules. Conversely,
if $\{(\widetilde{V}_n,\mu_n)\}_{n\geq 0}$ is a tower of extended affine Hecke algebra modules and
the representation maps $\widetilde{\sigma}_n: \mathcal{H}_n\rightarrow\textup{End}(V_n)$ factorize through $\psi_n$, then the tower descends to a tower of extended affine Temperley-Lieb algebra modules. 
We will freely use these lifts and descents of towers in the sequel of the paper.
		
The tower of extended affine Temperley-Lieb modules relevant for the dense loop model
is constructed from the skein category $\mathcal{S}=\mathcal{S}(t^{\frac{1}{4}})$ of the annulus,
defined in \cite{Al-Qasimi:2017aa}. We shortly recall here the basic features of the category $\mathcal{S}$. For further details, we refer to \cite[\S 3]{Al-Qasimi:2017aa}.

The category $\mathcal{S}$ is the complex linear category with objects $\mathbb{Z}_{\geq 0}$ and with the space of morphisms $\textup{Hom}_{\mathcal{S}}(m,n)$ being the linear span of planar isotopy classes of $(m,n)$-tangle diagrams on the annulus $A:=\{z\in\mathbb{C} \,\, | \,\, 1\leq |z|\leq 2\}$, with $m$ and $n$ marked ordered points on the inner and outer boundary respectively, modulo the Kauffman skein relation
 	\begin{align}
		&\psset{unit=0.8}\begin{pspicture}[shift=-0.9](-1,-1)(1,1)
			\pscircle[fillstyle=solid,fillcolor=diskin,linewidth=0.8pt,linestyle=dotted](0,0){1}
			\degrees[8]
			\psline[linecolor= line,linewidth=1.2pt](1;1)(1;5)
			\psline[linecolor= blue,linewidth=1.2pt](1;3)(0.2;3)
			\psline[linecolor= blue,linewidth=1.2pt](1;7)(0.2;7)	
		\end{pspicture}
 = t^{\frac{1}{4}}\; 
		\begin{pspicture}[shift=-0.9](-1,-1)(1,1)
			\pscircle[fillstyle=solid,fillcolor=diskin,linewidth=0.8pt,linestyle=dotted ](0,0){1}
			\degrees[8]
			\pscurve[linecolor= line,linewidth=1.2pt](1;3)(0.9;3)(0.4;2)(0.9;1)(1;1)
			\pscurve[linecolor= line,linewidth=1.2pt](1;5)(0.9;5)(0.4;6)(0.9;7)(1;7)
		\end{pspicture}
		+ t^{-\frac{1}{4}} \;
		\begin{pspicture}[shift=-0.9](-1,-1)(1,1)
			\pscircle[fillstyle=solid,fillcolor=diskin,linewidth=0.8pt,linestyle=dotted](0,0){1}
			\degrees[8]
			\pscurve[linecolor= line,linewidth=1.2pt](1;5)(0.9;5)(0.4;4)(0.9;3)(1;3)
			\pscurve[linecolor= line,linewidth=1.2pt](1;1)(0.9;1)(0.4;0)(0.9;7)(1;7)
		\end{pspicture} \; ;\label{kauffman}
	\end{align}
	and the (null-homotopic) loop removal relation 		
	\begin{align}
		&\psset{unit=0.8}\begin{pspicture}[shift=-0.9](-1,-1)(1,1)
			\pscircle[fillstyle=solid,fillcolor=diskin,linewidth=0.8pt,linestyle=dotted](0,0){1}
			\pscircle[linecolor= line,linewidth=1.2pt](0;0){0.5} 
		\end{pspicture}
	=	-(t^{\frac{1}{2}} +t^{-\frac{1}{2}}) \;
		\begin{pspicture}[shift=-0.9](-1,-1)(1,1)
			\pscircle[fillstyle=solid,fillcolor=diskin,linewidth=0.8pt,linestyle=dotted](0,0){1}
		\end{pspicture}\;.\label{loopremoval}
		\end{align} 
We consider here planar isotopies that fix the boundary of $A$ pointwise.		
The ordered marked points on the boundary are $\xi_m^{i-1}$ ($1\leq i\leq m$) and
$\xi_n^{j-1}$ ($1\leq j\leq n$) with $\xi_\ell:=e^{2\pi \mathrm{i}/\ell}$. In these equations the disc shows the local neighbourhood in the annulus where the diagrams differ.
Let  $L$ be an $(l,m)$-tangle diagram and $L'$ an $(m,n)$-tangle diagram.
The composition $[L'] \circ [L]$ of the corresponding equivalence classes in $\mathcal{S}$ is
$[L'\circ L]$, with $L'\circ L$ the $(l,n)$-tangle diagram obtained
 by placing $L$ inside $L'$ such that the outer boundary points of $L$ match with the inner boundary points of $L'$. For example,
 \begin{equation*}
\psset{unit=0.7}
\begin{pspicture}[shift=-1.8](-1.8,-1.8)(1.8,1.8)
      	\SpecialCoor
    	\pscircle[fillstyle=solid, fillcolor=diskin,linewidth=1.2pt](0,0){1.5}
    	\pscircle[fillstyle=solid, fillcolor=white, linewidth=1.2pt](0,0){0.5}
    	\degrees[8]
     	\rput{0}(1.7;0){\tiny $1$}
    	\rput{0}(1.7;4){\tiny $2$}
	\rput{0}(0.25;0){\tiny $1$}
    	\rput{0}(0.25;4){\tiny $2$}
	\pscurve[linecolor=line,linewidth=1.2pt](1.5;0)(1.4;0)(1.2;1)(1.25;2)(1.2;3)(1.4;4)(1.5;4)
	\pscurve[linecolor=line,linewidth=1.2pt](0.5;0)(0.6;0)(0.7;2)(0.6;4)(0.5;4)
   \end{pspicture}
   \circ
        \begin{pspicture}[shift=-1.8](-1.8,-1.8)(1.8,1.8)
    	\SpecialCoor
    	\pscircle[fillstyle=solid, fillcolor=diskin,linewidth=1.2pt](0,0){1.5}
    	\pscircle[fillstyle=solid, fillcolor=white, linewidth=1.2pt](0,0){0.5}
    	\degrees[8]
     	\rput{0}(1.7;0){\tiny $1$}
    	\rput{0}(1.7;4){\tiny $2$}
	\rput{0}(0.26;0){\tiny $1$}
    	\rput{0}(0.26;2){\tiny $2$}
    	\rput{0}(0.26;4){\tiny $3$}
    	\rput{0}(0.26;6){\tiny $4$}
	\pscurve[linecolor=line,linewidth=1.2pt](1.5;0)(1.4;0)(1.1;6)(1.4;4)(1.5;4)
	\pscurve[linecolor=line,linewidth=1.2pt](0.5;0)(0.6;0)(0.7;1)(0.6;2)(0.5;2)
	\pscurve[linecolor=line,linewidth=1.2pt](0.5;4)(0.6;4)(0.7;5)(0.6;6)(0.5;6)
   \end{pspicture}   
    =
    \begin{pspicture}[shift=-1.8](-1.8,-1.8)(1.8,1.8)
    	\SpecialCoor
    	\pscircle[fillstyle=solid, fillcolor=diskin,linewidth=1.2pt](0,0){1.5}
    	\pscircle[fillstyle=solid, fillcolor=white, linewidth=1.2pt](0,0){0.5}
    	\degrees[8]
     	\rput{0}(1.7;0){\tiny $1$}
    	\rput{0}(1.7;4){\tiny $2$}
	\rput{0}(0.26;0){\tiny $1$}
    	\rput{0}(0.26;2){\tiny $2$}
    	\rput{0}(0.26;4){\tiny $3$}
    	\rput{0}(0.26;6){\tiny $4$}
	\pscurve[linecolor=line,linewidth=1.2pt](1.5;0)(1.4;0)(1.2;1)(1.25;2)(1.2;3)(1.4;4)(1.5;4)
	\pscurve[linecolor=line,linewidth=1.2pt](0.5;0)(0.6;0)(0.7;1)(0.6;2)(0.5;2)
	\pscurve[linecolor=line,linewidth=1.2pt](0.5;4)(0.6;4)(0.7;5)(0.6;6)(0.5;6)
	\pscircle[linecolor=line,linewidth=1.2pt](0,0){1}
   \end{pspicture}  
\end{equation*}

By \cite[Prop. 2.3.7]{Green:1998aa} and \cite[Thm. 5.3]{Al-Qasimi:2017aa} we have an isomorphism 
$\theta_n: \mathcal{TL}_n(t^{\frac{1}{2}})\overset{\sim}{\longrightarrow}
\textup{End}_{\mathcal{S}(t^{\frac{1}{4}})}(n)$ of algebras for $n\geq 0$, with the algebra  isomorphism $\theta_n$ for $n\geq 1$ determined by
	\begin{equation*}
	\psset{unit=0.8}
			\rho \mapsto   
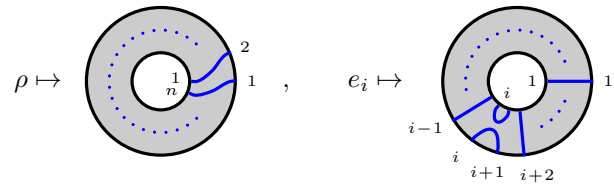
			 \begin{pspicture}[shift=-1.4](-1.5,-1.5)(2,1.5)
    	\pscircle[fillstyle=solid,fillcolor = diskin,linewidth=1.2pt](0,0){1.25}
    	\pscircle[fillstyle=solid, fillcolor=white, linewidth=1.2pt](0,0){0.5}
    	\degrees[16]
     	\rput{0}(1.5;0){\tiny $1$}
	\rput{0}(1.5;1){\tiny $2$}
     	\rput{0}(0.25;0.5){\tiny $1$}
	\rput{0}(0.25;14){\tiny $n$}
	\pscurve[linecolor=line,linewidth=1.2pt](1.25;0)(1.15;0)(0.85;15.5)(0.6;15)(0.5;15)
	\pscurve[linecolor=line,linewidth=1.2pt](1.25;1)(1.15;1)(0.85;0.5)(0.6;0)(0.5;0)
	\psarc[linestyle=dotted, linecolor=line,linewidth=1.2pt](0,0){0.85}{2}{14}
   \end{pspicture},
   \qquad
			e_i \mapsto    
			\begin{pspicture}[shift=-1.65](-1.75,-1.75)(2,1.5)
    	\SpecialCoor
    	\pscircle[fillstyle=solid,fillcolor = diskin,linewidth=1.2pt](0,0){1.25}
    	\pscircle[fillstyle=solid, fillcolor=white, linewidth=1.2pt](0,0){0.5}
    	\degrees[17]
     	\rput{0}(1.5;0){\tiny $1$}
	\rput{0}(1.7;9.8){\tiny $i\!-\!1$}
    	\rput{0}(1.6;10.9){\tiny $i$}
	\rput{0}(1.6;11.9){\tiny $i\!+\!1$}
	\rput{0}(1.6;13.3){\tiny $i\!+\!2$}
     	\rput{0}(0.25;0){\tiny $1$}
    	\rput{0}(0.25;10.8){\tiny $i$}
	\psline[linecolor=line,linewidth=1.2pt](0.5;0)(1.25;0)
	\psline[linecolor=line,linewidth=1.2pt](0.5;10)(1.25;10)
	\psline[linecolor=line,linewidth=1.2pt](0.5;13)(1.25;13)
	\pscurve[linecolor=line,linewidth=1.2pt](0.5;11)(0.6;11)(0.75;11.5)(0.6;12)(0.5;12)
	\pscurve[linecolor=line,linewidth=1.2pt](1.25;11)(1.15;11)(0.9;11.5)(1.15;12)(1.25;12)
	\psarc[linestyle=dotted, linecolor=line,linewidth=1.2pt](0,0){0.85}{1}{9}
	\psarc[linestyle=dotted, linecolor=line,linewidth=1.2pt](0,0){0.85}{14}{16}
   \end{pspicture}
\end{equation*}
and for $n=0$ by
	\begin{equation*}
	\psset{unit=0.8}
	X \mapsto    
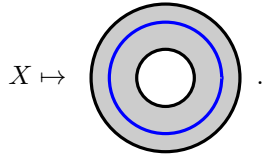
	\begin{pspicture}[shift=-1.4](-1.5,-1.5)(1.5,1.5)
    	\pscircle[fillstyle=solid,fillcolor = diskin,linewidth=1.2pt](0,0){1.25}
    	\pscircle[fillstyle=solid, fillcolor=white, linewidth=1.2pt](0,0){0.5}
	\pscircle[linecolor = line,linewidth=1.2pt](0,0){0.95}
   \end{pspicture}.
   \end{equation*}
 Moreover, in \cite[Def. 6.1]{Al-Qasimi:2017aa} an arc insertion functor $\mathcal{I}: \mathcal{S}\rightarrow\mathcal{S}$ is defined using a natural monoidal structure on $\mathcal{S}$. It maps $n$ to $n+1$ and, on morphisms, it inserts on the level of link diagrams a new arc connecting the inner and outer boundary while going underneath all arcs it meets
(the particular winding of the new arc is subtle, see \cite[\S 6]{Al-Qasimi:2017aa} for the details). 
 The resulting algebra homomorphisms $\mathcal{I}|_{\textup{End}_{\mathcal{S}}(n)}:
 \textup{End}_{\mathcal{S}}(n)\rightarrow\textup{End}_{\mathcal{S}}(n+1)$ coincides with the algebra homomorphism $\mathcal{I}_n$ by the identification of $\textup{End}_{\mathcal{S}}(n)$ with $\mathcal{TL}_n(t^{\frac{1}{2}})$ through the isomorphism $\theta_n$, see
 \cite[Prop. 8.3]{Al-Qasimi:2017aa}.

Let $v\in \mathbb{C}^*$ and set $u:=t^{\frac{1}{4}}v+t^{-\frac{1}{4}}v^{-1}$.
The one-parameter family of link pattern towers 
\[
    V_0(u)\overset{\phi_0}{\longrightarrow}V_1(v)\overset{\phi_1}{\longrightarrow}
    V_2(u)\overset{\phi_2}{\longrightarrow} V_3(v)\overset{\phi_3}{\longrightarrow}\cdots 
  \]
of extended affine Temperley-Lieb algebra modules is now defined as follows
(see \cite[\S 10]{Al-Qasimi:2017aa}). For $n=2k$ the $\mathcal{TL}_{2k}$-module $V_{2k}(u)$ is defined as 
\[
V_{2k}(u):=\textup{Hom}_{\mathcal{S}}(0,2k)\otimes_{\mathcal{TL}_0}\mathbb{C}_0^{(u)},
\]
where $\textup{Hom}_{\mathcal{S}}(0,2k)$ is endowed with its canonical 
$(\mathcal{TL}_{2k},\mathcal{TL}_0)$-bimodule structure and $\mathbb{C}_0^{(u)}$ denotes the one-dimensional representation of $\mathcal{TL}_0=\mathbb{C}[X]$ defined by $X\mapsto u$.
For $n=2k-1$ the $\mathcal{TL}_{2k-1}$-module $V_{2k-1}(v)$ is defined as
\[
V_{2k-1}(v):=\textup{Hom}_{\mathcal{S}}(1,2k-1)\otimes_{\mathcal{TL}_1}\mathbb{C}_1^{(v)}
\]
with $\mathbb{C}_1^{(v)}$ denoting the one-dimensional representations of $\mathcal{TL}_1=\mathbb{C}[\rho^{\pm 1}]$ defined by $\rho\mapsto v$. For $Y\in\textup{Hom}_{\mathcal{S}}(0,2k)$
we write
$Y_u:=Y\otimes_{\mathcal{TL}_0}1$ for the corresponding element in $V_{2k}(u)$. Similarly,
for $Z\in\textup{Hom}_{\mathcal{S}}(1,2k-1)$ we write $Z_v:=Z\otimes_{\mathcal{TL}_1}1$ for the corresponding
element in $V_{2k-1}(v)$. We sometimes omit the dependence of the representations $V_{2k}(u)$ and $V_{2k-1}(v)$ on $u=t^{\frac{1}{4}}v+t^{-\frac{1}{4}}v^{-1}$
and $v$, if it is clear from context.

The intertwiners $\phi_n$ ($n\geq 0$) are defined as follows. Consider the skein element
$$ U: = t^{\frac{1}{4}}
 \begin{pspicture}[shift=-1.1](-1.5,-1.25)(1.5,1.25)
	\SpecialCoor
	\pscircle[fillstyle=solid,fillcolor=diskin,linewidth=1.2pt](0,0){1}
	\pscircle[fillstyle=solid, fillcolor=white, linewidth=1.2pt](0,0){0.25}
	\degrees[8]
	\pscurve[linecolor=line,linewidth=1.2pt](1;0)(0.9;0)(0.7;1.5)(0.7;2.5)(0.9;4)(1;4)
	\rput{0}(1.2;0){\tiny $1$}
	\rput{0}(1.2;4){\tiny $2$}
\end{pspicture}
+ v
 \begin{pspicture}[shift=-1.1](-1.5,-1.25)(1.5,1.25)
	\SpecialCoor
	\pscircle[fillstyle=solid,fillcolor=diskin,linewidth=1.2pt](0,0){1}
	\pscircle[fillstyle=solid, fillcolor=white, linewidth=1.2pt](0,0){0.25}
	\degrees[8]
	\pscurve[linecolor=line,linewidth=1.2pt](1;0)(0.9;0)(0.7;6.5)(0.7;5.5)(0.9;4)(1;4)
	\rput{0}(1.2;0){\tiny $1$}
	\rput{0}(1.2;4){\tiny $2$}
\end{pspicture} \in \textup{Hom}_{\mathcal{S}}(0,2).$$
Then
\begin{equation*}
\begin{split}
\phi_{2k}([L]_u)&:= \mathcal{I}([L])_v,\\
\phi_{2k-1}([L^\prime]_{v})&:=(\mathcal{I}([L^\prime])\circ U)_u,
\end{split}
\end{equation*}
for a $(0,2k)$-link diagram $L$ and a $(1,2k-1)$-link diagram $L^\prime$.

\begin{example}\label{exampleonannulus}
	\begin{align*}
		\begin{pspicture}[shift=-1.4](-1.5,-1.5)(1.5,1.5)
    			\SpecialCoor
    			\pscircle[fillstyle=solid, fillcolor=diskin,linewidth=1.2pt](0,0){1}
			\pscircle[fillstyle=solid, fillcolor=white, linewidth=1.2pt](0,0){0.25}
    			\degrees[4]
     			\rput{0}(1.2;0){\tiny $1$}
    			\rput{0}(1.2;1){\tiny $2$}
    			\rput{0}(1.2;2){\tiny $3$}
			\rput{0}(1.2;3){\tiny $4$}
			\pscurve[linecolor=line,linewidth=1.2pt](1;1)(0.9;1)(-0.5;1.5)(0.9;2)(1;2)
			\pscurve[linecolor=line,linewidth=1.2pt](1;3)(0.9;3)(0.7;3.5)(0.9;0)(1;0)
  		 \end{pspicture}
		& \overset{\phi_4}{\longmapsto}
		\begin{pspicture}[shift=-1.4](-1.5,-1.5)(1.5,1.5)
    			\SpecialCoor
    			\pscircle[fillstyle=solid, fillcolor=diskin,linewidth=1.2pt](0,0){1}
			\pscircle[fillstyle=solid, fillcolor=white, linewidth=1.2pt](0,0){0.25}
    			\degrees[5]
     			\rput{0}(1.2;0){\tiny $1$}
    			\rput{0}(1.2;1){\tiny $2$}
    			\rput{0}(1.2;2){\tiny $3$}
			\rput{0}(1.2;3){\tiny $4$}
			\rput{0}(1.2;4){\tiny $5$}
			\pscurve[linecolor=line,linewidth=1.2pt](1;4)(0.9;4)(0.8;3.8)(0.5;3)(0.45;2)(0.4;1)(0.35;0)(0.25;0)
			\psdot[linecolor=diskin,dotsize=0.3](0.64; 3.5)
			\psdot[linecolor=diskin,dotsize=0.3](0.5; 2.9)
			\pscurve[linecolor=line,linewidth=1.2pt](1;1)(0.9;1)(0.7;0.5)(-0.4;1.5)(0.7;2.5)(0.9;2)(1;2)
			\pscurve[linecolor=line,linewidth=1.2pt](1;3)(0.9;3)(0.65;4)(0.9;0)(1;0)
  		 \end{pspicture} \\
		 \begin{pspicture}[shift=-1.4](-1.5,-1.5)(1.5,1.5)
    			\SpecialCoor
    			\pscircle[fillstyle=solid, fillcolor=diskin,linewidth=1.2pt](0,0){1}
			\pscircle[fillstyle=solid, fillcolor=white, linewidth=1.2pt](0,0){0.25}
    			\degrees[3]
     			\rput{0}(1.2;0){\tiny $1$}
    			\rput{0}(1.2;1){\tiny $2$}
    			\rput{0}(1.2;2){\tiny $3$}
			\pscurve[linecolor=line,linewidth=1.2pt](1;1)(0.9;1)(0.5;0.5)(0.35;0)(0.25;0)
			\pscurve[linecolor=line,linewidth=1.2pt](1;2)(0.9;2)(0.5;2.5)(0.9;3)(1;3)
  		 \end{pspicture}
	&\overset{\phi_3}{\longmapsto}
		 t^{\frac{1}{4}}
 			 \begin{pspicture}[shift=-1.4](-1.5,-1.5)(1.5,1.5)
    				\SpecialCoor
    				\pscircle[fillstyle=solid, fillcolor=diskin,linewidth=1.2pt](0,0){1}
				\pscircle[fillstyle=solid, fillcolor=white, linewidth=1.2pt](0,0){0.25}
    				\degrees[4]
   			  	\rput{0}(1.2;0){\tiny $1$}
    				\rput{0}(1.2;1){\tiny $2$}
    				\rput{0}(1.2;2){\tiny $3$}
				\rput{0}(1.2;3){\tiny $4$}
				\pscurve[linecolor=line,linewidth=1.2pt](1;1)(0.9;1)(0.35;2)(0.9;3)(1;3)
				\psdot[linecolor=diskin,dotsize=0.3](0.5; 2.8)
				\pscurve[linecolor=line,linewidth=1.2pt](1;0)(0.9;0)(0.5;2.9)(0.9;2)(1;2)
  			 \end{pspicture}
  		 +v
   			  \begin{pspicture}[shift=-1.4](-1.5,-1.5)(1.5,1.5)
                                	\SpecialCoor
                                	\pscircle[fillstyle=solid, fillcolor=diskin,linewidth=1.2pt](0,0){1}
				\pscircle[fillstyle=solid, fillcolor=white, linewidth=1.2pt](0,0){0.25}
                                	\degrees[4]
                                 	\rput{0}(1.2;0){\tiny $1$}
                                	\rput{0}(1.2;1){\tiny $2$}
                                	\rput{0}(1.2;2){\tiny $3$}
                            	\rput{0}(1.2;3){\tiny $4$}
                            	\pscurve[linecolor=line,linewidth=1.2pt](1;1)(0.9;1)(0.35;0)(0.9;3)(1;3)
                            	\psdot[linecolor=diskin,dotsize=0.3](0.5; 3.3)
                            	\pscurve[linecolor=line,linewidth=1.2pt](1;0)(0.9;0)(0.5;3)(0.9;2)(1;2)
   			\end{pspicture} 
		\end{align*}

\end{example}
The rather peculiar form of the intertwiners $\phi_{2k-1}$ can be explained in terms of a 
Roger-Yang \cite{Roger:2014aa} type graded algebra structure on the total space
$V_0(u)\oplus V_1(v)\oplus V_2(u)\oplus\ldots$ of the link pattern tower, see 
\cite[Rem. 8.11]{Al-Qasimi:2017aa}.

Let $\mathbb{D}=\{z\in\mathbb{C} \,\, | \,\, |z|\leq 2\}$ and $\mathbb{D}^*:=\mathbb{D}{\setminus}\{0\}$.
A punctured link pattern of size $2k$ is a perfect matching of the $2k$ equally spaced marked points  $2\xi_{2k}^{i-1}$ ($1\leq i\leq 2k$) on the boundary of $\mathbb{D}^*$ by $k$ non-intersecting arcs lying within $\mathbb{D}^*$. A punctured link pattern of size $2k-1$ is a perfect matching of the $2k$ marked points
$2\xi_{2k-1}^{j-1}$ ($1\leq j<2k$) and $0$ by $k$ non-intersecting arcs lying within $\mathbb{D}$. Only the endpoints of the arcs are allowed to lie on $\{0\}\cup\partial\mathbb{D}$.
Two link patterns are regarded the same if they are planar isotopic by a planar isotopy fixing $0$ and the boundary $\partial\mathbb{D}$ of $\mathbb{D}$ pointwise. 
 The arc connecting $0$ to the outer boundary of $\mathbb{D}$ is called the 
\emph{defect line}.
An arc that connects two points on the boundary are sometimes referred to as an \emph{arch} and an arch that connects two consecutive points that does not contain the puncture is called a \emph{little arch}. We denote the set of punctured link patterns of size $n$ by $\mathcal{L}_{n}$.
As an example, the following punctured link patterns \\
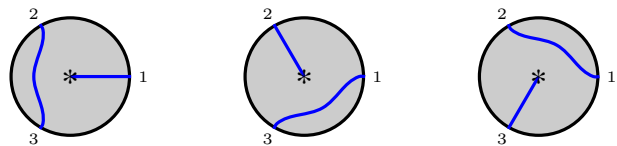
$$\psset{unit=0.8}
 \begin{pspicture}(-1.3,-1.3)(1.3,1.3)
    	\SpecialCoor
    	\pscircle[fillstyle=solid, fillcolor=diskin,linewidth=1.2pt](0,0){1}
    	\degrees[3]
     	\rput{0}(1.2;0){\tiny $1$}
    	\rput{0}(1.2;1){\tiny $2$}
    	\rput{0}(1.2;2){\tiny $3$}
	\psline[linecolor=line,linewidth=1.2pt](1;0)(0;0)
	\pscurve[linecolor=line,linewidth=1.2pt](1;1)(0.9;1)(0.6;1.5)(0.9;2)(1;2)
	\psdot[dotstyle=asterisk,dotscale=1.5](0;0)
   \end{pspicture} 
 \hspace{1cm}   
        \begin{pspicture}(-1.3,-1.3)(1.3,1.3)
    	\SpecialCoor
    	\pscircle[fillstyle=solid, fillcolor=diskin,linewidth=1.2pt](0,0){1}
    	\degrees[3]
     	\rput{0}(1.2;0){\tiny $1$}
    	\rput{0}(1.2;1){\tiny $2$}
    	\rput{0}(1.2;2){\tiny $3$}
	\psline[linecolor=line,linewidth=1.2pt](1;1)(0;1)
	\pscurve[linecolor=line,linewidth=1.2pt](1;2)(0.9;2)(0.6;2.5)(0.9;3)(1;3)
	\psdot[dotstyle=asterisk,dotscale=1.5](0;0)
   \end{pspicture}   
  \hspace{1cm}
        \begin{pspicture}(-1.3,-1.3)(1.3,1.3)
    	\SpecialCoor
    	\pscircle[fillstyle=solid, fillcolor=diskin,linewidth=1.2pt](0,0){1}
    	\degrees[3]
     	\rput{0}(1.2;0){\tiny $1$}
    	\rput{0}(1.2;1){\tiny $2$}
    	\rput{0}(1.2;2){\tiny $3$}
	\psline[linecolor=line,linewidth=1.2pt](1;2)(0;2)
	\pscurve[linecolor=line,linewidth=1.2pt](1;0)(0.9;0)(0.6;0.5)(0.9;1)(1;1)
	\psdot[dotstyle=asterisk,dotscale=1.5](0,0)
   \end{pspicture}
$$
constitute $\mathcal{L}_3$.

For twist parameter $v=1$ we can naturally identify
the $n$th representation space $V_n$ in the link pattern tower with $\mathbb{C}[\mathcal{L}_n]$
as a vector space by shrinking the hole $\{z\in\mathbb{C} \,\, | \,\, |z|\leq 1\}$ of the annulus to $0$.
The resulting action of $\mathcal{TL}_n$ on $\mathbb{C}[\mathcal{L}_n]$ can be explicitly described skein theoretically, see \cite[\S 8]{Al-Qasimi:2017aa}.

\section{qKZ equations on the space of link patterns} \label{subsec-linkqkz}

In this section we fix $v=1$. We discuss the qKZ equations associated with 
the $\mathcal{TL}_n$-modules $V_n\simeq\mathbb{C}[\mathcal{L}_n]$ ($n\geq 0$) from the link pattern tower, and we derive
necessary conditions for the existence of qKZ towers of solutions. The {\it existence} of qKZ towers of solutions will be the subject of later sections.

Let $L_{\cap}=L_\cap^{(n)} \in \mathcal{L}_n$ denote the link patterns
$$
  \psset{unit=0.8cm}
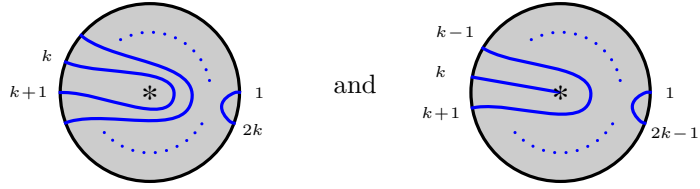
  \begin{pspicture}[shift=-1.8](-2.25,-1.8)(2.25,1.8)
    	\pscircle[fillstyle=solid, fillcolor=diskin,linewidth=1.2pt](0,0){1.5}
    	\degrees[360]
	\psdot[dotstyle=asterisk,dotscale=1.5](0,0)
	\rput{0}(1.8;0){\tiny$1$}
	\rput{0}(1.8;-20){\tiny$2k$}
	\rput{0}(1.8;160){\tiny$k$}
	\rput{0}(2;180){\tiny$k\!+\!1$}
	\pscurve[linecolor=line,linewidth=1.2pt](1.5;0)(1.4;0)(1.2;-10)(1.4;-20)(1.5;-20)
	\pscurve[linecolor=line,linewidth=1.2pt](1.5;160)(1.4;160)(-0.4;170)(1.4;180)(1.5;180)
	\pscurve[linecolor=line,linewidth=1.2pt](1.5;140)(1.4;140)(-0.7;170)(1.4;200)(1.5;200)
	\psarc[linestyle=dotted,linecolor=line,linewidth=1.2pt](0,0){1}{10}{120}
	\psarc[linestyle=dotted,linecolor=line,linewidth=1.2pt](0,0){1}{220}{-30}
   \end{pspicture}   
   \hspace{0.5cm} \text{ and }\hspace{0.5cm} 
  \begin{pspicture}[shift=-1.8](-2.25,-1.8)(2.25,1.8)
    	\pscircle[fillstyle=solid, fillcolor=diskin,linewidth=1.2pt](0,0){1.5}
    	\degrees[360]
	\rput{0}(1.8;0){\tiny$1$}
	\rput{0}(2;-20){\tiny$2k\!-\!1$}
	\rput{0}(2;150){\tiny$k\!-\!1$}
	\rput{0}(2;170){\tiny$k$}
	\rput{0}(2;190){\tiny$k\!+\!1$}
	\pscurve[linecolor=line,linewidth=1.2pt](1.5;0)(1.4;0)(1.2;-10)(1.4;-20)(1.5;-20)
	\pscurve[linecolor=line,linewidth=1.2pt](1.5;150)(1.4;150)(-0.5;170)(1.4;190)(1.5;190)
	\psline[linecolor=line,linewidth=1.2pt](1.5;170)(0;0)
	\psarc[linestyle=dotted,linecolor=line,linewidth=1.2pt](0,0){1}{10}{120}
	\psarc[linestyle=dotted,linecolor=line,linewidth=1.2pt](0,0){1}{220}{-30}
	\psdot[dotstyle=asterisk,dotscale=1.5](0,0)
   \end{pspicture}   
   $$
for $n=2k$ and $2k-1$, respectively. We call $L_\cap\in\mathcal{L}_n$ the {\it fully nested diagram}. For 
$g^{(n)}(\mathbf{z})=\sum_{L\in\mathcal{L}_n}g_L^{(n)}(\mathbf{z})L\in
V_n(\mathbf{z})$
we call $g_{L_\cap}^{(n)}(\mathbf{z})$ the {\it fully nested component} of $g^{(n)}(\mathbf{z})$.

The fully nested component plays an important role in the analysis of polynomial twisted-symmetric solutions $g^{(n)}(\mathbf{z})\in V_n^{(c)}(\mathbf{z})^{\nabla(W_n)}$ of the qKZ equations. In
\cite[\S 2.2]{Di-Francesco:2006aa} and \cite[\S 3.5.2]{Kasatani:2007aa} it was remarked that such solutions are uniquely determined by their fully nested component, and an explicit expression for the fully nested component was determined in case the solution is polynomial of total degree $\frac{1}{2}n(n-1)$ (existence of such a solution is a subtle issue). We recall these results here, extend them to qKZ solutions taking values in $V_{n+1}^{\mathcal{I}_n}$, and show how these results combined lead to explicit braid recursion relations. 

\begin{lemma}
Let $n\geq 1$, $q,c\in\mathbb{C}^*$ and let
\[
g^{(n)}(\mathbf{z})=\sum_{L\in\mathcal{L}_n}g_L^{(n)}(\mathbf{z})L\in
V_n^{(c)}(\mathbf{z})
\]
with coefficients $g_L^{(n)}(\mathbf{z})\in\mathbb{C}(\mathbf{z})$ ($L\in\mathcal{L}_n$). Then
$g^{(n)}(\mathbf{z})\in 
V_n^{(c)}(\mathbf{z})^{\nabla(W_n)}$
 if and only if for all $L\in\mathcal{L}_n$ and $1\leq i<n$,
\begin{equation}\label{recursion}
\begin{split}
g_L^{(n)}(\mathbf{z})&=b(z_{i+1}/z_i)g_L^{(n)}(s_i\mathbf{z})+\sum_{L^\prime\in\mathcal{L}_n:\, e_iL^\prime\sim L}\gamma_{L^\prime,L}^{(i)}a(z_{i+1}/z_i)g_{L^\prime}^{(n)}(s_i\mathbf{z}),\\
g_{L}^{(n)}(\mathbf{z})&=c^{-1}g_{\rho^{-1}L}^{(n)}(z_2,\ldots,z_n,q^{-1}z_1),
\end{split}
\end{equation}
where $e_iL^\prime\sim L$ means that $L$ is obtained from $e_iL^\prime$ by removing the loops in $e_iL^\prime$ (there is in fact at most one loop). The coefficient $\gamma_{L^\prime,L}^{(i)}$ is 
$$  \gamma^{(i)}_{L^\prime,L} = \begin{cases} 
						-(t^{\frac{1}{2}} +t^{-\frac{1}{2}})  &\text{ if } e_iL^\prime \text{ has a null-homotopic loop,}  \\
						t^{\frac{1}{4}} +t^{-\frac{1}{4}}  &\text{ if } e_iL^\prime \text{ has a non null-homotopic loop,} \\
						1& \text{ otherwise}.
						\end{cases}
$$
\end{lemma}
\begin{proof}
This follows directly by rewriting the qKZ equations
\begin{equation*}
\begin{split}
g^{(n)}(\mathbf{z})&=R_i(z_{i+1}/z_i)g^{(n)}(\ldots,z_{i+1},z_i,\ldots),\qquad 1\leq i<n,\\
g^{(n)}(\mathbf{z})&=c^{-1}\rho g^{(n)}(z_2,\ldots,z_n,q^{-1}z_1)
\end{split}
\end{equation*}
component-wise.
\end{proof}
 For the following lemmas concerning the uniqueness of solutions we need to impose
 that the loop weights $-(t^{\frac{1}{2}}+t^{-\frac{1}{2}})$ and 
 $t^{\frac{1}{4}}+t^{-\frac{1}{4}}$ are both nonzero.
\begin{lemma}\label{unique1} Let $n\geq 1$, $q,c\in\mathbb{C}^*$ and $t^{\frac{1}{4}}\in\mathbb{C}^*$
with $(t^{\frac{1}{2}}+1)(t+1)\not=0$. Let
\[
g^{(n)}(\mathbf{z})=\sum_{L\in\mathcal{L}_n}g_L^{(n)}(\mathbf{z})L\in 
V_n^{(c)}(\mathbf{z})^{\nabla(W_n)}.
\]
{\bf (a)} If $g_{L_\cap}^{(n)}(\mathbf{z})=0$ then $g^{(n)}(\mathbf{z})=0$.\\
{\bf (b)} If $g_{L_\cap}^{(n)}(\mathbf{z})\in\mathbb{C}[\mathbf{z}]$ is a homogeneous polynomial
of total degree $m$, then so is $g_L^{(n)}(\mathbf{z})$ for all $L\in\mathcal{L}_n$.
\end{lemma}
\begin{proof}
In Appendix A we show by induction that, given $g_{L_\cap}(\mathbf{z})$,  the recursion relations \eqref{recursion} determine the other coefficients $g_L^{(n)}(\mathbf{z})$ ($L\in\mathcal{L}_n$)
uniquely. For this the first equation in \eqref{recursion} is used in the following way:
for $L^\prime\in\mathcal{L}_n$ and $1\leq i<n$ such that $L^\prime$ does not have a little arch between $i$ and $i+1$, denote by $L\in\mathcal{L}_n$ the link
pattern such that $e_iL^\prime\sim L$, then $g_{L^\prime}^{(n)}(\mathbf{z})$ can be computed from other base components
by the formula
\[
\gamma_{L^\prime,L}^{(i)}a(z_{i+1}/z_i)g_{L^\prime}^{(n)}(s_i\mathbf{z})=g_L(\mathbf{z})-b(z_{i+1}/z_i)g_L(s_i\mathbf{z})
-\sum_{L^{\prime\prime}\in\mathcal{L}_n{\setminus}\{L^\prime\}:\, e_iL^{\prime\prime}\sim L}
\gamma_{L^{\prime\prime},L}^{(i)}a(z_{i+1}/z_i)g_{L^{\prime\prime}}^{(n)}(s_i\mathbf{z})
\]
since $\gamma_{L^\prime,L}^{(i)}\not=0$.
By substituting the explicit expressions of the weights $a(x)$ and $b(x)$, this can be rewritten as
\begin{equation*}
\begin{split}
\gamma_{L^\prime,L}^{(i)}(z_{i+1}-z_i)g_{L^\prime}^{(n)}(s_i\mathbf{z})&=
(1-s_i)\left(t^{\frac{1}{2}}z_i-t^{-\frac{1}{2}}z_{i+1}\right)g_L(\mathbf{z})\\
&\;\;-(z_{i+1}-z_i)\sum_{L^{\prime\prime}\in\mathcal{L}_n{\setminus}\{L^\prime\}:\, 
e_iL^{\prime\prime}\sim L}\gamma_{L^{\prime\prime},L}^{(i)}g_{L^{\prime\prime}}^{(n)}(s_i\mathbf{z}),
\end{split}
\end{equation*}
from which it is clear that $g_{L^\prime}^{(n)}(\mathbf{z})$ will be a homogeneous polynomial of total
degree $m$ if $g_L^{(n)}(\mathbf{z})$ and 
$g_{L^{\prime\prime}}^{(n)}(\mathbf{z})$ are homogeneous polynomials of total degree $m$.
\end{proof}
A similar result holds true for the restricted modules $V_{n+1}^{\mathcal{I}_n}$:
\begin{lemma}\label{unique2}
Let $n\geq 1$, $q,c\in\mathbb{C}^*$ and $t^{\frac{1}{4}}\in\mathbb{C}^*$ such that $(t^{\frac{1}{2}}+1)(t+1)\not=0$.
Let
\[
g^{(n)}(\mathbf{z})=\sum_{L\in\mathcal{L}_{n+1}}g_L^{(n)}(\mathbf{z})L\in 
V_{n+1}^{\mathcal{I}_n, (c)}(\mathbf{z})^{\nabla(W_n)}.
\]
{\bf (a)} If $g_{L_\cap}^{(n)}(\mathbf{z})=0$ with $L_\cap=L_\cap^{(n+1)}\in\mathcal{L}_{n+1}$
the fully nested diagram, then $g^{(n)}(\mathbf{z})=0$.\\
{\bf (b)} If $g_{L_\cap}^{(n)}(\mathbf{z})\in\mathbb{C}[\mathbf{z}]$ is a homogeneous polynomial
of total degree $m$, then so is $g_L^{(n)}(\mathbf{z})$ for all $L\in\mathcal{L}_{n+1}$.
\end{lemma}
\begin{proof}
The proof is similar to the proof of the previous lemma, but the check that the recursion relations
coming from the qKZ equations for the representation $V_{n+1}^{\mathcal{I}_n,(c)}$ determine all components in terms of the fully nested component $g_{L_\cap^{(n+1)}}^{(n)}(\mathbf{z})$
is more subtle. The details are given in Appendix A.
\end{proof}

\begin{corollary}\label{degreeleading}
Let $n\geq 1$ and 
\[
g^{(n)}(\mathbf{z})=\sum_{L\in\mathcal{L}_n}g_L^{(n)}(\mathbf{z})L\in V_n^{(c)}(\mathbf{z})^{\nabla(W_n)}.
\]
Then, 
\[
g_{L_\cap}^{(n)}(\mathbf{z})=C_n(\mathbf{z})\prod_{1\leq i<j\leq n}\left(t^{\frac{1}{2}}z_j-
t^{-\frac{1}{2}}z_i\right)
\]
with $C_n(\mathbf{z})\in\mathbb{C}(\mathbf{z})^{S_n}$. If in addition $g_{L_\cap}^{(n)}(\mathbf{z})$ is a homogeneous
polynomial of total degree $m$ and $(t^{\frac{1}{2}}+1)(t+1)\not=0$, then $m\geq \frac{1}{2}n(n-1)$ and 
$C_n(\mathbf{z})$ is a homogeneous symmetric polynomial of total degree $m-\frac{1}{2}n(n-1)$.
\end{corollary}
\begin{proof}
Note that $L_\cap$ does not have a little arch connecting $i$ and $i+1$ for $1\leq i<n$.
By the recursion relation \eqref{recursion}, it follows that
\begin{equation}\label{hulp}
g_{L_\cap}^{(n)}(s_i\mathbf{z})\left(t^{\frac{1}{2}}z_{i+1}-t^{-\frac{1}{2}}z_i\right)=
g_{L_\cap}^{(n)}(\mathbf{z})\left(t^{\frac{1}{2}}z_i-t^{-\frac{1}{2}}z_{i+1}\right)
\end{equation}
for $1\leq i<n$. The first result now follows immediately. 

For the second statement, suppose that $g_{L_\cap}^{(n)}(\mathbf{z})$ is a homogeneous polynomial
of total degree $m$. Then, 
\eqref{hulp} and $t^2\not=1$ imply that $g_{L_\cap}^{(n)}(\mathbf{z})$
is divisible by $t^{\frac{1}{2}}z_{2}-t^{-\frac{1}{2}}z_1$ in $\mathbb{C}[\mathbf{z}]$ and the resulting quotient is invariant under interchanging $z_1$ and $z_{2}$. One now proves by induction on $r$ that
$g_{L_\cap}^{(n)}(\mathbf{z})$ is divisible by $\prod_{1\leq i<j\leq r}\left(t^{\frac{1}{2}}z_j-t^{-\frac{1}{2}}z_i\right)$ in $\mathbb{C}[\mathbf{z}]$ and the resulting quotient is symmetric in $z_1,\ldots,z_r$. The second statement then follows by taking $r=n$. 
\end{proof}

It follows from the previous result that if the loop weights are nonzero and if there exists a nonzero
$g^{(n)}\in\textup{Sol}_n(V_n;q,c_n)$ with coefficients being homogeneous of total degree $\frac{1}{2}n(n-1)$,
then it is unique up to a nonzero scalar multiple and 
\[
g^{(n)}_{L_\cap}(\mathbf{z})=\kappa\prod_{1\leq i<j\leq n}\left(t^{\frac{1}{2}}z_j-t^{-\frac{1}{2}}z_i\right)
\]
for some $\kappa\in\mathbb{C}^*$. 

The following lemma is important in the analysis of qKZ
towers of solutions relative to the link pattern tower $\{(V_n,\phi_n)\}_{n\geq 0}$.
\begin{lemma}\label{intertwinerfullynested}
For $L\in\mathcal{L}_n$, consider the expansion
\[
\phi_n(L)=\sum_{L^\prime\in\mathcal{L}_{n+1}}c_{L,L^\prime}L^\prime\qquad (c_{L,L^\prime}\in\mathbb{C})
\]
of $\phi_n(L)$ in terms of the linear basis $\mathcal{L}_n$ of $V_n$. Then,
$c_{L,L_\cap^{(n+1)}}=t^{-\frac{1}{4}\lfloor n/2\rfloor}\delta_{L,L_\cap^{(n)}}$.
\end{lemma}
\begin{proof}
For $n=2k$, consider a  link pattern $L\in \mathcal{L}_{2k}$ that has a little arch connecting $i,i+1$ for some $i\in\{1,\ldots,2k-1\}$.
All the link patterns in the image $\phi_{2k}(L)$ also contain the same little arch since the inserted defect line at the skein module level does not cross it (possibly after an appropriate number of applications of Reidemeister II moves).
The only link pattern that does not contain a little arch connecting $i,i+1$ for any $1\leq i <2k$ is $L_\cap$.
By the mapping $\phi_{2k}$ we have at the skein module level
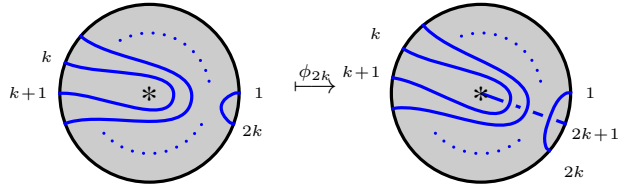
\begin{align*}
	\psset{unit=0.8}
	 \begin{pspicture}[shift=-1.8](-2.25,-1.8)(2.25,1.8)
    	\pscircle[fillstyle=solid, fillcolor=diskin,linewidth=1.2pt](0,0){1.5}
    	\degrees[360]
	\psdot[dotstyle=asterisk,dotscale=1.5](0,0)
	\rput{0}(1.8;0){\tiny$1$}
	\rput{0}(1.8;-20){\tiny$2k$}
	\rput{0}(1.8;160){\tiny$k$}
	\rput{0}(2;180){\tiny$k\!+\!1$}
	\pscurve[linecolor=line,linewidth=1.2pt](1.5;0)(1.4;0)(1.2;-10)(1.4;-20)(1.5;-20)
	\pscurve[linecolor=line,linewidth=1.2pt](1.5;160)(1.4;160)(-0.4;170)(1.4;180)(1.5;180)
	\pscurve[linecolor=line,linewidth=1.2pt](1.5;140)(1.4;140)(-0.7;170)(1.4;200)(1.5;200)
	\psarc[linestyle=dotted,linecolor=line,linewidth=1.2pt](0,0){1}{10}{120}
	\psarc[linestyle=dotted,linecolor=line,linewidth=1.2pt](0,0){1}{220}{-30}
   \end{pspicture}  
   \overset{\phi_{2k}}{\longmapsto}
  \begin{pspicture}[shift=-1.8](-2.25,-1.8)(2.25,1.8)
    	\pscircle[fillstyle=solid, fillcolor=diskin,linewidth=1.2pt](0,0){1.5}
    	\degrees[360]
	\rput{0}(1.8;0){\tiny$1$}
	\rput{0}(2;-20){\tiny$2k\!+\!1$}
	\rput{0}(2;-40){\tiny$2k$}
	\rput{0}(2;150){\tiny$k$}
	\rput{0}(2;170){\tiny$k\!+\!1$}
	\psline[linecolor=line,linewidth=1.2pt](1.5;-20)(0;0)
	\psdot[dotscale=1.5,linecolor=diskin](1.2;-20)
	\psdot[dotscale=1.4,linecolor=diskin](-0.5;160)
	\psdot[dotscale=1.4,linecolor=diskin](-0.8;160)
	\pscurve[linecolor=line,linewidth=1.2pt](1.5;0)(1.4;0)(1.2;-20)(1.4;-40)(1.5;-40)
	\pscurve[linecolor=line,linewidth=1.2pt](1.5;150)(1.4;150)(-0.5;160)(1.4;170)(1.5;170)
	\pscurve[linecolor=line,linewidth=1.2pt](1.5;130)(1.4;130)(-0.8;160)(1.4;190)(1.5;190)
	\psarc[linestyle=dotted,linecolor=line,linewidth=1.2pt](0,0){1}{10}{110}
	\psarc[linestyle=dotted,linecolor=line,linewidth=1.2pt](0,0){1}{220}{-50}
	\psdot[dotstyle=asterisk,dotscale=1.5](0,0)
   \end{pspicture}  
\end{align*}
and note that the image has $k$ under-crossings. Resolving all the crossings using the Kauffman skein relations gives a linear combination of link patterns.
The contribution to link pattern $L_\cap \in \mathcal{L}_{2k+1}$ comes from taking the smoothing \id for each crossing \cross.
Each of these contributions gives a factor $t^{-\frac{1}{4}}$, which establishes the result for $n$ even.

For the case $n=2k-1$ odd the first step of the argument is similar.
The only link pattern that does not contain a little arch connecting $i,i+1$ for any $1\leq i <2k-1$ is $L_\cap$.
By the mapping $\phi_{2k-1}$ we have at the skein module level

$$\psset{unit=0.8}
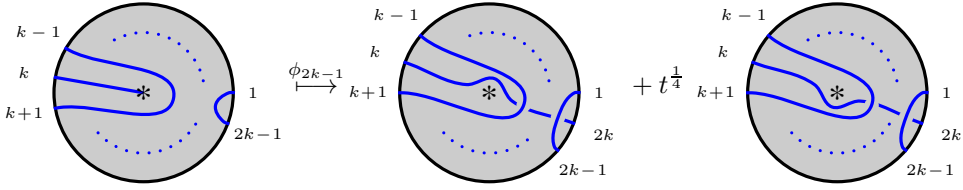
  \begin{pspicture}[shift=-1.8](-2.25,-1.8)(2.25,1.8)
    	\pscircle[fillstyle=solid, fillcolor=diskin,linewidth=1.2pt](0,0){1.5}
    	\degrees[360]
	\rput{0}(1.8;0){\tiny$1$}
	\rput{0}(2;-20){\tiny$2k\!-\!1$}
	\rput{0}(2;150){\tiny$k-1$}
	\rput{0}(2;170){\tiny$k$}
	\rput{0}(2;190){\tiny$k\!+\!1$}
	\pscurve[linecolor=line,linewidth=1.2pt](1.5;0)(1.4;0)(1.2;-10)(1.4;-20)(1.5;-20)
	\pscurve[linecolor=line,linewidth=1.2pt](1.5;150)(1.4;150)(-0.5;170)(1.4;190)(1.5;190)
	\psline[linecolor=line,linewidth=1.2pt](1.5;170)(0;0)
	\psarc[linestyle=dotted,linecolor=line,linewidth=1.2pt](0,0){1}{10}{120}
	\psarc[linestyle=dotted,linecolor=line,linewidth=1.2pt](0,0){1}{220}{-30}
	\psdot[dotstyle=asterisk,dotscale=1.5](0,0)
   \end{pspicture}   
   \overset{\phi_{2k-1}}{\longmapsto}
     \begin{pspicture}[shift=-1.8](-2.25,-1.8)(2.25,1.8)
    	\pscircle[fillstyle=solid, fillcolor=diskin,linewidth=1.2pt](0,0){1.5}
    	\degrees[360]
	\rput{0}(1.8;0){\tiny$1$}
	\rput{0}(2;-40){\tiny$2k\!-\!1$}
	\rput{0}(2;140){\tiny$k-1$}
	\rput{0}(2;160){\tiny$k$}
	\rput{0}(2;180){\tiny$k\!+\!1$}
	\rput{0}(2;-20){\tiny$2k$}
	\psarc[linestyle=dotted,linecolor=line,linewidth=1.2pt](0,0){1}{10}{110}
	\pscurve[linecolor=line,linewidth=1.2pt](1.5;160)(1.4;160)(0.4;160)(0.25;70)(0.4;-20)(1.4;-20)(1.5;-20)
	\psdot[dotscale=1.5,linecolor=diskin](1.2;-20)
	\psdot[dotscale=1.5,linecolor=diskin](0.6;-20)
	\pscurve[linecolor=line,linewidth=1.2pt](1.5;0)(1.4;0)(1.2;-20)(1.4;-40)(1.5;-40)
	\pscurve[linecolor=line,linewidth=1.2pt](1.5;140)(1.4;140)(-0.6;160)(1.4;180)(1.5;180)
	\psarc[linestyle=dotted,linecolor=line,linewidth=1.2pt,](0,0){1}{220}{-50}
	\psdot[dotstyle=asterisk,dotscale=1.5](0,0)
   \end{pspicture}  
   + t^{\frac{1}{4}}\;\;
    \begin{pspicture}[shift=-1.8](-2.25,-1.8)(2.25,1.8)
    	\pscircle[fillstyle=solid, fillcolor=diskin,linewidth=1.2pt](0,0){1.5}
    	\degrees[360]
	\rput{0}(1.8;0){\tiny$1$}
	\rput{0}(2;-40){\tiny$2k\!-\!1$}
	\rput{0}(2;140){\tiny$k-1$}
	\rput{0}(2;160){\tiny$k$}
	\rput{0}(2;180){\tiny$k\!+\!1$}
	\rput{0}(2;-20){\tiny$2k$}
	\psarc[linestyle=dotted,linecolor=line,linewidth=1.2pt](0,0){1}{10}{120}
	\pscurve[linecolor=line,linewidth=1.2pt](1.5;160)(1.4;160)(0.4;160)(0.25;-110)(0.4;-20)(1.4;-20)(1.5;-20)
	\psdot[dotscale=1.5,linecolor=diskin](1.2;-20)
	\psdot[dotscale=1.5,linecolor=diskin](0.6;-20)
	\pscurve[linecolor=line,linewidth=1.2pt](1.5;0)(1.4;0)(1.2;-20)(1.4;-40)(1.5;-40)
	\pscurve[linecolor=line,linewidth=1.2pt](1.5;140)(1.4;140)(-0.6;160)(1.4;180)(1.5;180)
	\psarc[linestyle=dotted,linecolor=line,linewidth=1.2pt](0,0){1}{220}{-30}
	\psdot[dotstyle=asterisk,dotscale=1.5](0,0)
   \end{pspicture}  
   $$
and note that each term in the image has $k-1$ under-crossings. 
Resolving all the crossings using Kauffman's skein relations gives a linear combination of link patterns.
The contributions to the link pattern $L_\cap \in \mathcal{L}_{2k}$ come from taking the smoothing \id for each crossing \cross in the first term.
Each of these contributions gives a factor $t^{-\frac{1}{4}}$, which establishes the result for $n$ odd.
\end{proof}
The next lemma provides necessary conditions on the parameters $q,c_n$ for the existence of a qKZ tower of solutions of minimal degree relative to the link pattern tower.
\begin{lemma}\label{CRUClem}
Let $v=1$ and $q,c_n, t^{\frac{1}{4}}\in\mathbb{C}^* \; (n\geq 1)$ with $(t^{\frac{1}{2}}+1)(t+1)\not=0$. Suppose that for each $n\geq 1$ there exists a $g^{(n)}\in\textup{Sol}_n(V_n;q,c_n)$ with 
\[
g^{(n)}_{L_\cap}(\mathbf{z})=\prod_{1\leq i<j\leq n}\left(t^{\frac{1}{2}}z_j-t^{-\frac{1}{2}}z_i\right).
\]
Write $g^{(0)}:=1\in V_0$.

Then, the following two statements are equivalent:
\begin{enumerate}
\item[{\bf (a)}]
$\bigl(g^{(n)}\bigr)_{n\geq 0}$ is a qKZ tower of solutions relative to the link pattern tower $\{(V_n,\phi_n)\}_{n\geq 0}$.
\item[{\bf (b)}]
$q=t^{\frac{3}{2}}$,
$c_n=\bigl(-t^{-\frac{3}{4}}\bigr)^{n-1}$ ($n\geq 1$) and $c_0=t^{\frac{1}{4}}+t^{-\frac{1}{4}}$.
\end{enumerate}
If these equivalent conditions are satisfied then 
$\lambda_n:=q^{-1}$ ($n\geq 1$), $\lambda_0=1$, 
\[
h^{(n)}(\mathbf{z})=t^{\frac{1}{4}(\lfloor n/2\rfloor-2n)}z_1z_2\cdots z_n\qquad (n\geq 1)
\]
and $h^{(0)}=1$. In other words, the corresponding braid recursion relations are then given by
\begin{equation}\label{braidrecursionrelation}
g^{(n+1)}(z_1,\ldots,z_n,0)=t^{\frac{1}{4}(\lfloor n/2\rfloor-2n)}z_1z_2\cdots z_n
\phi_n\bigl(g^{(n)}(z_1,\ldots,z_n)\bigr),\qquad n\geq 0.
\end{equation}
\end{lemma}
\begin{proof}
Note that for $n\geq 1$,
\begin{equation}\label{degreec}
(c_n)^ng^{(n)}(\mathbf{z})=\rho^ng^{(n)}(q^{-1}z_1,\ldots,q^{-1}z_n)=q^{-\frac{1}{2}n(n-1)}g^{(n)}(\mathbf{z})
\end{equation}
since $\rho^n$ acts as the identity on $V_n$ and $g^{(n)}$ is homogeneous of total degree
$\frac{1}{2}n(n-1)$. Hence, $(c_n)^n=q^{-\frac{1}{2}n(n-1)}$ ($n\geq 1$).
Furthermore, $c_1=1$ since $g^{(1)}$ is constant.

By the rank descent lemma, we have
\[
g^{(n+1)}(z_1,\ldots,z_n,0)\in\textup{Sol}_n\bigl(V_{n+1}^{\mathcal{I}_n};q,-t^{-\frac{3}{4}}c_{n+1}),
\]
while the representation lift lemma gives 
$\phi_n(g^{(n)}(z_1,\ldots,z_n))\in\textup{Sol}_n\bigl(V_{n+1}^{\mathcal{I}_n};q,c_n)$.
The fully nested component of $g^{(n+1)}(z_1,\ldots,z_n,0)$ is
\[
g_{L_\cap^{(n+1)}}^{(n+1)}(z_1,\ldots,z_n,0)=\bigl(-t^{-\frac{1}{2}}\bigr)^nz_1z_2\cdots z_n
\prod_{1\leq i<j\leq n}\left(t^{\frac{1}{2}}z_j-t^{-\frac{1}{2}}z_i\right).
\]
Using Lemma \ref{intertwinerfullynested}, the fully nested component of 
$\phi_n(g^{(n)}(z_1,\ldots,z_n))$ is 
\[
t^{-\frac{1}{4}\lfloor n/2\rfloor}\prod_{1\leq i<j\leq n}\left(t^{\frac{1}{2}}z_j-t^{-\frac{1}{2}}z_i\right).
\]
{\bf (a)}$\Rightarrow${\bf (b)}: assume that $(g^{(n)})_{n\geq 0}$ is a qKZ tower of solutions.
Then, the above analysis of the fully nested components implies that $\lambda_n=q^{-1}$ and 
\[
h^{(n)}(\mathbf{z})=t^{\frac{1}{4}(\lfloor n/2\rfloor-2n)}z_1z_2\cdots z_n
\]
for $n\geq 1$,
while $\lambda_0=1$, $h^{(0)}=1$ for $n=0$. Hence, the corresponding braid recursion takes on the explicit form \eqref{braidrecursionrelation}. Note that $c_0=t^{\frac{1}{4}}+t^{-\frac{1}{4}}$ since $g^{(0)}=1$. For $n\geq 1$, the left-hand side of \eqref{braidrecursionrelation} lies in 
$\textup{Sol}_n(V_{n+1}^{\mathcal{I}_n};q,-t^{-\frac{3}{4}}c_{n+1})$, while the right-hand side lies
in $\textup{Sol}_n(V_{n+1}^{\mathcal{I}_n};q,q^{-1}c_n)$, hence, the twist parameters $c_n$ must satisfy
$c_{n+1}=-q^{-1}t^{\frac{3}{4}}c_n$ ($n\geq 1$). Since $c_1=1$, we conclude that
\[
c_n=\bigl(-q^{-1}t^{\frac{3}{4}}\bigr)^{n-1},\qquad n\geq 1.
\]
Combined with \eqref{degreec} we obtain for $n\geq 1$,
\[
\bigl(q^{-2}t^{\frac{3}{2}}\bigr)^{\frac{1}{2}n(n-1)}=q^{-\frac{1}{2}n(n-1)},
\]
which is satisfied if and only if $q=t^{\frac{3}{2}}$. It follows that $c_n=(-t^{-\frac{3}{4}})^{n-1}$
for $n\geq 1$, as desired.\\
{\bf (b)}$\Rightarrow${\bf (a)}: in view of Lemmas \ref{unique1} and \ref{unique2} we only have to show that under the parameter conditions as stated in {\bf (b)}, the fully nested components of the left and right-hand side of \eqref{braidrecursionrelation} match. This can be confirmed by a direct computation.
\end{proof}

We can now state the main theorem of the paper. 
\begin{theorem}\label{mainTHM}
Let $t^{\frac{1}{4}}\in\mathbb{C}^*$ with $(t^{\frac{1}{2}}+1)(t+1)\not=0$ and set $v=1$, $q=t^{\frac{3}{2}}$. 
There exists, for all $n\geq 1$, 
a unique solution $g^{(n)}(\mathbf{z})\in\textup{Sol}_n\bigl(V_n;t^{\frac{3}{2}}, (-t^{-\frac{3}{4}})^{n-1}\bigr)$
homogeneous of total degree $\frac{1}{2}n(n-1)$, such that 
\[
g^{(n)}_{L_\cap}(\mathbf{z})=\prod_{1\leq i<j\leq n}\left(t^{\frac{1}{2}}z_j-t^{-\frac{1}{2}}z_i\right).
\]
Then, $(g^{(n)})_{n\geq 0}$, with $g^{(0)}:=1\in\textup{Sol}_0(V_0;t^{\frac{3}{2}},t^{\frac{1}{4}}+t^{-\frac{1}{4}})$, is a qKZ tower of solutions, with 
the associated braid recursion relations given by \eqref{braidrecursionrelation}.
\end{theorem}
The proof of the theorem will be given in Section \ref{sec-CM}. The key step is 
the construction of $g^{(n)}(\mathbf{z})$ for generic $t^{\frac{1}{4}}\in\mathbb{C}^*$ in terms of specialised non-symmetric dual Macdonald polynomials using the Cherednik-Matsuo correspondence \cite{Stokman:2011aa} and using results of Kasatani \cite{Kasatani:2005aa}.  The generic conditions on $t^{\frac{1}{4}}$ can then be removed by noting that the constructed solution $g^{(n)}(\mathbf{z})$ is well defined over $\mathbb{C}(t^{\frac{1}{4}})$ and the fact that the coefficients $g_L^{(n)}(\mathbf{z})$
for $L\in\mathcal{L}_n$ are regular at the values $t^{\frac{1}{4}}\in\mathbb{C}^*$ for which $(t^{\frac{1}{2}}+1)(t+1)\not=0$. Indeed, $g_{L_\cap}^{(n)}(\mathbf{z})$ is clearly regular at $t^{\frac{1}{4}}\in\mathbb{C}^*$. By the recursion relations expressing $g_L^{(n)}(\mathbf{z})$ in terms of $g_{L_\cap}^{(n)}(\mathbf{z})$
(see the proof of Lemma \ref{unique1} and Appendix A), it then follows inductively that 
all coefficients $g_L^{(n)}(\mathbf{z})$ ($L\in\mathcal{L}_n$) are regular at the values 
$t^{\frac{1}{4}}\in\mathbb{C}^*$ 
for which $(t^{\frac{1}{2}}+1)(t+1)\not=0$.

\begin{remark}\label{remTHM}
Note that for $t^{\frac{1}{4}}=\exp(\pi \mathrm{i}/3)$, we have $t^{\frac{3}{2}}=1$
and 
\[
-t^{\frac{1}{2}}-t^{-\frac{1}{2}}=1=t^{\frac{1}{4}}+t^{-\frac{1}{4}}.
\] 
The resulting
qKZ tower of solutions $(g^{(n)})_{n\geq 0}$ from Theorem \ref{mainTHM} is closely related to the inhomogeneous
dense O(1) loop model on the half-infinite cylinder, see Section \ref{sec-groundstate} and
\cite{Di-Francesco:2006aa}. In fact, the constituents 
$g^{(n)}\in\textup{Sol}_n\bigl(V_n;1,1\bigr)$ then are the renormalized ground states of the
inhomogeneous O(1) dense loop models on the half-infinite cylinder. In this case, the braid
recursion relations reduce to 
\begin{equation*}
\begin{split}
g^{(2k)}(z_1,\ldots,z_{2k-1},0)&=(-1)^kt^{-\frac{1}{2}}z_1\cdots z_{2k-1}\phi_{2k-1}(
g^{(2k-1)}(z_1,\ldots,z_{2k-1})),\\
g^{(2k+1)}(z_1,\ldots,z_{2k},0)&=(-1)^kz_1\cdots z_{2k}\phi_{2k}\bigl(g^{2k)}(z_1,\ldots,z_{2k})).
\end{split}
\end{equation*}
\end{remark}

\section{Existence of solution for $t^{\frac{1}{4}}= \exp(\pi \mathrm{i} /3)$} \label{sec-groundstate}

In this section we recall the construction of the polynomial solutions $g^{(n)}(\mathbf{z}) \in \textup{Sol}_n(V_n(1);1,1)$ of degree $\frac{1}{2}n(n-1)$ for $v=1$ and $t^{\frac{1}{4}}=\exp(\pi \mathrm{i}/3)$ (see Theorem \ref{mainTHM}). In this special case, the construction of the qKZ tower of solutions is facilitated by the fact that the underlying integrable model, the inhomogeneous dense O(1) loop model on the half-infinite cylinder, is stochastic. This allows one to construct $g^{(n)}(\mathbf{z})$ as a suitably renormalized version of the ground state of the inhomogeneous dense O(1) loop model, following \cite{Di-Francesco:2006aa}.

The section begins with discussing the Temperley-Lieb transfer operator and then we specialize the analysis to the inhomogeneous dense O(1) loop model on the half-infinite cylinder.
In this section $v=1$. 

\subsection{Transfer operator}

The transfer operator $\widehat{T}^{(n)} :=\widehat{T}(x; z_1,\ldots,z_n) : \C[\mathcal{L}_n] \to \C[\mathcal{L}_n] $ can be defined as follows \cite{Di-Francesco:2005aa,Di-Francesco:2006aa}.
For $n>0$ consider the following two tiles 
$$ \ltileb \hspace{3cm} \ltilea $$
which we denote by $\tau^\text{nw}$ and $\tau^\text{ne}$, respectively, where `nw' and `ne' indicate that the north edge of the tile is connected to the west or east edge by an arc.
Then, $\widehat{T}^{(n)}(x;\mathbf{z})=\widehat{T}^{(n)}(x;z_1,\ldots,z_n)$ is defined by 
$$ 
\widehat{T}^{(n)} (x;\mathbf{z}):=\sum_{\tau_1,\ldots, \tau_n}\left (\prod_{i=1}^n P_{\tau_i}(x/z_i)\right)  
\psset{unit=0.6}
\begin{pspicture}[shift=-1.8](-2.2,-2)(2,2)
    	\pscircle[fillstyle=solid,fillcolor= diskin,linewidth=1.2pt](0,0){2}
    	\pscircle[fillstyle=solid, fillcolor=white, linewidth=1.2pt](0,0){1}
    	\degrees[14]
     	\psline(1;0)(2;0)
	\psline(1;1)(2;1)
	\psline(1;2)(2;2)
	\psline(1;10)(2;10)
	\psline(1;11)(2;11)
	\psarc[linestyle=dotted,linewidth=1.2pt](0,0){1.5}{3}{9}
	\psarc[linestyle=dotted,linewidth=1.2pt](0,0){1.5}{12}{13}
	\rput{4}(1.5;0.5){\small$\tau_1$}
	\rput{5}(1.5;1.5){\small$\tau_2$}
	\rput{0}(1.5;10.5){\small$\tau_i$}
   \end{pspicture}  
   $$
where $\tau_i \in \{\tau^\text{nw}, \tau^\text{ne}\}$, 
\begin{equation*}
\begin{split}
P_{\tau^\text{nw}}(x/z_i) &= a(x/z_i)  =  \frac{x -z_i}{t^{\frac{1}{2}}z_i-t^{-\frac{1}{2}} x},\\
P_{\tau^\text{ne}}(x/z_i) &= b(x/z_i) = \frac{t^{\frac{1}{2}}x - t^{-\frac{1}{2}}z_i}{t^{\frac{1}{2}}z_i -t^{-\frac{1}{2}} x}.
\end{split}
\end{equation*}
Note that the  inner boundary of the annulus is always taken as the north edge of the tile.
Moreover, for the case $n=1$, tiling the annulus is done by stretching the tile so that the east and west edges are identified.
The string of tiles covering the annulus can immediately be interpreted as an element in $\mathcal{S}_n(t^{\frac{1}{4}})$. Hence, by the algebra isomorphism $\theta_n: \mathcal{TL}_n(t^{\frac{1}{2}})\overset{\sim}{\longrightarrow}
\textup{End}_{\mathcal{S}\left(t^{\frac{1}{4}}\right)}(n)$ we have $\widehat{T}^{(n)}(x;\mathbf{z}) \in \C(x,\mathbf{z}) \otimes \mathcal{TL}_n(t^{\frac{1}{2}}).$

The case $n=0$ is special.
We define $\widehat{T}^{(0)} := \theta_0(X)$ (recall that $\mathcal{TL}_0=\mathbb{C}[X]$).
We also point out that since $\mathcal{TL}_1=\mathbb{C}[\rho,\rho^{-1}]$ we have 
$$\widehat{T}^{(1)} (x;z_1)=\frac{x -z_1}{t^{\frac{1}{2}}z_1-t^{-\frac{1}{2}}x} \theta_1(\rho^{-1})+ \frac{t^{\frac{1}{2}}x - t^{-\frac{1}{2}}z_1}{t^{\frac{1}{2}}z_1 -t^{-\frac{1}{2}} x} \theta_1(\rho).$$

We will drop the isomorphism $\theta_n$ when it is clear from context.
Using diagrams we write the $R$-operator as
\begin{equation} \label{Rop}
\psset{unit=0.8}
R_i(z_{i+1}/z_i) = \frac{z_{i+1} -z_i}{t^{\frac{1}{2}}z_i-t^{-\frac{1}{2}} z_{i+1}}
\begin{pspicture}[shift=-1.65](-1.75,-1.75)(1.75,1.5)
    	\SpecialCoor
    	\pscircle[fillstyle=solid,fillcolor = diskin,linewidth=1.2pt](0,0){1.25}
    	\pscircle[fillstyle=solid, fillcolor=white, linewidth=1.2pt](0,0){0.5}
    	\degrees[17]
     	\rput{0}(1.5;0){\tiny $z_1$}
	\rput{0}(1.7;9.9){\tiny $z_{i-1}$}
    	\rput{0}(1.6;10.9){\tiny $z_i$}
	\rput{0}(1.6;12){\tiny $z_{i+1}$}
	\rput{0}(1.6;13.3){\tiny $z_{i+2}$}
     	\rput{0}(0.25;0){\tiny $1$}
    	\rput{0}(0.25;10.8){\tiny $i$}
	\psline[linecolor=line,linewidth=1.2pt](0.5;0)(1.25;0)
	\psline[linecolor=line,linewidth=1.2pt](0.5;10)(1.25;10)
	\psline[linecolor=line,linewidth=1.2pt](0.5;13)(1.25;13)
	\pscurve[linecolor=line,linewidth=1.2pt](0.5;11)(0.6;11)(0.75;11.5)(0.6;12)(0.5;12)
	\pscurve[linecolor=line,linewidth=1.2pt](1.25;11)(1.15;11)(0.9;11.5)(1.15;12)(1.25;12)
	\psarc[linestyle=dotted, linecolor=line,linewidth=1.2pt](0,0){0.85}{1}{9}
	\psarc[linestyle=dotted, linecolor=line,linewidth=1.2pt](0,0){0.85}{14}{16}
   \end{pspicture}
    + \frac{t^{\frac{1}{2}}z_{i+1} - t^{-\frac{1}{2}}z_i}{t^{\frac{1}{2}}z_i -t^{-\frac{1}{2}} z_{i+1}} 
    \begin{pspicture}[shift=-1.65](-1.75,-1.75)(1.75,1.5)
    	\SpecialCoor
    	\pscircle[fillstyle=solid,fillcolor = diskin,linewidth=1.2pt](0,0){1.25}
    	\pscircle[fillstyle=solid, fillcolor=white, linewidth=1.2pt](0,0){0.5}
    	\degrees[17]
    	\rput{0}(1.5;0){\tiny $z_1$}
	\rput{0}(1.7;9.9){\tiny $z_{i-1}$}
    	\rput{0}(1.6;10.9){\tiny $z_i$}
	\rput{0}(1.6;12){\tiny $z_{i+1}$}
	\rput{0}(1.6;13.3){\tiny $z_{i+2}$}
     	\rput{0}(0.25;0){\tiny $1$}
    	\rput{0}(0.25;10.8){\tiny $i$}
	\psline[linecolor=line,linewidth=1.2pt](0.5;0)(1.25;0)
	\psline[linecolor=line,linewidth=1.2pt](0.5;10)(1.25;10)
	\psline[linecolor=line,linewidth=1.2pt](0.5;13)(1.25;13)
	\psline[linecolor=line,linewidth=1.2pt](0.5;12)(1.25;12)
	\psline[linecolor=line,linewidth=1.2pt](0.5;11)(1.25;11)
	\psarc[linestyle=dotted, linecolor=line,linewidth=1.2pt](0,0){0.85}{1}{9}
	\psarc[linestyle=dotted, linecolor=line,linewidth=1.2pt](0,0){0.85}{14}{16}
   \end{pspicture},
\end{equation}
and also as 
\[R_i(z_{i+1}/z_i)
	 =
	  \begin{pspicture}[shift=-1.65](-1.75,-1.75)(2,1.5)
    	\SpecialCoor
    	\pscircle[fillstyle=solid,fillcolor = diskin,linewidth=1.2pt](0,0){1.25}
    	\pscircle[fillstyle=solid, fillcolor=white, linewidth=1.2pt](0,0){0.5}
    	\degrees[17]
    	\rput{0}(1.5;0){\tiny $z_1$}
	\rput{0}(1.7;9.9){\tiny $z_{i-1}$}
    	\rput{0}(1.6;10.9){\tiny $z_i$}
	\rput{0}(1.6;12){\tiny $z_{i+1}$}
	\rput{0}(1.6;13.3){\tiny $z_{i+2}$}
     	\rput{0}(0.25;0){\tiny $1$}
    	\rput{0}(0.25;10.8){\tiny $i$}
	\psline[linecolor=line,linewidth=1.2pt](0.5;0)(1.25;0)
	\psline[linecolor=line,linewidth=1.2pt](0.5;10)(1.25;10)
	\psline[linecolor=line,linewidth=1.2pt](0.5;13)(1.25;13)
	\pscurve[linecolor=line,linewidth=1.2pt](0.5;11)(0.6;11)(0.85;11.5)(1.15;12)(1.25;12)
	\pscurve[linecolor=line,linewidth=1.2pt](1.25;11)(1.15;11)(0.85;11.5)(0.6;12)(0.5;12)
	\psarc[linestyle=dotted, linecolor=line,linewidth=1.2pt](0,0){0.85}{1}{9}
	\psarc[linestyle=dotted, linecolor=line,linewidth=1.2pt](0,0){0.85}{14}{16}
   \end{pspicture} 
   \]
where we view the crossing in the annulus as a weighted sum of the two diagrams given in \eqref{Rop}. 
Using the diagram description of the $R$-operator, the Yang-Baxter equations and inversion relation (lines 1 and 3 of  \eqref{ybe}) can be depicted as
\begin{equation}\label{ybed}
\psset{unit=0.6cm}
\begin{pspicture}[shift=-3.5](11,7)
\psline*[linecolor=diskin](1.5,5.5)(2.5,6)(4.5,4.5)(4.5,2.5)(2.5,1)(1.5,1.5)
\psline[linestyle=dotted,linewidth=0.8pt](1.5,5.5)(2.5,6)(4.5,4.5)(4.5,2.5)(2.5,1)(1.5,1.5)(1.5,5.5)
\psline[linewidth=1.2pt,linecolor=line](1.5,5.5)(4.5,2.5)
\psline[linewidth=1.2pt,linecolor=line](4.5,4.5)(1.5,1.5)
\psline[linewidth=1.2pt,linecolor=line](2.5,6)(2.5,1)
\rput(1,6){$x$}
\rput(2.5,6.5){$y$}
\rput(5,5){$z$}
\rput(1,1){$z$}
\rput(2.5,0.5){$y$}
\rput(5,2){$x$}
\rput(5.5,3.5){$=$}
\psline*[linecolor=diskin](9.5,5.5)(8.5,6)(6.5,4.5)(6.5,2.5)(8.5,1)(9.5,1.5)
\psline[linestyle=dotted,linewidth=0.8pt](9.5,5.5)(8.5,6)(6.5,4.5)(6.5,2.5)(8.5,1)(9.5,1.5)(9.5,5.5)
\psline[linewidth=1.2pt,linecolor=line](9.5,5.5)(6.5,2.5)
\psline[linewidth=1.2pt,linecolor=line](6.5,4.5)(9.5,1.5)
\psline[linewidth=1.2pt,linecolor=line](8.5,6)(8.5,1)
\rput(10,6){$z$}
\rput(8.5,6.5){$y$}
\rput(6,5){$x$}
\rput(10,1){$x$}
\rput(8.5,0.5){$y$}
\rput(6,2){$z$}
\end{pspicture}
\text{ and } \qquad %
\begin{pspicture}[shift=-3](8,6)
\psframe*[linecolor=diskin](0.5,1)(3.5,5)
\psline[linestyle=dotted,linewidth=0.8pt](0.5,1)(3.5,1)(3.5,5)(0.5,5)(0.5,1)
\psline[linewidth=1.2pt,linecolor=line](1,1)(3,3)(1,5)
\psline[linewidth=1.2pt,linecolor=line](3,1)(1,3)(3,5)
\rput(1,0.5){$x$}
\rput(1,5.5){$x$}
\rput(3,0.5){$y$}
\rput(3,5.5){$y$}
\rput(4,3){$=$}
\psframe*[linecolor=diskin](4.5,1)(7.5,5)
\psline[linestyle=dotted,linewidth=0.8pt](4.5,1)(7.5,1)(7.5,5)(4.5,5)(4.5,1)
\psline[linewidth=1.2pt,linecolor=line](5,1)(5,5)
\psline[linewidth=1.2pt,linecolor=line](7,1)(7,5)
\rput(7,0.5){$y$}
\rput(7,5.5){$y$}
\rput(5,0.5){$x$}
\rput(5,5.5){$x$}
\end{pspicture}
\end{equation}
respectively.
The area within the dotted lines is a local neighbourhood in the annulus.

The transfer operator can now be defined in terms of the $R$-operators,
$R_i(x)$ for $ i \in \Z/n\Z$, as follows.
Let 
$$ M^{(n)}_0(x;\mathbf{z}) := \rho R_{n-1}(x/z_{n}) R_{n-2}(x/z_{n-1}) \cdots R_0(x/z_{1})   \in \mathcal{TL}_{n+1} $$ 
be the monodromy operator where we view the auxiliary point as $n+1\equiv 0$ (modulo $n+1$).
Then,
$$ \widehat{T}^{(n)}(x;\mathbf{z}) := \text{cl}_0\left( M^{(n)}_0(x;\mathbf{z})\right) $$
where $\text{cl}_0$ corresponds to the tangle closure \cite{Goodman:2006aa} at the auxiliary point $0$.
In this specific case $\text{cl}_0$ amounts to disconnecting the two arcs from the  inner- and outer boundary points labeled `0' and connecting them in $\textup{End}_{\mathcal{S}}(n)$ by an arc that under-crosses all arcs one meets. 

The transfer operators with different values of $x$ commute in $\mathcal{TL}_n$,
\[
\lbrack \widehat{T}_n(x;\mathbf{z}),\widehat{T}_n(x^\prime,\mathbf{z})\rbrack=0.
\]
This can be shown by interlacing two T operators with $R$-operators.
In the literature, it is usually shown diagrammatically using the inversion relation and Yang-Baxter equation \eqref{ybed} of the $R$-operators.
For an example of this technique, we refer the reader to \cite{Di-Francesco:2005aa} for dense loop models and \cite{Baxter:1982aa} in general.
Using the Yang-Baxter equation and the relations involving $\rho$ (see \eqref{relTL}) one shows that 
\begin{equation}\label{eq-rttr}
\begin{split}
R_i(z_{i+1}/z_i)  \widehat{T}^{(n)}(x; \ldots,z_{i+1},z_{i},\ldots) &= \widehat{T}^{(n)}(x;\mathbf{z}) R_i(z_{i+1}/z_i),\\
\rho \widehat{T}^{(n)}(x; z_{2},\ldots,z_{n},z_1) &= \widehat{T}^{(n)}(x; \mathbf{z}) \rho.
\end{split}
\end{equation}

In  \cite{Di-Francesco:2006aa} the authors made the crucial observation that the $R$-operators 
$R_i(0), R_i(\infty)\in\mathcal{TL}_n$ can be interpreted as a single crossing in the skein description of the element, 
$$R_{i}(0) = -t^{-\frac{3}{4}}  
	   \begin{pspicture}[shift=-1.4](-1.5,-1.5)(2,1.5)
    	\pscircle[fillstyle=solid,fillcolor = diskin,linewidth=1.2pt](0,0){1.25}
    	\pscircle[fillstyle=solid, fillcolor=white, linewidth=1.2pt](0,0){0.5}
    	\degrees[18]
	\rput{0}(1.5;0){\tiny $1$}
	\rput{0}(0.25;0){\tiny $1$}
     	\rput{0}(1.5;12){\tiny $i$}
	\rput{0}(1.5;13){\tiny $i+1$}
	 \rput{0}(0.25;12){\tiny $i$}
	   \psline[linecolor=line,linewidth=1.2pt](1.25;0)(0.5;0)
	\psline[linecolor=line,linewidth=1.2pt](1.25;11)(0.5;11)
	\psline[linecolor=line,linewidth=1.2pt](1.25;14)(0.5;14)
	\psline[linecolor=line,linewidth=1.2pt](1.25;15)(0.5;15)
	\pscurve[linecolor=line,linewidth=1.2pt](1.25;13)(1.15;13)(0.8;12.5)(0.6;12)(0.5;12)
	\psdot[linecolor=diskin,dotsize=0.25](0.8;12.5)
	\pscurve[linecolor=line,linewidth=1.2pt](1.25;12)(1.15;12)(0.8;12.5)(0.6;13)(0.5;13)
	\psarc[linestyle=dotted, linecolor=line,linewidth=1.2pt](0,0){0.8}{1}{10}
	\psarc[linestyle=dotted, linecolor=line,linewidth=1.2pt](0,0){0.8}{16}{17}
   \end{pspicture},  
   \qquad
   R_{i}(\infty) = -t^{\frac{3}{4}}  
	   \begin{pspicture}[shift=-1.4](-1.5,-1.5)(2,1.5)
    	\pscircle[fillstyle=solid,fillcolor = diskin,linewidth=1.2pt](0,0){1.25}
    	\pscircle[fillstyle=solid, fillcolor=white, linewidth=1.2pt](0,0){0.5}
    	\degrees[18]
	\rput{0}(1.5;0){\tiny $1$}
	\rput{0}(0.25;0){\tiny $1$}
     	\rput{0}(1.5;12){\tiny $i$}
	\rput{0}(1.5;13){\tiny $i+1$}
	 \rput{0}(0.25;12){\tiny $i$}
	   \psline[linecolor=line,linewidth=1.2pt](1.25;0)(0.5;0)
	\psline[linecolor=line,linewidth=1.2pt](1.25;11)(0.5;11)
	\psline[linecolor=line,linewidth=1.2pt](1.25;14)(0.5;14)
	\psline[linecolor=line,linewidth=1.2pt](1.25;15)(0.5;15)
	\pscurve[linecolor=line,linewidth=1.2pt](1.25;12)(1.15;12)(0.8;12.5)(0.6;13)(0.5;13)
	\psdot[linecolor=diskin,dotsize=0.25](0.8;12.5)
	\pscurve[linecolor=line,linewidth=1.2pt](1.25;13)(1.15;13)(0.8;12.5)(0.6;12)(0.5;12)
	\psarc[linestyle=dotted, linecolor=line,linewidth=1.2pt](0,0){0.8}{1}{10}
	\psarc[linestyle=dotted, linecolor=line,linewidth=1.2pt](0,0){0.8}{16}{17}
   \end{pspicture}.$$
Consequently, 
$$ 
\widehat{T}^{(n)} (x; z_1,\ldots,z_{n-1},0)=  -t^{\frac{3}{4}}  \sum_{\tau_1,\ldots, \tau_{n-1}}\left (\prod_{i=1}^{n-1} P_{\tau_i}(x/z_i) \right)  
\psset{unit=0.6}
\begin{pspicture}[shift=-1.8](-2.2,-2)(2.2,2)
    	\pscircle[fillstyle=solid,fillcolor= diskin,linewidth=1.2pt](0,0){2}
    	\pscircle[fillstyle=solid, fillcolor=white, linewidth=1.2pt](0,0){1}
    	\degrees[14]
     	\psline(1;0)(2;0)
	\psline(1;2)(2;2)
	\psline(1;10)(2;10)
	\psline(1;12)(2;12)
	\psarc[linestyle=dotted,linewidth=1.2pt](0,0){1.5}{3}{9}
	\psarc[linestyle=dotted,linewidth=1.2pt](0,0){1.5}{12}{12}
	\psline[linecolor=line,linewidth=1.2pt](1;13)(2;13)
	\psdot[linecolor=diskin,dotsize=0.25](1.5;13)
	\psarc[linecolor=line,linewidth=1.2pt](0,0){1.5}{12}{14}
	\rput{5}(1.5;1){\small$\tau_1$}
	\rput{0}(1.5;11){\small$\tau_{n-1}$}
   \end{pspicture} .
   $$
Noting this over crossing and recalling the algebra map $\mathcal{I}_{n-1}: \mathcal{TL}_{n-1}
\rightarrow \mathcal{TL}_n$ arising from the arc insertion functor,
we obtain the following braid recursion relation for the transfer operator, which is due to \cite[\S 2.4]{Di-Francesco:2006aa}:
\begin{proposition} \label{tmat}
For $n\geq 1$,
$$ \widehat{T}^{(n)}(x;z_1,\ldots, z_{n-1},0) = -t^{\frac{3}{4}}\mathcal{I}_{n-1} \left( \widehat{T}^{(n-1)}(x;z_1,\ldots, z_{n-1})\right).$$
\end{proposition}

\subsection{The inhomogeneous dense O(1) loop model}
The transfer operator $\widehat{T}^{(n)}(x;\mathbf{z})\in\mathcal{TL}_n$ acting
on the link pattern tower representation $V_n$ in the special case $v=1$ is by definition the
transfer operator $T^{(n)}(x;\mathbf{z})\in\textup{End}(V_n)$
of the inhomogeneous dense $\textup{O}(-t^{\frac{1}{2}}-t^{-\frac{1}{2}})$ loop
model on the punctured disc \cite{Di-Francesco:2006aa,Kasatani:2007aa}. 
We specialize in this section further to the case
$t^{\frac{1}{4}}=\exp(\pi \mathrm{i}/3)$, in which case
\[
-t^{\frac{1}{2}} - t^{-\frac{1}{2}} = 1= t^{\frac{1}{4}} + t^{-\frac{1}{4}}.
\]
This means that all loops can be removed by a factor $1$.
As we shall discuss in a moment, the resulting $\textup{O}(1)$-model is not only Bethe integrable but also stochastic. We identfy $V_n$ with $\mathbb{C}[\mathcal{L}_n]$ as vector spaces (see the end
of Section \ref{EATLsection}).

In \cite{Di-Francesco:2006aa} the authors stated the existence and uniqueness of a suitably normalized ground state of the inhomogeneous dense $\textup{O}(1)$ loop model, with $\mathbf{z}$ regarded as formal variables. 
For the convenience of the reader, we provide a full proof of this result. It uses the irreducibility and stochasticity of the transfer operator $\widehat{T}^{(n)}(x;\mathbf{z})$ for a particular parameter regime, and it uses the algebraic dependence of $\widehat{T}^{(n)}(x;\mathbf{z})$ on $x$ and $\mathbf{z}$. 

Consider the matrix $A^{(n)}(x;\mathbf{z}):=(A_{LL'} (x;\mathbf{z}))_{L,L' \in \mathcal{L}_n}$ of
$T^{(n)}(x;\mathbf{z})$ with respect to the link pattern basis,
\[
	T^{(n)}(x;\mathbf{z}) L' = \sum_{L\in\mathcal{L}_n} A_{LL'}(x;\mathbf{z})L.
\]
The coefficients $A_{LL^\prime}(x;\mathbf{z})$ depend rationally on $x,z_1,\ldots,z_n$.
For the special value $t^{\frac{1}{4}}=\exp(\pi \mathrm{i}/3)$ the Boltzmann weights $a(x)$ and $b(x)$ (see
\eqref{weights}) satisfy 
\[
a(x)+b(x)=1,
\]
hence $\sum_{L\in\mathcal{L}_n}A_{LL^\prime}(x;\mathbf{z})=1$ for all $L^\prime\in\mathcal{L}_n$.
Furthermore, we have $0<a(x)<1$ if $x=e^{\mathrm{i}\theta}$ with $0<\theta<2\pi/3$; hence, $A^{(n)}(x;\mathbf{z})$
is left-stochastic if $x/z_j=e^{\mathrm{i}\theta_j}$ with $0<\theta_j<2\pi/3$ for $j=1,\ldots,n$. In this situation, 
$A^{(n)}(x;\mathbf{z})$ is irreducible; this follows from the fact that each $L\in\mathcal{L}_n$ is a 
cyclic vector for the $\mathcal{TL}_n$-module $V_n$, which can be proven as follows.

For $n=2k$ even, let $L_{ln}\in\mathcal{L}_{2k}$ be the \emph{least-nested} link pattern, which is the link pattern that has little arches connecting boundary points $(2i-1,2i)$ for $1\leq i \leq k$ such that the little-arches do not contain the puncture.
All $L\in\mathcal{L}_n$ can be mapped to $L_{ln}$ by acting with $e_{1} e_{3} \cdots e_{2k-1}$.
In turn, $L_{ln}$ can be mapped to the fully nested link pattern $L_\cap$ by the action of $\rho^{k}g_kg_{k-1}\cdots g_2$ with $g_i:=e_{i}e_{i+2}\cdots e_{2k-i}$.
Lastly, by the inductive argument in Appendix \ref{uniqueproof}, $L_\cap$ can be mapped to any $L\in\mathcal{L}_n$.
The case for $n$ odd is analogous.
\begin{lemma}
Let $v=1$, $q=1$ and $t^{\frac{1}{4}}=\exp(\pi \mathrm{i}/3)$.
There exists a unique $\widehat{g}^{(n)}(\mathbf{z})=\sum_{L\in\mathcal{L}_n}\widehat{g}^{(n)}_L(\mathbf{z})L$ with $\widehat{g}_L^{(n)}(\mathbf{z})\in \mathbb{C}(\mathbf{z})$ such that
\[
T^{(n)}(x;\mathbf{z})\widehat{g}^{(n)}(\mathbf{z})=\widehat{g}^{(n)}(\mathbf{z})
\]
for all $x\in\mathbb{C}$ and such that $\sum_{L\in\mathcal{L}_n}\widehat{g}^{(n)}_L(\mathbf{z})=1$. Furthermore,
\[
\widehat{g}^{(n)}(\mathbf{z})\in \bigl(\mathbb{C}(\mathbf{z})\otimes V_n\bigr)^{\nabla(W_n)}.
\]
\end{lemma}
\begin{proof}
Consider $A^{(n)}(\mathbf{z}):=A^{(n)}(1;\mathbf{z})$. Since the matrix coefficients $A_{LL^\prime}(\mathbf{z}):=A_{LL^\prime}(1;\mathbf{z})$ satisfy $\sum_{L\in\mathcal{L}_n}A_{LL^\prime}(\mathbf{z})=1$, we have
$\textup{det}\bigl(A^{(n)}(\mathbf{z})-1\bigr)=0$ and hence there exists a nonzero vector
$\kappa(\mathbf{z})=\bigl(\kappa_L(\mathbf{z})\bigr)_{L\in\mathcal{L}_n}$ with $\kappa_L(\mathbf{z})\in\mathbb{C}(\mathbf{z})$ such that $A^{(n)}(\mathbf{z})\kappa(\mathbf{z})=\kappa(\mathbf{z})$. Consider
\[
N(\mathbf{z}):=\sum_{L\in\mathcal{L}_n}\kappa_L(\mathbf{z}).
\]
Note that $A^{(n)}(\mathbf{z})$ is irreducible left-stochastic if $z_j=e^{-\mathrm{i}\theta_j}$ with $0<\theta_j<2\pi/3$; hence, for generic specialized values of the rapidities in this stochastic parameter regime, $A^{(n)}(\mathbf{z})$
has a one-dimensional eigenspace with eigenvalue $1$, spanned by the Frobenius-Perron eigenvector $v^{FP}(\mathbf{z})$, and the Frobenius-Perron eigenvector $v^{FP}(\mathbf{z})$ (normalized such that the sum of the coefficients is one), has the property that all its coefficients are positive. Hence, for generic values of the rapidities in the stochastic parameter regime, $N(\mathbf{z})\not=0$. In particular, $N(\mathbf{z})\in\mathbb{C}(\mathbf{z}){\setminus}\{0\}$, and we may set
$\widehat{g}^{(n)}(\mathbf{z}):=\sum_{L\in\mathcal{L}_n}g_L^{(n)}(\mathbf{z})L$ with 
$\widehat{g}^{(n)}_L(\mathbf{z}):=\kappa_L(\mathbf{z})/N(\mathbf{z})\in\mathbb{C}(\mathbf{z})$. Then, 
\[
T^{(1)}(1;\mathbf{z})\widehat{g}^{(n)}(\mathbf{z})=\widehat{g}^{(n)}(\mathbf{z})
\]
and $\sum_{L\in\mathcal{L}_n}\widehat{g}_L^{(n)}(\mathbf{z})=1$. It follows from restricting to the stochastic parameter regime again that these two properties determine 
$\widehat{g}^{(n)}(\mathbf{z})$ uniquely.

Let $x\in\mathbb{C}$ and set 
\[
\widehat{g}^{(n)}(x;\mathbf{z}):=T^{(n)}(x;\mathbf{z})\widehat{g}^{(n)}(\mathbf{z}).
\]
Write
\[
\widehat{g}^{(n)}(x;\mathbf{z})=\sum_{L\in\mathcal{L}_n}\widehat{g}_{L}^{(n)}(x;\mathbf{z})L
\]
with $\widehat{g}_L^{(n)}(x;\mathbf{z})\in\mathbb{C}(\mathbf{z})$. Since $\lbrack T^{(n)}(1;\mathbf{z}),
T^{(n)}(x;\mathbf{z})\rbrack=0$ we have 
\[
T^{(n)}(1;\mathbf{z})\widehat{g}^{(n)}(x;\mathbf{z})=\widehat{g}^{(n)}(x;\mathbf{z}).
\]
Since $\sum_{L\in\mathcal{L}_n}A_{LL^\prime}(x;\mathbf{z})=1$ for all $L^\prime\in\mathcal{L}_n$,
we furthermore have $\sum_{L\in\mathcal{L}_n}\widehat{g}_L^{(n)}(x;\mathbf{z})=1$. Hence
$\widehat{g}^{(n)}(x;\mathbf{z})=\widehat{g}^{(n)}(\mathbf{z})$, i.e.
\[
T^{(n)}(x;\mathbf{z})\widehat{g}^{(n)}(\mathbf{z})=\widehat{g}^{(n)}(\mathbf{z}).
\]
This completes the proof of the uniqueness and existence of $\widehat{g}_n(\mathbf{z})$.

For the second statement, let $1\leq i<n$ and set 
$h_i(\mathbf{z}):=R_i(z_{i+1}/z_i)\widehat{g}^{(n)}(s_i\mathbf{z})$. Then, by the first formula
of \eqref{eq-rttr},
\[
\widehat{T}^{(n)}(x;\mathbf{z})h_i(\mathbf{z})=h_i(\mathbf{z}),
\]
and the sum of the coefficients of $h_i(\mathbf{z})$ is one since $a(x)+b(x)=1$. Hence
$h_i(\mathbf{z})=\widehat{g}^{(n)}(\mathbf{z})$, i.e.,
\[
R_i(z_{i+1}/z_i)\widehat{g}^{(n)}(s_i\mathbf{z})=\widehat{g}^{(n)}(\mathbf{z}).
\]
In the same way, one shows that $\rho \widehat{g}^{(n)}(z_2,\ldots,z_n,z_1)=\widehat{g}^{(n)}(\mathbf{z})$, now using the second equality of \eqref{eq-rttr}. This completes the proof of the lemma.
\end{proof}
Now we are ready to prove Theorem \ref{mainTHM} in the special case that $t^{\frac{1}{4}}=
\exp(\pi i/3)$. By Corollary \ref{degreeleading}, the fully
nested component is of the form
\[
\widehat{g}^{(n)}_{L_\cap}(\mathbf{z})=C_n(\mathbf{z})\prod_{1\leq i<j\leq n}
\left(t^{\frac{1}{2}}z_j-t^{-\frac{1}{2}}z_i\right)
\]
with $0\not=C_n(\mathbf{z})\in\mathbb{C}(z)^{S_n}$. Since in the present situation $q=1$ and $C_n(\mathbf{z})$ is symmetric, we have that the renormalized function
\[
g^{(n)}(\mathbf{z}):=C_n(\mathbf{z})^{-1}\widehat{g}^{(n)}(\mathbf{z})
\]
is also a symmetric solution of the qKZ equations, $g^{(n)}(\mathbf{z})\in
\bigl(\mathbb{C}(\mathbf{z})\otimes V_n\bigr)^{\nabla(W_n)}$.
Now $g^{(n)}(\mathbf{z})$ has fully nested component
\begin{equation}\label{fnc}
g^{(n)}_{L_\cap}(\mathbf{z})=\prod_{1\leq i<j\leq n}
\left(t^{\frac{1}{2}}z_j-t^{-\frac{1}{2}}z_i\right).
\end{equation}
By Lemma \ref{unique1} we conclude that $g^{(n)}(\mathbf{z})\in\textup{Sol}_n(V_n;1,1)$
is a homogeneous polynomial solution of total degree $\frac{1}{2}n(n-1)$, which completes the proof of Theorem \ref{mainTHM} in the special case that $t^{\frac{1}{4}}=\exp(\pi i/3)$. 

\begin{remark}
From Proposition \ref{tmat} it follows immediately that
$$\mathcal{I}_n (\widehat{T}^{(n)}(x;z_1,\ldots, z_n)) g^{(n+1)}(\mathbf{z},0) = g^{(n+1)}(\mathbf{z},0).$$
when $t^{\frac{1}{4}}=\exp(\pi i/3)$.
In \cite{Di-Francesco:2006aa} the authors use this equation to prove the braid recursion relation 
for $v=1$, $1=q$ $t^{\frac{1}{4}}=\exp(\pi i/3)$ and $n$ even (see Remark \ref{remTHM}). 
However, they implicitly assume that $g^{(n+1)}(\mathbf{z},0)$ is uniquely characterized as 
ground state of  $\widehat{T}^{(n+1)}(1;\mathbf{z},0)$, which is though not clear since 
there is no stochastic parameter regime when one of the rapidities is set equal to zero.
We have circumvented this
problem here by using the characterization of $g^{(n)}(\mathbf{z})$ as a twisted symmetric solution of qKZ equations.
\end{remark}

\section{Existence of solutions for generic $t^{\frac{1}{4}}$} \label{sec-CM}

In this section we construct for generic $t^{\frac{1}{4}}$ (i.e., for values $t^{\frac{1}{4}}$ in a nonempty Zariski open subset of $\mathbb{C}$) nontrivial polynomial twisted symmetric solutions to the qKZ equations for link pattern modules, leading to the proof of Theorem \ref{mainTHM} for generic $t^{\frac{1}{4}}\in\mathbb{C}^*$. As we remarked in the paragraph following Theorem \ref{mainTHM}, the generic condition on $t^{\frac{1}{4}}$ can subsequently be weakened to the condition that the loop weights are nonzero. 

A major difference between the generic case and the case that $t^{\frac{1}{4}}=\exp(\pi \mathrm{i} /3)$ is that we do not have the argument of a stochastic matrix to construct $g^{(n)}$ using the Frobenius-Perron theorem.
We instead use the Cherednik-Matsuo correspondence \cite{Stokman:2011aa}.  This is different from the approach in \cite{Kasatani:2007aa}, where Kazhdan-Lusztig bases are used.

In order to be able to apply the Cherednik-Matsuo correspondence, we first need to identify the
link pattern representations $V_n$ with principal series representations. This is done in the first subsection, for general twist parameter $v$. In the subsequent subsection, we recall the Cherednik-Matsuo correspondence and rephrase it in terms of dual $Y$-operators. In the last subsection,
we prove Theorem \ref{mainTHM} by constructing the polynomial solution of the qKZ equation
from dual non-symmetric Macdonald polynomials with specialised parameters.

For fixed $v \in \C^*$ the link patterns $\mathcal{L}_n$ form a (non-canonical) basis of $V_n$.
We can naturally identify $V_n$ with $\C[\mathcal{L}_n]$ as a vector space by shrinking the hole $\{ z \in \C  \; |\; |z| \leq 1\}$ of the annulus to $0$. 
A choice needs to be made for the winding of the defect line, unless $v=1$.

\subsection{$V_n$ as a principal series module}
In this section we take $n\geq 2$, and we fix $v\in\mathbb{C}^*$. We recall first the definition of the principal series representation $M^I(\gamma)$ of the affine Hecke algebra
$\mathcal{H}_n=\mathcal{H}_n(t^{\frac{1}{2}})$.

Let  $\epsilon_i$ ($1 \leq i \leq n$) denote the standard basis of $\mathbb{R}^n$. Set $R_0:= \{ \epsilon_i-\epsilon_{j} | 1\leq i\neq j \leq n\}$, the root system of type $A_{n-1}$. We take  $R_0^{+}:= \{ \epsilon_i-\epsilon_{j} | 1\leq i<j\leq  n \}$ the set of positive roots. The corresponding simple roots are $\alpha_i:= \epsilon_i - \epsilon_{i+1} $ ($1 \leq i <n$). We write $s_\alpha$ ($\alpha\in R_0$) for the reflection in $\alpha$. Then, the simple reflections $s_i:=s_{\alpha_i}$ ($1 \leq i<n$) correspond to the simple
neighboring transpositions $i\leftrightarrow i+1$. For $\alpha=\epsilon_i-\epsilon_j\in R_0$ we write $\mathbf{z}^\alpha=z_i/z_j$ and $Y^\alpha=Y_i/Y_j$ in $\mathbb{C}[\mathbf{z}^{\pm 1}]$
and $\mathcal{H}_n$, respectively.

For $I \subseteq \{ 1, \ldots, n-1 \}$, we write
\[
T^I=T^{I,t}:=\{ \gamma\in(\mathbb{C}^*)^n\,\,\, | \,\,\, \gamma_i/\gamma_{i+1}=t^{-1}\quad
\forall\,\, i\in I\}.
\]
For $\gamma\in T^I$ let $\chi_{\gamma}^I:=\mathcal{H}_I(t^{\frac{1}{2}})\rightarrow
\mathbb{C}$ be the one-dimensional representation of the parabolic subalgebra
$\mathcal{H}_I=\mathcal{H}_I(t^{\frac{1}{2}}):= \C \langle Y^{\pm1}_j , T_i | i \in I, j=1\ldots,n \rangle$ of
$\mathcal{H}_n(t^{\frac{1}{2}})$ satisfying $\chi_\gamma^I(Y_j)=\gamma_j$ ($1\leq j\leq n$)
and $\chi_\gamma^I(T_i)=t^{-\frac{1}{2}}$ ($i\in I$). It is well defined since $\gamma\in T^I$.
The corresponding principal series module $M^I(\gamma)$ with central character $\gamma$
is
\[
M^I(\gamma): = \mathcal{H}_n\otimes_{\mathcal{H}_I} \C_{\chi^I_{\gamma}}.
\]
Comparing with  the notations from \cite[\S 4.3]{Stokman:2011aa}: $(k,m,\zeta,H(k))$ correspond to our $(-t^{\frac{1}{2}},n,\rho,\mathcal{H}_n(t^{\frac{1}{2}}))$. The principal series module 
$M^I(\gamma)$ then corresponds to the principal series module $M^{-t^{\frac{1}{2}},-,I}(\gamma)$ from
\cite[Lem. 2.5]{Stokman:2011aa}.
 
Let $S_{n,I}=\langle s_i\,\, | \,\, i\in I\rangle\subseteq S_n$ be the standard parabolic subgroup generated by the simple neighboring transpositions $s_i$ ($i\in I$), and $S_n^I$ the minimal coset
representatives of $S_n/S_{n,I}$. For $w\in S_n$, let $T_w\in\mathcal{H}_n^0$ be the element
$T_w=T_{i_1}T_{i_2}\cdots T_{i_r}$ if $w=s_{i_1}s_{i_2}\cdots s_{i_r}$ is a reduced expression.
This is well defined since the $T_i$'s satisfy the braid relations. A linear basis of $M^I(\gamma)$
is given by $\{v_w^I(\gamma):=T_w\otimes_{\mathcal{H}_I}1_{\chi_\gamma^I}\}_{w\in S_n^I}$.

For a finite-dimensional left $\mathcal{H}_n$-module $V$ and $\xi\in (\mathbb{C}^*)^n$, we define
the subspace of vectors of weight $\xi$ by
\[
V_\xi:=\{v\in V \,\, | \,\, Y_jv=\xi_jv\quad (1\leq j\leq n)\}.
\]
The module $V$ is said to be calibrated  if $V=\bigoplus_{\xi}V_\xi$. 

For $1 \leq i < n$ set
\begin{equation}\label{Ii}
		I_i:= T_i( 1-  Y^{\alpha_i}) + (t^{-\frac{1}{2}} - t^{\frac{1}{2}})Y^{\alpha_i} \in \mathcal{H}_n.
\end{equation}
The following theorem is well known, see
 \cite[Thrm. 2.8, Cor. 2.9]{Stokman:2011aa} and references therein.

\begin{theorem}\label{intertwinerthm}
For $w\in S_n$ and $w=s_{i_1}s_{i_2}\cdots s_{i_r}$ a reduced expression,
\[
I_w:=I_{j_1}I_{j_2}\cdots I_{j_r}\in\mathcal{H}_n
\]
is well defined (independent of the choice of reduced expression).
Furthermore, for all $f(\mathbf{z})\in\mathbb{C}[\mathbf{z}^{\pm 1}]$ and $w\in S_n$ we have
\begin{equation*}
\begin{split}
I_wf(Y)&=(wf)(Y)I_w,\\
I_{w^{-1}}I_w&=e_w(Y)
\end{split}
\end{equation*}
in $\mathcal{H}_n$, with 
\[
e_w(\mathbf{z}):=\prod_{\alpha\in R_0^+\cap w^{-1}R_0^-}
\left(t^{\frac{1}{2}}-t^{-\frac{1}{2}}\mathbf{z}^\alpha\right)
\left(t^{\frac{1}{2}}-t^{-\frac{1}{2}}\mathbf{z}^{-\alpha}\right)\in\mathbb{C}[\mathbf{z}^{\pm 1}].
\]
\end{theorem}

If $V$ is a left $\mathcal{H}_n$-module,  then the previous theorem implies that $I_w(V_\xi)\subseteq V_{w\xi}$ for $w\in S_n$ and $\xi\in (\mathbb{C}^*)^n$.

It is known that $M^I(\gamma)$ is calibrated
for generic $\gamma\in T^I$ with corresponding weight decomposition
\[
M^I(\gamma)=\bigoplus_{w\in S_n^I}M^I(\gamma)_{w\gamma},\qquad
M^I(\gamma)_{w\gamma}=\mathbb{C}b_w^I(\gamma)
\]
with $b_w^I(\gamma):=I_w\otimes_{\mathcal{H}_I}1_{\chi_\gamma^I}$
(see e.g., \cite[Prop. 2.12]{Stokman:2011aa}, for the specific additional conditions on $\gamma$). 

We now view the $\mathcal{TL}_n$-module $V_n$ from the link pattern tower as an $\mathcal{H}_n$-module  through the surjective algebra map $\psi_n: \mathcal{H}_n\twoheadrightarrow
\mathcal{TL}_n$ satisfying $\psi_n(T_i)=e_i+t^{-\frac{1}{2}}$ ($1\leq i<n$) and $\psi_n(\rho)=\rho$.
The aim is to show that $V_n$ is isomorphic to $M^I(\gamma)$ for an appropriate subset
$I\subseteq\{1,\ldots,n-1\}$ and $\gamma\in T^I$ for generic $t^{\frac{1}{4}}$. As a first step, we create explicit weight vectors in $V_n$.

Write $k=\lfloor\frac{n}{2}\rfloor$ and let $J \subseteq \{ 1,\ldots,k \}$ be a subset, say $J =\{ j_1, \dots ,j_r \} , 1 \leq j_1 < \cdots < j_r\leq k$. 
Let $D^{n}_J$ be the element in $V_{n}$ shown in Figure \ref{dj}.
Note that in the definition of $D^{n}_J$, the arches $(2m-1, 2m)$ include the hole of the annulus if $m \in J$, and $(2j_{s+1}-1, 2j_{s+1})$ is positioned over  $(2j_{s}-1, 2j_{s})$.	
Furthermore,  $D^{2k+1}_J$ is obtained from $D^{2k}_J$ by inserting the defect line at $2k+1$, which is positioned over all other paths. 

\begin{figure}[t]
    \centering
    \begin{subfigure}[b]{0.45\textwidth}
       \begin{pspicture}[shift=-2.9](-3,-3)(3,2.5)
				\SpecialCoor
				\pscircle[fillstyle=solid,fillcolor=diskin,linewidth=1.2pt](0,0){2}
    				\pscircle[fillstyle=solid, fillcolor=white, linewidth=1.2pt](0,0){0.25}
				\pscurve[linecolor=line, linewidth=1.2pt](2;0)(1.8;2)(1.8;8)(2;10)
				\pscurve[linecolor=line, linewidth=1.2pt](2;50)(1.8;52)(1.8;58)(2;60)
				\pscurve[linecolor=line, linewidth=1.2pt](2;110)(1.8;112)(1.8;118)(2;120) 
				\pscurve[linecolor=line, linewidth=1.2pt](2;180)(1.8;182)(1.8;188)(2;190)
				\pscurve[linecolor=line, linewidth=1.2pt](2;220)(1.8;222)(1.8;228)(2;230) 
				\pscurve[linecolor=line, linewidth=1.2pt](2;270)(1.8;272)(1.8;278)(2;280)	
				\pscurve[linecolor=line, linewidth=1.2pt](2;310)(1.8;312)(1.8;318)(2;320)	
				\pscurve[linecolor=line,linewidth=1.2pt](2;90)(1.9;90)(1.3;135)(1.3;180)(1.3;225)(1.3;270)(1.3;315)(1.3;0)(1.3;45)(1.9;80)(2;80)
				\definecolor{diskmid}{gray}{0.8}
				\psdot[linecolor=diskin,dotsize=0.3](1.3;185)
				\psdot[linecolor=diskin,dotsize=0.3](1.3;225)
				\psdot[linecolor=diskin,dotsize=0.3](1.3;285)
				\psdot[linecolor=diskin,dotsize=0.3](1.3;305)
				\pscurve[linecolor=line,linewidth=1.2pt](2;210)(1.9;210)(1;255)(1;300)(1;345)(1;30)(1;75)(1;110)(1;155)(1.9;200)(2;200)
				\psdot[linecolor=diskin,dotsize=0.3](0.9;275)
				\psdot[linecolor=diskin,dotsize=0.3](1;312)
				\pscurve[linecolor=line,linewidth=1.2pt](2;300)(1.9;300)(0.6;345)(0.6;30)(0.6;75)(0.6;110)(0.6;155)(0.6;200)(0.6;245)(1.9;290)(2;290)
				\psarc[linecolor=blue,linestyle=dotted, linewidth=1.2pt](0,0){1.8}{20}{40}
				\psarc[linecolor=blue,linestyle=dotted, linewidth=1.2pt](0,0){1.8}{130}{170}
				\psarc[linecolor=blue,linestyle=dotted, linewidth=1.2pt](0,0){1.8}{240}{260}
				\psarc[linecolor=blue,linestyle=dotted, linewidth=1.2pt](0,0){1.8}{330}{-10}
				\rput{0}(2.2;0){\tiny$1$}
				\rput{0}(2.2;10){\tiny$2$}
				\rput{0}(2.25;70){\tiny$2j_1\!-\!1$}
				\rput{0}(2.2;90){\tiny$2j_1$}
				\rput{0}(2.5;198){\tiny$2j_2\!-\!1$}
				\rput{0}(2.3;210){\tiny$2j_2$}	
				\rput{-45}(2.5;293){\tiny$2j_r\!-\!1$}
				\rput{0}(2.3;305){\tiny$2j_r$}	
		\end{pspicture}
        \caption{$D^{2k}_J \in V_{2k}$ }
        \label{dj2k}
    \end{subfigure}
    ~ 
    \begin{subfigure}[b]{0.45 \textwidth}
         \begin{pspicture}[shift=-2.9](-3,-3)(3,2.5)
				\SpecialCoor
				\pscircle[fillstyle=solid,fillcolor=diskin, linewidth=1.2pt](0,0){2}
    				\pscircle[fillstyle=solid, fillcolor=white, linewidth=1.2pt](0,0){0.25}
				\pscurve[linecolor=line, linewidth=1.2pt](2;0)(1.8;2)(1.8;8)(2;10)
				\pscurve[linecolor=line, linewidth=1.2pt](2;50)(1.8;52)(1.8;58)(2;60)
				\pscurve[linecolor=line, linewidth=1.2pt](2;110)(1.8;112)(1.8;118)(2;120) 
				\pscurve[linecolor=line, linewidth=1.2pt](2;180)(1.8;182)(1.8;188)(2;190)
				\pscurve[linecolor=line, linewidth=1.2pt](2;220)(1.8;222)(1.8;228)(2;230) 
				\pscurve[linecolor=line, linewidth=1.2pt](2;270)(1.8;272)(1.8;278)(2;280)	
				\pscurve[linecolor=line, linewidth=1.2pt](2;310)(1.8;312)(1.8;318)(2;320)	
				\pscurve[linecolor=line,linewidth=1.2pt](2;90)(1.9;90)(1.5;135)(1.5;180)(1.5;225)(1.5;270)(1.5;315)(1.5;0)(1.5;45)(1.9;80)(2;80)
				\definecolor{diskin}{gray}{0.8}
				\psdot[linecolor=diskin,dotsize=0.3](1.5;190)
				\psdot[linecolor=diskin,dotsize=0.3](1.5;220)
				\psdot[linecolor=diskin,dotsize=0.3](1.5;285)
				\psdot[linecolor=diskin,dotsize=0.3](1.5;305)
				\psdot[linecolor=diskin,dotsize=0.3](1.5;350)
				\pscurve[linecolor=line,linewidth=1.2pt](2;210)(1.9;210)(1.2;255)(1.2;300)(1.2;345)(1.2;30)(1.2;75)(1.2;110)(1.2;155)(1.9;200)(2;200)
				\definecolor{diskmid}{gray}{0.78}
				\psdot[linecolor=diskin,dotsize=0.3](1.2;275)
				\psdot[linecolor=diskin,dotsize=0.3](1.2;312)
				\psdot[linecolor=diskin,dotsize=0.3](1.2;350)
				\pscurve[linecolor=line,linewidth=1.2pt](2;300)(1.9;300)(0.8;345)(0.8;30)(0.8;75)(0.8;110)(0.8;155)(0.8;200)(0.8;245)(1.9;290)(2;290)
				\psdot[linecolor=diskin,dotsize=0.3](0.8;342)
				\pscurve[linecolor=line,linewidth=1.2pt](2;350)(1.9;350)(1;350)(0.5;270)(0.5;225)(0.5;180)(0.5;135)(0.5;90)(0.5;45)(0.35;0)(0.25;0)	
				\psarc[linecolor=blue,linestyle=dotted, linewidth=1.2pt](0,0){1.8}{20}{40}
				\psarc[linecolor=blue,linestyle=dotted, linewidth=1.2pt](0,0){1.8}{130}{170}
				\psarc[linecolor=blue,linestyle=dotted, linewidth=1.2pt](0,0){1.8}{240}{260}
				\psarc[linecolor=blue,linestyle=dotted, linewidth=1.2pt](0,0){1.8}{330}{-20}
				\rput{0}(2.2;0){\tiny$1$}
				\rput{0}(2.2;10){\tiny$2$}
				\rput{0}(2.25;70){\tiny$2j_1\!-\!1$}
				\rput{0}(2.2;90){\tiny$2j_1$}
				\rput{0}(2.5;198){\tiny$2j_2\!-\!1$}
				\rput{0}(2.3;210){\tiny$2j_2$}	
				\rput{-45}(2.5;293){\tiny$2j_r\!-\!1$}
				\rput{0}(2.3;305){\tiny$2j_r$}	
				\rput{0}(2.6;350){\tiny$2k+1$}
		\end{pspicture}
        \caption{$D^{2k+1}_J \in V_{2k+1}$}
        \label{dj2kp1}
    \end{subfigure}
    \caption{The element $D^n_J \in V_n$ }\label{dj}
\end{figure}
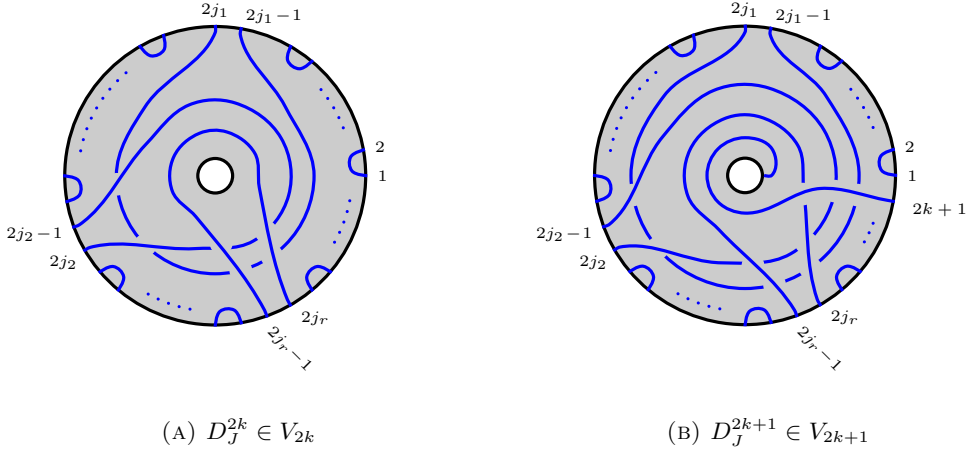

We require the skein theoretic description of $\psi_n(Y_j) \in \mathcal{TL}_n$. {}From the expression
$Y_j=T_{j-1}^{-1}\cdots T_1^{-1}\rho T_{n-1}\cdots T_j$ we obtain
\[
\psi_n(Y_j)=t^{\frac{2j-n-1}{4}}\widehat{Y}_j
\]
with $\widehat{Y}_j \in \mathcal{TL}_n\simeq \textup{End}_{\mathcal{S}}(n)$ the skein class of
$$\psset{unit=0.8cm}
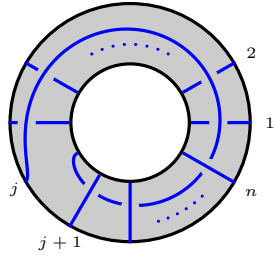
 \begin{pspicture}[shift=-2.2](-2.2,-2.2)(2.4,2)
    	\pscircle[fillstyle=solid, fillcolor=diskin,linewidth=1.2pt](0,0){2}
    	\pscircle[fillstyle=solid, fillcolor=white, linewidth=1.2pt](0,0){1}
	\degrees[12]
	\psline[linecolor=line,linewidth=1.2pt](2;0)(1;0)
	\psdot[linecolor=diskin,dotsize=0.3](1.45;0)
	\psline[linecolor=line,linewidth=1.2pt](2;1)(1;1)
	\psdot[linecolor=diskin,dotsize=0.3](1.5;1)
	\psline[linecolor=line,linewidth=1.2pt](2;5)(1;5)
	\psdot[linecolor=diskin,dotsize=0.3](1.6;5)
	\psline[linecolor=line,linewidth=1.2pt](2;6)(1;6)
	\psdot[linecolor=diskin,dotsize=0.3](1.7;6)
    \rput{7}(0;0){ \psplot[algebraic, plotpoints=400,linecolor=line,polarplot=true,linewidth=1.2pt]{0}{2 Pi mul }{(2/Pi)*ASIN(x/Pi-1) +(4*Pi-x)/(2*Pi)}}
    \psdot[linecolor=diskin,dotsize=0.3](1.3;8)
    \psdot[linecolor=diskin,dotsize=0.3](1.4;9)
    \psdot[linecolor=diskin,dotsize=0.3](1.5;11)
    \psline[linecolor=line,linewidth=1.2pt](2;8)(1;8)
    \psline[linecolor=line,linewidth=1.2pt](2;9)(1;9)
    \psline[linecolor=line,linewidth=1.2pt](2;11)(1;11)
    \psarc[linestyle=dotted, linecolor=line,linewidth=1.2pt](0,0){1.3}{2}{4}
    \psarc[linestyle=dotted, linecolor=line,linewidth=1.2pt](0,0){1.7}{9.5}{10.5}      
    \rput{0}(2.3;0){\tiny$1$}
       \rput{0}(2.3;1){\tiny$2$}     
    \rput{0}(2.2;7){\tiny$j$}
      \rput{0}(2.3;8){\tiny$j+1$}
         \rput{0}(2.3;11){\tiny$n$}
   \end{pspicture}     
$$
Set $\epsilon_n:=(-1)^n$.
\begin{lemma}\label{Qlemma}
Write $\underline{\eta}=(\eta_1,\ldots,\eta_{\lfloor n/2\rfloor})$ with
$\eta_j\in\{v^{\epsilon_n},t^{-\frac{1}{2}}v^{-\epsilon_n}\}$. 
Let $\widehat{\xi}(\underline{\eta})\in (\mathbb{C}^*)^n$ be given by
\begin{equation*}
\widehat{\xi}(\underline{\eta}):=
\begin{cases}
(\eta_1^{-1},\eta_1,\eta_2^{-1},\eta_2,\ldots,\eta_{\lfloor n/2\rfloor}^{-1},
\eta_{\lfloor n/2\rfloor})\qquad\quad \hbox{ if }\,\, n\,\, \textup{ even},\\
(\eta_1^{-1},\eta_1,\eta_2^{-1},\eta_2,\ldots,\eta_{\lfloor n/2\rfloor}^{-1},
\eta_{\lfloor n/2\rfloor},v)\qquad \hbox{ if }\,\, n\,\, \textup{ odd}
\end{cases}
\end{equation*}
and write for $J\subseteq \{1,\ldots,\lfloor n/2\rfloor\}$,
\[
c_J(\underline{\eta}):=t^{\frac{\#J}{4}}\prod_{j\in J}\eta_j^{-1}.
\]
Then, $\widehat{Y}_jQ_n(\underline{\eta})=\widehat{\xi}_j(\underline{\eta})Q_n(\underline{\eta})$ in $V_n$ for $j=1,\ldots,n$, where
\[
Q_n(\underline{\eta}):=\sum_{J\subseteq \{1,\ldots,\lfloor n/2\rfloor\}}c_J(\underline{\eta})D_J^n\in V_n.
\]
In particular we have $Q_n(\underline{\eta})\in V_{n,\xi(\underline{\eta})}$ with weight
\begin{equation*}
\xi(\underline{\eta})=
\begin{cases}
(t^{\frac{1-n}{4}}\eta_1^{-1},t^{\frac{3-n}{4}}\eta_1,
t^{\frac{5-n}{4}}\eta_2^{-1},\ldots,t^{\frac{n-3}{4}}\eta_{\lfloor n/2\rfloor}^{-1},
t^{\frac{n-1}{4}}\eta_{\lfloor n/2\rfloor})\qquad\qquad\quad \hbox{ if }\,\, n\,\, \textup{ even},\\
(t^{\frac{1-n}{4}}\eta_1^{-1},t^{\frac{3-n}{4}}\eta_1,t^{\frac{5-n}{4}}\eta_2^{-1},\ldots,
t^{\frac{n-5}{4}}\eta_{\lfloor n/2\rfloor}^{-1},
t^{\frac{n-3}{4}}\eta_{\lfloor n/2\rfloor},t^{\frac{n-1}{4}}v)\qquad \hbox{ if }\,\, n\,\, \textup{ odd}.
\end{cases}
\end{equation*}
\end{lemma}
\begin{proof}
It suffices to show that $\widehat{Y}_jQ_n(\underline{\eta})=
\widehat{\xi}_j(\underline{\eta})Q_n(\underline{\eta})$. Write $k:=\lfloor n/2\rfloor$.
There are three cases to consider, $j=2i,2i-1$ (for $1\leq i \leq k$) and, if $n$ is odd, $j=n=2k+1$. 
We consider first $j=2i$. Note that by the definition of $D^{n}_J$ an arch is placed on top of the previous arch if they both encircle the hole of the annulus.
For $\widehat{Y}_{2i}D_j^n$, the path connected to $2i$ that is wound around the diagram passes over all paths connected to $l < 2i$ and under all paths connected to $l >2i$.
Due to these properties, the action of $\widehat{Y}_{2i}$ on $D_J^n$ will only affect the arch $(2i-1,2i)$ and leave the others unchanged.

Consider now $\widehat{Y}_{2i}Q_{n}(\underline{\eta})$ and combine the terms $J$ and $J\cup \{i\}$ for subsets $J$
not containing $i$,
\begin{align*}
	\widehat{Y}_{2i}  Q_{n}(\underline{\eta}) &= \sum _{J \subseteq  \{ 1,\ldots,k \}} c_J(\underline{\eta})\widehat{Y}_{2i}D_J^n\\
				  		   &=   \sum _{J \subseteq  \{ 1,\ldots,k \} \backslash \{i\}} c_J(\underline{\eta})
						   \widehat{Y}_{2i} \left( D^{n}_{J} +   t^{\frac{1}{4}}\eta^{-1}_{i}  D^{n}_{J\cup\{i\}} \right).
\end{align*}

Focusing on the action of $\widehat{Y}_{2i}$ on the terms in the bracket, we claim that
	\begin{align*}
		\widehat{Y}_{2i}\left( D^{n}_{J} + t^{\frac{1}{4}}\eta^{-1}_{i}  D^{n}_{J\cup\{i\}} \right) = (v^{\epsilon_n} + t^{-\frac{1}{2}} v^{-\epsilon_n} -t^{-\frac{1}{2}}\eta^{-1}_{i}) D^{n}_{J} + t^{\frac{1}{4}} D^{n}_{J\cup\{i\}} 
	\end{align*} 
for all $J \subseteq \{1, \dots ,k\} \backslash \{i\}$. Since $\eta_i$ satisfies
\[
\eta_i=v^{\epsilon_n}+t^{-\frac{1}{2}}v^{-\epsilon_n}-t^{-\frac{1}{2}}\eta_i^{-1}
\]
it then follows that $\widehat{Y}_{2i}Q_n(\underline{\eta})=\eta_iQ_n(\underline{\eta})$. To prove the claim we show
\begin{equation*}
\begin{split}
\widehat{Y}_{2i}D_J^n&=(v^{\epsilon_n}+t^{-\frac{1}{2}}v^{-\epsilon_n})D_J^n+
t^{\frac{1}{4}}D_{J\cup\{i\}}^n,\\
\widehat{Y}_{2i}D_{J\cup\{i\}}^n&=-t^{-\frac{3}{4}}D_J^n
\end{split}
\end{equation*}
for $J\subseteq\{1,\ldots,\lfloor n/2\rfloor\}{\setminus}\{i\}$.
These equalities follow from the following diagrammatic calculations, in which we omit all paths that are not involved in the computation. 
The first diagrammatic computation is for  $\widehat{Y}_{2i}D^{2k}_{J}$ in $V_{2k}$, the second for $\widehat{Y}_{2i}D^{2k+1}_{J}$ in $V_{2k+1}$ (note that the defect line creates a subtle difference) and the third for  $\widehat{Y}_{2i} D^{n}_{J\cup\{i\}}$ in $V_n$ (in this case, the defect line does not affect the calculation):
{\psset{unit=0.8cm}
\begin{align*}
\widehat{Y}_{2i} 
\begin{pspicture}[shift=-1.4](-1.7,-1.5)(1.25,1.5)
	\SpecialCoor
	\pscircle[fillstyle=solid, fillcolor=diskin,linewidth=1.2pt](0,0){1}
	\pscircle[fillstyle=solid, fillcolor=white, linewidth=1.2pt](0,0){0.25}
	\degrees[10]
	\pscurve[linecolor=line,linewidth=1.2pt](1;6)(0.9;6)(0.6;6.4)(0.6;6.6)(0.9;7)(1;7)
	\rput{0}(1.3;6){\tiny $2i\!-\!\!1$}
	\rput{0}(1.2;7){\tiny $2i$}
\end{pspicture} 
&=
  \begin{pspicture}[shift=-0.9](-1.5,-1)(1.5,1)
	\SpecialCoor
	\pscircle[fillstyle=solid, fillcolor=diskin,linewidth=1.2pt](0,0){1}
	\pscircle[fillstyle=solid, fillcolor=white, linewidth=1.2pt](0,0){0.25}
	\degrees[10]
	\pscurve[linecolor=line,linewidth=1.2pt](1;6)(0.8;6)(0.35;6.5)(0.5;7)(0.6;7)
	\psdot[linecolor=diskin,dotsize=0.25](0.7;6)
	\psecurve[linecolor=line,linewidth=1.2pt](0.5;7)(0.6;7)(0.65;7)(0.7;8)(0.7;9)(0.7;0)(0.7;1)(0.7;2)(0.7;3)(0.7;4)(0.7;5)(0.7;6)(0.7;6.5)(0.9;7)(1;7)(1;7)
	\rput{0}(1.3;6){\tiny $2i\!-\!\!1$}
	\rput{0}(1.2;7){\tiny $2i$}
\end{pspicture} 
= t^{\frac{1}{4}}
 \begin{pspicture}[shift=-0.9](-1.25,-1)(1.25,1)
	\SpecialCoor
	\pscircle[fillstyle=solid, fillcolor=diskin, linewidth=1.2pt](0,0){1}
	\pscircle[fillstyle=solid, fillcolor=white, linewidth=1.2pt](0,0){0.25}
	\degrees[10]
	\pscurve[linecolor=line,linewidth=1.2pt](1;7)(0.9;7)(0.6;8)(0.6;9)(0.6;0)(0.6;1)(0.6;2)(0.6;3)(0.6;4)(0.6;5)(0.9;6)(1;6)
	\rput{0}(1.3;6){\tiny $2i\!-\!\!1$}
	\rput{0}(1.2;7){\tiny $2i$}
\end{pspicture} 
+ ( v + t^{-\frac{1}{2}}v^{-1}) 
	\begin{pspicture}[shift=-1.4](-1.7,-1.5)(1.25,1.5)
	\SpecialCoor
	\pscircle[fillstyle=solid, fillcolor=diskin,linewidth=1.2pt](0,0){1}
	\pscircle[fillstyle=solid, fillcolor=white, linewidth=1.2pt](0,0){0.25}
	\degrees[10]
	\pscurve[linecolor=line,linewidth=1.2pt](1;6)(0.9;6)(0.6;6.4)(0.6;6.6)(0.9;7)(1;7)
	\rput{0}(1.3;6){\tiny $2i\!-\!\!1$}
	\rput{0}(1.2;7){\tiny $2i$}
	\end{pspicture} ;
\\
\widehat{Y}_{2i}
\begin{pspicture}[shift=-1.4](-1.7,-1.5)(1.25,1.5)
	\SpecialCoor
	\pscircle[fillstyle=solid, fillcolor=diskin,linewidth=1.2pt](0,0){1}
	\pscircle[fillstyle=solid, fillcolor=white, linewidth=1.2pt](0,0){0.25}
	\degrees[10]
	\pscurve[linecolor=line,linewidth=1.2pt](1;6)(0.9;6)(0.6;6.4)(0.6;6.6)(0.9;7)(1;7)
	\pscurve[linecolor=line,linewidth=1.2pt](0.25;0)(0.35;0)(0.4;1)(0.4;2)(0.4;3)(0.4;4)(0.4;5)(0.4;6)(0.4;7)(0.8;9)(1;9)
	\rput{0}(1.3;6){\tiny $2i\!-\!\!1$}
	\rput{0}(1.2;7){\tiny $2i$}
	\rput{0}(1.3;9){\tiny $2k\!\!+\!\!1$}
\end{pspicture} 
&=
  \begin{pspicture}[shift=-0.9](-1.5,-1)(1.5,1)
	\SpecialCoor
	\pscircle[fillstyle=solid, fillcolor=diskin,linewidth=1.2pt](0,0){1}
	\pscircle[fillstyle=solid, fillcolor=white, linewidth=1.2pt](0,0){0.25}
	\degrees[10]
	\psecurve[linecolor=line,linewidth=1.2pt](1;6)(1;6)(0.9;6)(0.7;6.5)(0.65;7)(0.7;8)
	\psdot[linecolor=diskin,dotsize=0.25](0.7;6.5)
	\psecurve[linecolor=line,linewidth=1.2pt](0.6;7)(0.65;7)(0.7;8)(0.7;9)(0.7;0)(0.7;1)(0.7;2)(0.7;3)(0.7;4)(0.7;5)(0.7;6)(0.7;6.5)(0.9;7)(1;7)(1;7)
	\psdot[linecolor=diskin,dotsize=0.25](0.7;8.9)
	\pscurve[linecolor=line,linewidth=1.2pt](0.25;0)(0.35;0)(0.4;1)(0.4;2)(0.4;3)(0.4;4)(0.4;5)(0.4;6)(0.4;7)(0.8;9)(1;9)
	\rput{0}(1.3;6){\tiny $2i\!-\!\!1$}
	\rput{0}(1.2;7){\tiny $2i$}
	\rput{0}(1.3;9){\tiny $2k\!\!+\!\!1$}
\end{pspicture} 
= t^{\frac{1}{4}}
 \begin{pspicture}[shift=-0.9](-1.25,-1)(1.25,1)
	\SpecialCoor
	\pscircle[fillstyle=solid, fillcolor=diskin, linewidth=1.2pt](0,0){1}
	\pscircle[fillstyle=solid, fillcolor=white, linewidth=1.2pt](0,0){0.25}
	\degrees[10]
	\pscurve[linecolor=line,linewidth=1.2pt](1;7)(0.9;7)(0.6;8)(0.6;9)(0.6;0)(0.6;1)(0.6;2)(0.6;3)(0.6;4)(0.6;5)(0.9;6)(1;6)
		\psdot[linecolor=diskin,dotsize=0.25](0.65;8.75)
	\pscurve[linecolor=line,linewidth=1.2pt](0.25;0)(0.35;0)(0.4;1)(0.4;2)(0.4;3)(0.4;4)(0.4;5)(0.4;6)(0.4;7)(0.8;9)(1;9)
	\rput{0}(1.3;6){\tiny $2i\!-\!\!1$}
	\rput{0}(1.2;7){\tiny $2i$}
		\rput{0}(1.3;9){\tiny $2k\!\!+\!\!1$}
\end{pspicture} 
+  t^{-\frac{1}{4}}
	\begin{pspicture}[shift=-1.4](-1.7,-1.5)(1.25,1.5)
	\SpecialCoor
	\pscircle[fillstyle=solid, fillcolor=diskin,linewidth=1.2pt](0,0){1}
	\pscircle[fillstyle=solid, fillcolor=white, linewidth=1.2pt](0,0){0.25}
	\degrees[10]
	\pscurve[linecolor=line,linewidth=1.2pt](1;6)(0.9;6)(0.8;6.4)(0.8;6.6)(0.9;7)(1;7)
	\pscircle[ linecolor=line,linewidth=1.2pt](0,0){0.65}
		\psdot[linecolor=diskin,dotsize=0.25](0.68;8.8)
	\pscurve[linecolor=line,linewidth=1.2pt](0.25;0)(0.35;0)(0.4;1)(0.4;2)(0.4;3)(0.4;4)(0.4;5)(0.4;6)(0.4;7)(0.8;9)(1;9)
	\rput{0}(1.3;6){\tiny $2i\!-\!\!1$}
	\rput{0}(1.2;7){\tiny $2i$}
		\rput{0}(1.3;9){\tiny $2k\!\!+\!\!1$}
	\end{pspicture} 
	\\
	&= 
	 t^{\frac{1}{4}}
 \begin{pspicture}[shift=-0.9](-1.25,-1)(1.25,1)
	\SpecialCoor
	\pscircle[fillstyle=solid, fillcolor=diskin, linewidth=1.2pt](0,0){1}
	\pscircle[fillstyle=solid, fillcolor=white, linewidth=1.2pt](0,0){0.25}
	\degrees[10]
	\pscurve[linecolor=line,linewidth=1.2pt](1;7)(0.9;7)(0.6;8)(0.6;9)(0.6;0)(0.6;1)(0.6;2)(0.6;3)(0.6;4)(0.6;5)(0.9;6)(1;6)
		\psdot[linecolor=diskin,dotsize=0.25](0.65;8.75)
	\pscurve[linecolor=line,linewidth=1.2pt](0.25;0)(0.35;0)(0.4;1)(0.4;2)(0.4;3)(0.4;4)(0.4;5)(0.4;6)(0.4;7)(0.8;9)(1;9)
	\rput{0}(1.3;6){\tiny $2i\!-\!\!1$}
	\rput{0}(1.2;7){\tiny $2i$}
		\rput{0}(1.3;9){\tiny $2k\!\!+\!\!1$}
\end{pspicture} 
+  
	\begin{pspicture}[shift=-1.4](-1.7,-1.5)(1.25,1.5)
	\SpecialCoor
	\pscircle[fillstyle=solid, fillcolor=diskin,linewidth=1.2pt](0,0){1}
	\pscircle[fillstyle=solid, fillcolor=white, linewidth=1.2pt](0,0){0.25}
	\degrees[10]
	\pscurve[linecolor=line,linewidth=1.2pt](1;6)(0.9;6)(0.8;6.4)(0.8;6.6)(0.9;7)(1;7)
		\psdot[linecolor=diskin,dotsize=0.25](0.68;8.8)
	\pscurve[linecolor=line,linewidth=1.2pt](0.25;0)(0.35;0)(0.8;9)(1;9)
	\rput{0}(1.3;6){\tiny $2i\!-\!\!1$}
	\rput{0}(1.2;7){\tiny $2i$}
		\rput{0}(1.3;9){\tiny $2k\!\!+\!\!1$}
	\end{pspicture} 
	+
	t^{-\frac{1}{2}}
	\begin{pspicture}[shift=-1.4](-1.7,-1.5)(1.25,1.5)
	\SpecialCoor
	\pscircle[fillstyle=solid, fillcolor=diskin,linewidth=1.2pt](0,0){1}
	\pscircle[fillstyle=solid, fillcolor=white, linewidth=1.2pt](0,0){0.25}
	\degrees[10]
	\pscurve[linecolor=line,linewidth=1.2pt](1;6)(0.9;6)(0.8;6.4)(0.8;6.6)(0.9;7)(1;7)
	\pscurve[linecolor=line,linewidth=1.2pt](0.25;0)(0.35;0)(0.4;1)(0.4;2)(0.4;3)(0.4;4)(0.4;5)(0.4;6)(0.4;7)(0.5;9)(0.6;1)(0.6;2)(0.6;3)(0.6;4)(0.6;5)(0.6;6)(0.6;7)(0.8;9)(1;9)
	\rput{0}(1.3;6){\tiny $2i\!-\!\!1$}
	\rput{0}(1.2;7){\tiny $2i$}
		\rput{0}(1.3;9){\tiny $2k\!\!+\!\!1$}
	\end{pspicture} 
	\\
	&= 
	 t^{\frac{1}{4}}
 \begin{pspicture}[shift=-0.9](-1.25,-1)(1.25,1)
	\SpecialCoor
	\pscircle[fillstyle=solid, fillcolor=diskin, linewidth=1.2pt](0,0){1}
	\pscircle[fillstyle=solid, fillcolor=white, linewidth=1.2pt](0,0){0.25}
	\degrees[10]
	\pscurve[linecolor=line,linewidth=1.2pt](1;7)(0.9;7)(0.6;8)(0.6;9)(0.6;0)(0.6;1)(0.6;2)(0.6;3)(0.6;4)(0.6;5)(0.9;6)(1;6)
		\psdot[linecolor=diskin,dotsize=0.25](0.65;8.75)
	\pscurve[linecolor=line,linewidth=1.2pt](0.25;0)(0.35;0)(0.4;1)(0.4;2)(0.4;3)(0.4;4)(0.4;5)(0.4;6)(0.4;7)(0.8;9)(1;9)
	\rput{0}(1.3;6){\tiny $2i\!-\!\!1$}
	\rput{0}(1.2;7){\tiny $2i$}
		\rput{0}(1.3;9){\tiny $2k\!\!+\!\!1$}
\end{pspicture} 
+  
 ( v^{-1} + t^{-\frac{1}{2}}v) 
	\begin{pspicture}[shift=-1.4](-1.7,-1.5)(1.25,1.5)
	\SpecialCoor
	\pscircle[fillstyle=solid, fillcolor=diskin,linewidth=1.2pt](0,0){1}
	\pscircle[fillstyle=solid, fillcolor=white, linewidth=1.2pt](0,0){0.25}
	\degrees[10]
	\pscurve[linecolor=line,linewidth=1.2pt](1;6)(0.9;6)(0.6;6.4)(0.6;6.6)(0.9;7)(1;7)
	\pscurve[linecolor=line,linewidth=1.2pt](0.25;0)(0.35;0)(0.4;1)(0.4;2)(0.4;3)(0.4;4)(0.4;5)(0.4;6)(0.4;7)(0.8;9)(1;9)
	\rput{0}(1.3;6){\tiny $2i\!-\!\!1$}
	\rput{0}(1.2;7){\tiny $2i$}
		\rput{0}(1.3;9){\tiny $2k\!\!+\!\!1$}
	\end{pspicture} ;
	\\
\widehat{Y}_{2i}  
\begin{pspicture}[shift=-1.4](-1.7,-1.5)(1.25,1.5)
	\SpecialCoor
	\pscircle[fillstyle=solid, fillcolor=diskin, linewidth=1.2pt](0,0){1}
	\pscircle[fillstyle=solid, fillcolor=white, linewidth=1.2pt](0,0){0.25}
	\degrees[10]
	\pscurve[linecolor=line,linewidth=1.2pt](1;7)(0.9;7)(0.6;8)(0.6;9)(0.6;0)(0.6;1)(0.6;2)(0.6;3)(0.6;4)(0.6;5)(0.9;6)(1;6)
	\rput{0}(1.3;6){\tiny $2i\!-\!\!1$}
	\rput{0}(1.2;7){\tiny $2i$}
\end{pspicture}
&=
 \begin{pspicture}[shift=-0.9](-1.5,-1)(1.5,1)
	\SpecialCoor
	\pscircle[fillstyle=solid, fillcolor=diskin,linewidth=1.2pt](0,0){1}
	\pscircle[fillstyle=solid, fillcolor=white, linewidth=1.2pt](0,0){0.25}
	\degrees[10]
	\pscurve[linecolor=line,linewidth=1.2pt](1;6)(0.8;6)(0.4;5)(0.4;4)(0.4;3)(0.4;2)(0.4;1)(0.4;0)(0.4;9)(0.4;8)(0.5;7)(0.6;7)
	\psdot[linecolor=diskin,dotsize=0.25](0.7;6)
	\psecurve[linecolor=line,linewidth=1.2pt](0.5;7)(0.6;7)(0.65;7)(0.7;8)(0.7;9)(0.7;0)(0.7;1)(0.7;2)(0.7;3)(0.7;4)(0.7;5)(0.7;6)(0.7;6.5)(0.9;7)(1;7)(1;7)
	\rput{0}(1.3;6){\tiny $2i\!-\!\!1$}
	\rput{0}(1.2;7){\tiny $2i$}
\end{pspicture}   = -t^{-\frac{3}{4}} 
  \begin{pspicture}[shift=-1.4](-1.25,-1.5)(1.25,1.5)
	\SpecialCoor
	\pscircle[fillstyle=solid, fillcolor=diskin, linewidth=1.2pt](0,0){1}
	\pscircle[fillstyle=solid, fillcolor=white, linewidth=1.2pt](0,0){0.25}
	\degrees[10]
	\pscurve[linecolor=line,linewidth=1.2pt](1;6)(0.9;6)(0.6;6.4)(0.6;6.6)(0.9;7)(1;7)
	\rput{0}(1.3;6){\tiny $2i\!-\!\!1$}
	\rput{0}(1.2;7){\tiny $2i$}
\end{pspicture} .
  \end{align*}
  }

The check that $\widehat{Y}_{2i-1}Q_n(\underline{\eta})=\eta_i^{-1}Q_n(\underline{\eta})$ is analogous.

The proof that $\widehat{Y}_{2k+1}Q_{2k+1}(\underline{\eta})=vQ_{2k+1}(\underline{\eta})$ with $n=2k+1$ odd is simpler.
All that the operator $\widehat{Y}_{2k+1}$ does is wind the defect line a full turn around the hole
of the annulus.
The operator keeps the defect line above all other curves. This full turn in $V_{2k+1}$ can then be removed by the multiplicative factor $v$.
\end{proof}
{}From now on we choose
\[
\underline{\eta}=(v^{\epsilon_n},\ldots,v^{\epsilon_n})
\]
and we write $Q_n$, $\xi$ and $\widehat{\xi}$ for the corresponding $Q_n(\underline{\eta})$,
$\xi(\underline{\eta})$ and $\widehat{\xi}(\underline{\eta})$. Concretely,
\[
Q_n=\sum_{J\subseteq\{1,\ldots,\lfloor n/2\rfloor\}}t^{\frac{\#J}{4}}v^{-\epsilon_n\#J}D_J^n\in V_{n,\xi},
\]
and
\begin{equation}\label{spectrum}
\xi=
\begin{cases}
(t^{\frac{1-n}{4}}v^{-1},t^{\frac{3-n}{4}}v,
t^{\frac{5-n}{4}}v^{-1},\ldots,t^{\frac{n-3}{4}}v^{-1},
t^{\frac{n-1}{4}}v)\qquad\qquad\quad \hbox{ if }\,\, n\, \textup{ even},\\
(t^{\frac{1-n}{4}}v,t^{\frac{3-n}{4}}v^{-1},t^{\frac{5-n}{4}}v,\ldots,
t^{\frac{n-5}{4}}v,
t^{\frac{n-3}{4}}v^{-1},t^{\frac{n-1}{4}}v)\qquad\quad\,\,\hbox{ if }\,\, n\,\, \textup{ odd}.
\end{cases}
\end{equation}
\begin{lemma}\label{Qnotzerolemma}
$Q_n \neq 0$ for generic $t^{\frac{1}{4}}$.
\end{lemma}
\begin{proof}
Consider first $n=2k$ even. Let $Z_{2k}\in\textup{Hom}_{\mathcal{S}}(2k,0)$ be the 
skein class of the $(2k,0)$-link diagram with little arches connecting $2i-1$ and $2i$ for
$1\leq i\leq k$. Composing on the left with $Z_{2k}$ defines a linear map 
$\textup{Hom}_{\mathcal{S}}(0,2k)\rightarrow \textup{End}_{\mathcal{S}}(0)$ that descends to 
a well-defined linear map $Z_{2k}: V_{2k}\rightarrow V_0\simeq \mathbb{C}$. Then
\[
Z_{2k}(Q_{2k})=\sum_{J\subseteq \{1,\ldots,k\}}
t^{\#J/4}v^{-\#J}\bigl(vt^{\frac{1}{4}}+v^{-1}t^{-\frac{1}{4}}\bigr)^{\#J}
\bigl(-t^{\frac{1}{2}}-t^{-\frac{1}{2}}\bigr)^{k-\#J},
\]
which is a nonzero Laurent polynomial in $t^{\frac{1}{4}}$ (look at its highest order term).

For $n=2k+1$ odd, we apply a similar argument, now using the element
$Z_{2k+1}\in\textup{Hom}_{\mathcal{S}}(2k+1,1)$ which is the skein class of the $(2k+1,1)$-link
diagram with little arches connecting the inner boundary points $2i-1$ and $2i$ ($1\leq i\leq k$) and with a defect line connecting the inner boundary point at $2k+1$ to the outer boundary point  at $1$.
Then, the resulting linear map $Z_{2k+1}: V_{2k+1}\rightarrow V_1\simeq\mathbb{C}$
maps $Q_{2k+1}$ to
\[
v^\kappa \sum_{J\subseteq\{1,\ldots,k\}}t^{\#J/4}v^{\#J}\bigl(v^{-1}t^{\frac{1}{4}}+vt^{-\frac{1}{4}}\bigr)^{\#J}
\bigl(-t^{\frac{1}{2}}-t^{-\frac{1}{2}}\bigr)^{k-\#J}
\]
for some $\kappa\in\mathbb{Z}$, which again is a nonzero Laurent polynomial in $t^{\frac{1}{4}}$
(the factor for the removal of a closed loop around the hole with a defect line running over it is
$(v^{-1}t^{\frac{1}{4}}+vt^{-\frac{1}{4}})$, as shown in the proof of the previous lemma).
\end{proof}
To establish an identification $V_n\simeq M^I(\gamma)$ as $\mathcal{H}_n$-modules,
we will use $Q_n$ and the intertwiners $I_w$ ($w\in S_n$) to construct the corresponding cyclic vector in $V_n$. But first we determine what the subset $I\subseteq\{1,\ldots,n-1\}$ should be.

Set  
\[
I^{(n)}:=\{1,\ldots,\lceil n/2\rceil-1,\lceil n/2\rceil+1,\ldots,n-1\}.
\]
The associated parabolic subgroup $S_{n,I^{(n)}}$ of $S_n$ is isomorphic to 
$S_k\times S_k$ if $n=2k$ even, and $S_{k}\times S_{k-1}$ if $n=2k-1$ is odd.

\begin{lemma}\label{Dimensionlemma}
$\text{\emph{Dim}}(V_{n}) = \#(S_{n}/S_{n,I^{(n)}}) = \# S_n^{I^{(n)}}$.
\end{lemma}
\begin{proof}
For $n=2k$ even, $\mathcal{L}_{2k}$ is in bijective correspondence with the set of binary words of length $2k$ with letters $\alpha,\beta$ of length $2k$ such that $k$ letters  are $\alpha$. 
The bijection is as follows. Orient the outer boundary of the punctured disc anticlockwise. 
Given $L\in\mathcal{L}_{2k}$, orient the arcs in $L$ in such a way that the closed oriented loop obtained by adding a piece of the oriented outer boundary of the punctured disc, is enclosing the puncture. Then, the word of length $2k$ in the letters $\{\alpha,\beta\}$ is obtained by putting $\alpha$ as the $i$th letter if the orientation of the arc at $i$ is away from $i$, and $\beta$ if it is towards $i$.

In the odd case $n=2k-1$ we create a bijective correspondence of $\mathcal{L}_{2k-1}$ with
the set of binary words of length $2k-1$ with letters $\alpha,\beta$ such that $k$ letters are $\alpha$
by a similar procedure, with the only addition that $\alpha$ is assigned to the outer boundary point that is connected to the puncture.

Clearly the cardinality of the set of such binary words is equal to
$\#(S_n/S_{n,I^{(n)}})$.
\end{proof}
\begin{remark}\label{mincoset}
The minimal coset representatives $S_n^{I^{(n)}}$ are described as the set of permutations
$\sigma\in S_n$ such that $\ell(\sigma s_i)=\ell(\sigma)+1$ for all $i\in I^{(n)}$, where $\ell$
is the length function of $S_n$. It follows that $S_n^{I^{(n)}}$ is the set of permutations $\sigma\in S_n$ such that
\[
\sigma(1)<\sigma(2)<\cdots<\sigma(\lceil n/2\rceil),
\qquad
\sigma(\lceil n/2\rceil+1)<\sigma(\lceil n/2\rceil+2)<\cdots<\sigma(n).
\]
\end{remark}

We define $w_n\in S_n$ as follows,
\begin{equation}\label{w2k}
\begin{split}
w_{2k}&=\begin{pmatrix}
		1 & 2 & 3 & 4 & \cdots & 2k-1 & 2k \\
		1 & k+1 & 2 & k+2 & \cdots & k & 2k
		\end{pmatrix},\\
w_{2k-1}&=\begin{pmatrix}
		1 & 2 & 3 & 4 & \cdots & 2k-2 & 2k-1 \\
		1 & k+1 & 2 & k+2 & \cdots & 2k-1 & k
		\end{pmatrix}		
\end{split}
\end{equation}
Note that $w_{2k}=\iota_{2k-1}(w_{2k-1})$ with $\iota_n: S_{n-1}\hookrightarrow S_n$ the natural
group embedding extending $\sigma\in S_{n-1}$ to a permutation of $\{1,\ldots,n\}$
by $\sigma(n)=n$. Note that
\[
w_{2k-1}^{-1}=\begin{pmatrix}
		1 & 2 & 3 & \cdots & k      & k+1 & k+2   &\cdots & 2k-1\\
		1 & 3 & 5 & \cdots  &2k-1 &   2   & 4 & \cdots & 2k-2
		\end{pmatrix}.
\]
It follows that $w_n^{-1}\in S_n^{I^{(n)}}$, cf. Remark \ref{mincoset}. We now 
define $\gamma=\gamma^{(n)}\in (\mathbb{C}^*)^n$ by 
\[
\gamma:=w_n\xi\in (\mathbb{C}^*)^n
\]
with $\xi$ the weight of $V_n$ as given by \eqref{spectrum}. Concretely we have
\begin{equation}\label{highestweight}
\gamma=
\begin{cases}
(t^{\frac{1-n}{4}}v^{-1},t^{\frac{5-n}{4}}v^{-1},\ldots,t^{\frac{n-3}{4}}v^{-1},
t^{\frac{3-n}{4}}v, t^{\frac{7-n}{4}}v,\ldots,t^{\frac{n-1}{4}}v)
\qquad\quad\,\, &\hbox{ if }\,\, n\, \textup{ even},\\
(t^{\frac{1-n}{4}}v, t^{\frac{5-n}{4}}v,\ldots,t^{\frac{n-1}{4}}v, t^{\frac{3-n}{4}}v^{-1}, t^{\frac{7-n}{4}}v^{-1},
\ldots, t^{\frac{n-3}{4}}v^{-1})\qquad\quad\,\,&\hbox{ if }\,\, n\,\, \textup{ odd}.
\end{cases}
\end{equation}
Note that $\gamma\in T^{I^{(n)}}$, and hence we have the associated principal series module
$M^{I^{(n)}}(\gamma)$.
\begin{theorem}\label{principalTHM}
For generic $t^{\frac{1}{4}}$, we have $M^{I^{(n)}}(\gamma)\simeq V_n$ as left $\mathcal{H}_n$-modules, with the isomorphism $M^{I^{(n)}}(\gamma)\overset{\sim}{\longrightarrow}
V_n$ mapping $v_e^{I^{(n)}}(\gamma)\in M^{I^{(n)}}(\gamma)_\gamma$ to
$I_{w_n}Q_n\in V_{n,\gamma}$, where $e$ is the unit element of the symmetric group $S_n$.
\end{theorem}
\begin{proof}
We have $I_{w_n}Q_n\in V_{n,\gamma}$ by Lemma \ref{Qlemma} and Theorem \ref{intertwinerthm}.
Furthermore, by Theorem \ref{intertwinerthm} again,
\[
I_{w_n^{-1}}I_{w_n}Q_n=e_{w_n}(\xi)Q_n
\]
and 
\[
e_{w_n}(\xi)=\prod_{\alpha\in R_0^+\cap w_n^{-1}R_0^-}
\left(t^{\frac{1}{2}}-t^{-\frac{1}{2}}\xi^\alpha\right)\left(t^{\frac{1}{2}}-t^{-\frac{1}{2}}\xi^{-\alpha}\right)\not=0
\]
for generic $t^{\frac{1}{4}}$, since $R_0^+\cap w_n^{-1}R_0^-$ consists of the 
roots $\epsilon_{2l}-\epsilon_{2m-1}$ ($l<m$). Hence, $I_{w_n}Q_n\not=0$ for generic $t^{\frac{1}{4}}$.
Consider now the vectors
\[
u_w:=I_wI_{w_n}Q_n\in V_{n,w\gamma},\qquad w\in S_n^{I^{(n)}}.
\]
Then, for $w\in S_n^{I^{(n)}}$ we have
\[
I_{w^{-1}}u_w=e_w(\gamma)I_{w_n}Q_n
\]
by Theorem \ref{intertwinerthm}. Furthermore, $e_w(\gamma)\not=0$ for generic $t^{\frac{1}{4}}$ since
for $w\in S_n^{I^{(n)}}$ we have
\[
R_0^+\cap w^{-1}R_0^-\subseteq \{\epsilon_l-\epsilon_m\,\,\, | \,\,\, 1\leq l\leq \lceil n/2\rceil\,\,\,
\& \,\,\, \lceil n/2\rceil+1\leq m<n\}
\]
for $w\in S_n^{I^{(n)}}$, in view of Remark \ref{mincoset}. It follows that 
$0\not=u_w\in V_{n,w\gamma}$ for all $w\in S_n^{I^{(n)}}$. Hence, by
Lemma \ref{Dimensionlemma}, for generic $t^{\frac{1}{4}}$,
\[
V_n=\bigoplus_{w\in S_n^{I^{(n)}}}V_{n,w\gamma}
\]
with $V_{n,w\gamma}=\mathbb{C}u_w$ for all $w\in S_n^{I^{(n)}}$, since 
the $w\gamma$'s ($w\in S_n^{I^{(n)}}$) are pairwise different for generic $t^{\frac{1}{4}}$.
It remains to show that
$T_iu_e=t^{-\frac{1}{2}}u_e$ for $i\in I^{(n)}$ and generic $t^{\frac{1}{4}}$. Fix $i\in I^{(n)}$. Then
$I_iu_e\in V_{n,s_i\gamma}=\{0\}$ 
since 
\[
s_i\gamma\not\in\{w\gamma\,\, | \,\, w\in S_n^{I^{(n)}}\}
\]
for generic $t^{\frac{1}{4}}$. By the explicit expression of the intertwiner $I_i$ (see \eqref{Ii}),
we then obtain 
\begin{equation*}
\begin{split}
0=I_iu_e&=(1-\gamma^{\alpha_i})T_iu_e+(t^{-\frac{1}{2}}-t^{\frac{1}{2}})\gamma^{\alpha_i}u_e\\
&=(1-t^{-1})T_iu_e+(t^{-\frac{1}{2}}-t^{\frac{1}{2}})t^{-1}u_e\\
&=(1-t^{-1})(T_i-t^{-\frac{1}{2}})u_e.
\end{split}
\end{equation*}
Hence, $T_iu_e=t^{-\frac{1}{2}}u_e$, as desired.
\end{proof}

\subsection{The Cherednik-Matsuo correspondence} 
Now that we have identified the link pattern modules $V_n$ with principal
series representations, we can apply the Cherednik-Matsuo correspondence to analyse the
existence of polynomial solutions of the associated qKZ equations. 

The Cherednik-Matsuo correspondence gives a bijective correspondence between meromorphic twisted symmetric solution to qKZ equations associated with a principal series module and suitable classes of meromorphic common eigenfunctions for the action of the $Y$-operators under the basic representation \cite{Stokman:2011aa,Kasatani:2007ab}.

The version of the Cherednik-Matsuo correspondence we need is as follows. If $I\subseteq\{1,\ldots,n-1\}$
and $w\in S_n$, then we write $\underline{w}\in S_{n,I}$ and $\overline{w}\in S_n^I$ for the unique
elements such that $w=\overline{w}\underline{w}$. Let $w_0\in S_n$ be the longest Weyl group
element, mapping $j$ to $n+1-j$ for $j=1,\ldots,n$. Let $I^*:=\{i^*\,\, | \,\, i\in I\}$ with
$i^*\in\{1,\ldots,n-1\}$ such that $\overline{w}_0(\alpha_i)=\alpha_{i^*}$.

\begin{theorem}\label{CMthm}
Fix $c\in\mathbb{C}^*$, $I \subseteq \{ 1, \ldots, n-1 \}$ and $\zeta\in T^I$.
	Then, we have a linear isomorphism
	\begin{align} \label{cmbijection}
	\left \{ f(\mathbf{z}) \in \C [\mathbf{z}]\,\,  \middle | \,\,
		\begin{matrix} 
			\pi_{n}^{t^{-\frac{1}{2}},q}(Y_j)f = c(\overline{w}_0\zeta^{-1})_jf  &  \text{ for all } 1\leq j \leq n  \\ 
			\pi_{n}^{t^{-\frac{1}{2}},q}(T_{i})f = t^{\frac{1}{2}} f & \text{ for all } i\in I^* 			
			\end{matrix}   
			\right\}
		\overset{\sim}{\underset{\textup{CM}_{I,\zeta}}{\longrightarrow}} \text{\emph{Sol}}_n(M^I(\zeta),q,c) 
	\end{align}	
with $\textup{CM}_{I,\zeta}$ given by
	\begin{align*}
		\textup{CM}_{I, \zeta}(f) : = \sum_{w \in S^I_n} \pi_{n}^{t^{-1/2},q} (T_{w \overline{w}_0^{-1} } ) f \otimes v^I_w(\zeta).
	\end{align*}
\end{theorem}
\begin{proof}
For $c=1$, this is an easy consequence
of \cite[Cor. 4.4 \& Thm. 4.14]{Stokman:2011aa}. For general $c$, it then follows
using the fact that $M^I(c^{-1}\zeta)\simeq M^I(\zeta)^{(c^{-1})}$ with isomorphism given by $v_w^I(c^{-1}\zeta)\mapsto v_w^I(\zeta)$ for $w\in S_n^{I}$, and 
\[
\textup{Sol}_n(M^I(\zeta)^{(c^{-1})};q,1)=
\textup{Sol}_n(M^I(\zeta);q,c).
\]
\end{proof}
We want to re-express the common eigenspace for $\pi^{t^{-\frac{1}{2}},q}(Y_j)$-operators
in the left-hand side of \eqref{cmbijection} in terms of the dual Cherednik operators, in order
to apply the results of \cite{Kasatani:2005aa} in the next subsection.
The dual $Y$-operators are defined by 
\[
\overline{Y}_j:=T_j\cdots T_{n-1}\rho^{-1}T_1^{-1}\cdots T_{j-1}^{-1}\in\mathcal{H}_n\qquad (1\leq j\leq n),
\]
cf. \cite[\S 2.2]{Kasatani:2005aa}. The relation to our commuting $Y$-operators
\[
Y_j=T_{j-1}^{-1}T_{j-2}^{-1}\cdots T_1^{-1}\rho T_{n-1}\cdots T_{j+1}T_j
\]
is as follows.
\begin{lemma}
We have in $\mathcal{H}_n$,
\begin{equation*}
\begin{split}
T_{w_0}T_i&=T_{n-i}T_{w_0},\qquad\quad 1\leq i<n,\\
T_{w_0}Y_j&=\overline{Y}_{n+1-j}^{-1}T_{w_0},\qquad 1\leq j\leq n.
\end{split}
\end{equation*}
\end{lemma}
\begin{proof}
The first identity is well known. For the second identity, it suffices to show that
\[
T_{w_0}\rho T_{w_0}^{-1}=T_{n-1}\cdots T_1\rho T_{n-1}^{-1}\cdots T_1^{-1}.
\]
This follows using $\rho T_i=T_{i+1}\rho$ and the fact that
\begin{equation*}
\begin{split}
w_0&=(s_{n-1}\cdots s_1)(s_{n-1}\cdots s_2)\cdots (s_{n-1}s_{n-2})s_{n-1},\\
&=(s_1\cdots s_{n-1})(s_{n-2}\cdots s_1)\cdots (s_{n-2}s_{n-3})s_{n-2}
\end{split}
\end{equation*}
are two reduced expressions for $w_0\in S_n$.
\end{proof}
Returning to the Cherednik-Matsuo correspondence (see Theorem \ref{CMthm}), we can reformulate
it as follows.
\begin{corollary}\label{CMdualcor}
Fix $c\in\mathbb{C}^*$, $I \subseteq \{ 1, \ldots, n-1 \}$ and $\zeta\in T^I$.
	Then, we have a linear isomorphism
	\begin{align} \label{cmbijection2}
	\left \{ f(\mathbf{z}) \in \C [\mathbf{z}]\,\,  \middle | \,\,
		\begin{matrix} 
			\pi_{n}^{t^{-\frac{1}{2}},q}(\overline{Y}_j)f = c^{-1}(\underline{w}_0\zeta)_jf  &  \text{ for all } 1\leq j \leq n  \\ 
			\pi_{n}^{t^{-\frac{1}{2}},q}(T_{n-i})f = t^{\frac{1}{2}} f & \text{ for all } i\in I^* 			
			\end{matrix}   
			\right\}
		\overset{\sim}{\underset{\overline{\textup{CM}}_{I,\zeta}}{\longrightarrow}} \text{\emph{Sol}}_n(M^I(\zeta),q,c) 
	\end{align}	
with $\overline{\textup{CM}}_{I,\zeta}$ given by
	\begin{align*}
		\overline{\textup{CM}}_{I, \zeta}(f) : = \sum_{w \in S^I_n}
		 \pi_{n}^{t^{-1/2},q} (T_{w^{-1}}^{-1}T_{\underline{w}_0}^{-1}) f \otimes v^I_w(\zeta).
	\end{align*}
\end{corollary}
\begin{proof}
By the previous lemma, $\pi^{t^{-\frac{1}{2}},q}(T_{w_0}^{-1})$ restricts to a linear isomorphism from the space defined by the left-hand side of \eqref{cmbijection2} onto the space defined by the left-hand side of \eqref{cmbijection}. Hence, it suffices to note that 
\[
\overline{\textup{CM}}_{I,\zeta}=\textup{CM}_{I,\zeta}\circ\pi^{t^{-\frac{1}{2}},q}(T_{w_0}^{-1}),
\]
which follows from the fact that for all $w\in S_n^I$,
\[
T_{w\overline{w}_0^{-1}}T_{w_0}^{-1}= T^{-1}_{w^{-1}} T_{\overline{w}_0^{-1}} T^{-1} _{\overline{w}_0^{-1}}T_{\underline{w}^{-1}_0}^{-1}=T_{w^{-1}}^{-1}T_{\underline{w}_0}^{-1}. 
\]
\end{proof}

\subsection{Dual nonsymmetric Macdonald polynomials}
In this subsection, we take $n\geq 2$.
The next step will be to introduce the polynomial eigenfunctions of the dual Cherednik operators
$\pi^{t^{-\frac{1}{2}},q}(\overline{Y}_j)$ ($1\leq j\leq n$), called the dual nonsymmetric
Macdonald polynomials. We follow Kasatani \cite{Kasatani:2005aa}: the $(t^{\frac{1}{2}},\omega,Y_j)$
in \cite{Kasatani:2005aa} corresponds to our $(-t^{-\frac{1}{2}},\rho^{-1},\overline{Y}_j)$.

For $\lambda\in\mathbb{Z}^n$, let 
\begin{align*}
		 \rho ( \lambda)  &:=  \frac{1}{2}\sum _{1 \leq i< j\leq n} \chi( \lambda_i -\lambda_j)(\epsilon_i -\epsilon_j),\\
		\chi(a) &:= \begin{cases} 
				1 & \text{ if } a \geq 0, \\
				-1 & \text{ if } a < 0.
				\end{cases}
\end{align*}
Then, 
$2 \rho( \lambda) = \sum^{n}_{i=1} d_i(\lambda) \epsilon_i$ with 
\begin{equation}\label{rhoformula}
d_i( \lambda)= 2 \# \{ j>i | \lambda_j = \lambda_i \}  + 2 \# \{ j | \lambda_i > \lambda_j \} + 1-n.
\end{equation}
 Write
\[
s_\lambda:=\bigl(-t^{-\frac{1}{2}}\bigr)^{2\rho(\lambda)}q^\lambda\in (\mathbb{C}^*)^n,\qquad \lambda\in\mathbb{Z}^n,
\]
i.e. $s_{\lambda}=(s_{\lambda,1},\ldots,s_{\lambda,n})$ with 
$s_{\lambda,i}=\bigl(-t^{-\frac{1}{2}}\bigr)^{d_i(\lambda)}q^{\lambda_i}$.

For generic $q$ and $t^{\frac{1}{4}}$ (or indeterminates), the monic dual nonsymmetric Macdonald polynomial
\[
E_\lambda=E_\lambda(\mathbf{z};-t^{-\frac{1}{2}},q)\in\mathbb{C}[\mathbf{z}^{\pm 1}]
\] 
of degree $\lambda\in\mathbb{Z}^n$ is the unique Laurent polynomial satisfying the eigenvalue equations
\[
\pi^{t^{-\frac{1}{2}},q}_n(f(\overline{Y}))E_\lambda=f(s_\lambda)E_\lambda\qquad
\text{ for all }\, f\in\mathbb{C}[\mathbf{z}^{\pm 1}]
\]
such that the coefficient of $\mathbf{z}^\lambda$ in the expansion of $E_\lambda$ in monomials 
$\{\mathbf{z}^\nu\}_{\nu\in\mathbb{Z}^n}$ is one. It is well known that $E_\lambda$ is homogeneous of total degree $|\lambda|:=\lambda_1+\cdots+\lambda_n$. In addition, $E_\lambda\in\mathbb{C}[\mathbf{z}]$ if and only if $\lambda\in\mathbb{Z}_{\geq 0}^n$. The intertwiners with respect to 
the dual $Y$-operators are defined by 
\[
B_i:=T_i\left(\overline{Y}_{i+1}\overline{Y}_i^{-1}-1\right)+t^{\frac{1}{2}}-t^{-\frac{1}{2}},\qquad 1\leq i<n,
\]
cf. \cite[Lemma 2.6]{Kasatani:2005aa}. Then, for $1\leq i<n$,
\begin{equation}\label{Brelation}
\pi^{t^{-\frac{1}{2}},q}(B_i)E_\lambda=-t^{\frac{1}{2}}
\left(\frac{(ts_{\lambda,i+1}s_{\lambda,i}^{-1}-1)(t^{-1}s_{\lambda,i+1}s_{\lambda,i}^{-1}-1)}
{(s_{\lambda,i+1}/s_{\lambda,i}-1)}\right)E_{s_i\lambda}
\end{equation}
if $\lambda\in\mathbb{Z}^n$ and $\lambda_{i}>\lambda_{i+1}$.

Kasatani \cite{Kasatani:2005aa} analyzed the dual nonsymmetric Macdonald polynomials
$E_\lambda$ with parameters specialized to $t^{-k-1}q^{r-1}=1$ with $1\leq k\leq n-1$ and $r\geq 2$.
In our situation, we are going to need the special case that $k=2$ and $r=3$, i.e., when
$t^{-3}q^2=1$ (cf. Theorem \ref{mainTHM}). In fact, for our purposes it suffices to take $q=t^{\frac{3}{2}}$. We recall some key results from \cite{Kasatani:2005aa} in this special case.

\begin{definition}\label{def-wheel}
We say $\lambda \in \Z^n$ has a \emph{neighbourhood} if it has a pair of indices $(i,j)$ such that condition 1 and 2 are satisfied:
\begin{enumerate}
	\item $\rho(\lambda)_i - \rho(\lambda)_j = 2,$ 
	\item \begin{enumerate}
			\item $\lambda_i - \lambda_j  \leq 1,$ \\
			or  
			\item $\lambda_i - \lambda_j =2$ and $j<i$.
		\end{enumerate}
\end{enumerate}
\end{definition}
Write
\begin{align*}
	S^{(2,3)} :&=\{\lambda \in \Z^n\,\,\, | \,\,\,  \lambda \text{ has a neighbourhood}\}\\
	B^{(2,3)} :&= \Z^n\backslash S^{(2,3)}. \\
\end{align*}
By \cite[Thm. 3.11]{Kasatani:2005aa}, the dual non-symmetric Macdonald polynomial
$E_\lambda$ can be specialized at $q=t^{\frac{3}{2}}$ if $\lambda\in B^{(2,3)}$. For 
$q=t^{\frac{3}{2}}$ write
\begin{equation*}
\begin{split}
Z^{(2,3)} := \{ \mathbf{z} \in \C^n\,\,\, | \,\,\, &\text{There exist distinct } i_1, i_2, i_3 \in \{1 ,\ldots ,n\}\\
							   & \text{and positive integers } r_1,r_2 \in \Z_{\geq 0} \\
							   & \text{such that } z_{i_{a+1}} = z_{i_{a}}t q^{r_a} \text{ for }a=1,2, \\
							   &r_1+r_2\leq 1, \text{ and } i_a < i_{a+1} \text{ if } r_a=0 \},
\end{split}							   
\end{equation*}
and define the ideal $I^{(2,3)}\subseteq\mathbb{C}[\mathbf{z}]$ by 
\[ 
I^{(2,3)}: = \left\{f \in \mathbb{C}[\mathbf{z}^{\pm 1}]\,\, | \,\,  f(\mathbf{z}) =0 \text{ for all } 
\mathbf{z}\in Z^{(2,3)}\right\}.
\]
Then, for $q=t^{\frac{3}{2}}$ and generic $t^{\frac{1}{4}}$, the ideal $I^{(2,3)}$ is a 
$\pi_n^{t^{-\frac{1}{2}},t^{\frac{3}{2}}}(\mathcal{H}_n(t^{-\frac{1}{2}}))$-submodule of 
$\mathbb{C}[\mathbf{z}^{\pm 1}]$ and 
\[
I^{(2,3)}=\bigoplus_{\mu\in B^{(2,3)}}\mathbb{C}E_\mu(\mathbf{z};-t^{-\frac{1}{2}},t^{\frac{3}{2}})
\]
by \cite[Thm. 3.11]{Kasatani:2005aa}.
\begin{remark}\label{rem-wheel}
The conditions $f(\mathbf{z})=0$ for $\mathbf{z}\in Z^{(2,3)}$ are 
known as \emph{wheel conditions}. 
It originally appeared in \cite{Feigin:2003aa} (see also \cite{Kasatani:2007aa}).
\end{remark}
We recall now the notion of a $(2,3)$-wheel in $\lambda\in\mathbb{Z}^n$, following
\cite[Def. 3.5]{Kasatani:2005aa}. 
\begin{definition}
Let $q=t^{\frac{3}{2}}$ and fix $\lambda\in\mathbb{Z}^n$. A three-tuple $(i_1,i_2,i_3)$ with distinct $i_1,i_2,i_3\in \{1,\ldots,n\}$ is called a $(2,3)$-wheel in $\lambda$ if there exists $r_1,r_2\in\mathbb{Z}_{\geq 0}$ such that
\[
s_{\lambda,i_2}^{-1}=s_{\lambda,i_1}^{-1}t^{-1}q^{r_1},\qquad
s_{\lambda,i_3}^{-1}=s_{\lambda,i_2}^{-1}t^{-1}q^{r_2}
\]
with $r_1+r_2\leq 1$, and $i_a<i_{a+1}$ if $r_a=0$ ($a=1,2$).
\end{definition}
Two wheels $(i_1,i_2,i_3)$ and $(j_1,j_2,j_3)$ in $\lambda$ are said to be equivalent if there exists
a $\sigma\in S_3$ such that $i_a=j_{\sigma^{-1}(a)}$ for $a=1,2,3$. We write $\#^{(2,3)}(\lambda)$
for the number of equivalence classes of $(2,3)$-wheels in $\lambda$. Note that 
(still under the assumption that $q=t^{\frac{3}{2}}$) we have $\#^{(2,3)}(\lambda)=0$ if and only if $s_\lambda^{-1}\in Z^{(2,3)}$.  Furthermore,
from \cite[\S 3]{Kasatani:2005aa} (below Definition 3.7) we have
\[
\{\mu\in\mathbb{Z}^n \,\, | \,\, \#^{(2,3)}(\mu)=0 \}\subseteq B^{(2,3)}.
\]

\subsection{Proof of Theorem \ref{mainTHM}}
Let $n\geq 2$ and specialize throughout this subsection $v=1$ and $q=t^{\frac{3}{2}}$. 
Furthermore, we set
\[
c_n:=\bigl(-t^{-\frac{3}{4}}\bigr)^{n-1},
\]
cf. Theorem \ref{mainTHM}. Recall the notation
$I^{(n)}=\{1,\ldots,\lceil n/2\rceil-1, \lceil n/2\rceil+1,\ldots,n-1\}$ and the central character
$\gamma=\gamma^{(n)}\in T^{I^{(n)}}$ with $v=1$ (see \eqref{highestweight}), so
\begin{equation}\label{highestweightv=1}
\gamma=
\begin{cases}
\bigl(t^{\frac{1-n}{4}},t^{\frac{5-n}{4}},\ldots,t^{\frac{n-3}{4}},
t^{\frac{3-n}{4}}, t^{\frac{7-n}{4}},\ldots,t^{\frac{n-1}{4}}\bigr)
\qquad\quad\,\,&\hbox{ if }\,\, n\, \textup{ even},\\
\bigl(t^{\frac{1-n}{4}}, t^{\frac{5-n}{4}},\ldots,t^{\frac{n-1}{4}}, t^{\frac{3-n}{4}}, t^{\frac{7-n}{4}},
\ldots, t^{\frac{n-3}{4}}\bigr)\qquad\quad\,\,&\hbox{ if }\,\, n\,\, \textup{ odd}.
\end{cases}
\end{equation}
In the even $n=2k$ case, the decomposition $w_0=\overline{w}_0\underline{w}_0$ of the
longest element $w_0\in S_{2k}$ as a product of $\overline{w}_0\in S_{2k}^{I^{(2k)}}$ and 
$\underline{w}_0\in S_{2k,I^{(2k)}}$ gives the expressions
\begin{equation*}
\begin{split}
\overline{w}_0&=\left(\begin{matrix} 1 & 2 & \cdots & k & k+1 & k+2 & \cdots & 2k\\
k+1 & k+2 & \cdots & 2k & 1 & 2 & \cdots & k\end{matrix}\right),\\
\underline{w}_0&=\left(\begin{matrix} 1 & 2 & \cdots & k & k+1 & k+2 & \cdots & 2k\\
k & k-1 & \cdots & 1 & 2k & 2k-1 & \cdots & k+1\end{matrix}\right).
\end{split}
\end{equation*}
Hence, for $n=2k$ even, we have $I^{(2k),*}=I^{(2k)}$ and 
\[
\{ 2k-i\,\, | \,\, i\in I^{(2k),*} \}=I^{(2k)}.
\]
In the odd $n=2k-1$ case, 
the decomposition $w_0=\overline{w}_0\underline{w}_0$ of the
longest element $w_0\in S_{2k-1}$ as a product of $\overline{w}_0\in S_{2k-1}^{I^{(2k-1)}}$ and 
$\underline{w}_0\in S_{2k-1,I^{(2k-1)}}$ gives the expressions
\begin{equation*}
\begin{split}
\overline{w}_0&=\left(\begin{matrix} 1 & 2 & \cdots & k & k+1 & k+2 & \cdots & 2k-1\\
k & k+1 & \cdots & 2k-1 & 1 & 2 & \cdots & k-1\end{matrix}\right),\\
\underline{w}_0&=\left(\begin{matrix} 1 & 2 & \cdots & k & k+1 & k+2 & \cdots & 2k-1\\
k & k-1 & \cdots & 1 & 2k-1 & 2k-2 & \cdots & k+1\end{matrix}\right).
\end{split}
\end{equation*}
Therefore, $I^{(2k-1),*}=\{1,\ldots,k-2,k,k+1,\ldots,2k-2\}$ and 
\[
\{ 2k-1-i\,\, | \,\, i\in I^{(2k-1),*} \}=I^{(2k-1)}.
\]
Hence, it follows from Theorem \ref{principalTHM} and Corollary \ref{CMdualcor} 
that for generic $t^{\frac{1}{4}}$, we have
	\begin{align} \label{cmbijection3}
	\left \{ f(\mathbf{z}) \in \C [\mathbf{z}]\,\,  \middle | \,\,
		\begin{matrix} 
			\pi_{n}^{t^{-\frac{1}{2}},t^{\frac{3}{2}}}(\overline{Y}_j)f = c_n^{-1}(\underline{w}_0\gamma)_jf  &  \text{ for all } 1\leq j \leq n  \\ 
			\pi_{n}^{t^{-\frac{1}{2}},t^{\frac{3}{2}}}(T_{i})f = t^{\frac{1}{2}} f & \text{ for all } i\in I^{(n)}		
			\end{matrix}   
			\right\}
		\overset{\sim}{\underset{\widetilde{\textup{CM}}_n}{\longrightarrow}} \text{\emph{Sol}}_n(V_n,q,c_n) 
	\end{align}	
with $\widetilde{\textup{CM}}_n$ given by
	\begin{align*}
		\widetilde{\textup{CM}}_{n}(f) : = \sum_{w \in S^{I^{(n)}}_n}
		 \pi_{n}^{t^{-\frac{1}{2}},t^{\frac{3}{2}}} (T_{w^{-1}}^{-1}T_{\underline{w}_0}^{-1}) f \otimes 
		 T_wI_{w_n}Q_n.
	\end{align*}
In the next lemma, we relate the spectral point $c_n^{-1}(\underline{w}_0\gamma)$ to the spectrum
of the dual nonsymmetric Macdonald polynomials.
\begin{lemma}
For $q=t^{\frac{3}{2}}$,
\[
c_n^{-1}\underline{w}_0\gamma^{(n)}=s_{\lambda^{(n)}}
\]
with $\lambda^{(n)}\in\mathbb{Z}^n$ given by
\begin{equation*}
\begin{split}
\lambda^{(2k)}&=(2k-2, 2k-4,\ldots,0,2k-1.2k-3,\ldots,1),\\
\lambda^{(2k-1)}&=(2k-2,2k-4,\ldots,0,2k-3,2k-5,\ldots,1).
\end{split}
\end{equation*}
\end{lemma}
\begin{proof}
By a direct computation, for $q=t^{\frac{3}{2}}$,
\begin{equation*}
\begin{split}
c_{2k}^{-1}\underline{w}_0\gamma^{(2k)}&=\bigl(-t^{2k-\frac{3}{2}},-t^{2k-\frac{5}{2}},\ldots,
-t^{k-\frac{1}{2}},-t^{2k-1},-t^{2k-2},\ldots,-t^k\bigr)=s_{\lambda^{(2k)}},\\
c_{2k-1}^{-1}\underline{w}_0\gamma^{(2k-1)}&=
\bigl(t^{2k-2},t^{2k-3},\ldots,t^{k-1},t^{2k-\frac{5}{2}},t^{2k-\frac{7}{2}},\ldots,t^{k-\frac{1}{2}}\bigr)=
s_{\lambda^{(2k-1)}}.
\end{split}
\end{equation*}
\end{proof}
Hence, we have for generic $t^{\frac{1}{4}}$, 
	\begin{align} \label{cmbijection4}
	\left \{ f(\mathbf{z}) \in \C [\mathbf{z}]\,\,  \middle | \,\,
		\begin{matrix} 
			\pi_{n}^{t^{-\frac{1}{2}},t^{\frac{3}{2}}}(p(\overline{Y}))f = p(s_{\lambda^{(n)}})f  &  \text{ for all } p(\mathbf{z})\in\mathbb{C}[\mathbf{z}^{\pm 1}]  \\ 
			\pi_{n}^{t^{-\frac{1}{2}},t^{\frac{3}{2}}}(T_{i})f = t^{\frac{1}{2}} f & \text{ for all } i\in I^{(n)}		
			\end{matrix}   
			\right\}
		\overset{\sim}{\underset{\widetilde{\textup{CM}}_n}{\longrightarrow}} \text{\emph{Sol}}_n(V_n,q,c_n) .
	\end{align}	
Next we need to verify that the dual non-symmetric Macdonald polynomials $E_{\lambda^{(n)}}$ are nonzero under the specialisation $q=t^{\frac{3}{2}}$. 
\begin{lemma}
$S_{n}\lambda^{(n)}\cap B^{(2,3)}=
\{\sigma\lambda^{(n)}\,\,\, | \,\,\, \sigma\in S_{n}^{I^{(n)}} \}$. In particular, $\lambda^{(n)}\in B^{(2,3)}$,
and for generic $t^{\frac{1}{4}}$,
\[
0\not=E_{\lambda^{(n)}}(\mathbf{z};-t^{-\frac{1}{2}},t^{\frac{3}{2}})\in\mathbb{C}[\mathbf{z}]
\] 
is 
well defined, homogeneous of total degree $\frac{1}{2}n(n-1)$, and satisfies
\[
\pi_{n}^{t^{-\frac{1}{2}},t^{\frac{3}{2}}}(p(\overline{Y}))E_{\lambda^{(n)}}(\cdot;t^{-\frac{1}{2}},t^{\frac{3}{2}}) = p(s_{\lambda^{(n)}})E_{\lambda^{(n)}}(\cdot;-t^{-\frac{1}{2}},t^{\frac{3}{2}}) \qquad
\forall\, p(\mathbf{z})\in\mathbb{C}[\mathbf{z}^{\pm 1}].
\]
\end{lemma}
\begin{proof}
Note that by Remark \ref{mincoset} and the definition of $\lambda^{(n)}$,
\begin{equation}\label{help}
\{\sigma\lambda^{(n)} \,\,\, | \,\,\, \sigma\in S_n^{I^{(n)}}\}=\{\mu\in S_n\lambda^{(n)} \,\,\, | \,\,\, \mu_i-\mu_j=2\, \Rightarrow i<j \}.
\end{equation}
Now suppose that $\mu\in S_n\lambda^{(n)}\cap B^{(2,3)}$ and $\mu_i-\mu_j=2$. By \eqref{rhoformula},
\begin{equation*}
\begin{split}
\rho(\mu)_i-\rho(\mu)_j&=\#\{r\,\, | \,\, \mu_r<\mu_i\}-\#\{r\,\, | \,\, \mu_r<\mu_j\}\\
&=\#\{r \,\, | \,\, \mu_j\leq \mu_r\leq \mu_j+1\}=2.
\end{split}
\end{equation*}
Since $\mu\in B^{(2,3)}$, this implies that $i<j$. By \eqref{help}, we conclude that $\mu\in 
\{\sigma\lambda^{(n)}\,\, | \,\, 
\sigma\in S_n^{I^{(n)}}\}$.

Conversely, suppose that $\mu\in \{\sigma\lambda^{(n)}\,\, | \,\, 
\sigma\in S_n^{I^{(n)}}\}$. Suppose that $\rho(\mu)_i-\rho(\mu)_j=2$. By \eqref{rhoformula}, this implies that
\[
\#\{r \,\, | \,\, \mu_r<\mu_i\}-\#\{r \,\, | \,\, \mu_r<\mu_j\}=2.
\]
It follows that $\mu_i>\mu_j$, and hence
\[
\#\{r \,\, | \,\, \mu_j\leq \mu_r<\mu_i\}=2,
\]
forcing $\mu_i-\mu_j=2$. By \eqref{help}, this implies that $i<j$, and hence $\mu\in S_n\lambda^{(n)}\cap B^{(2,3)}$. The remaining statements now follow immediately (note that the degree of $E_{\lambda^{(n)}}(\mathbf{z})$ is $\sum_{i=1}^n\lambda_i^{(n)}=\frac{1}{2}n(n-1)$).
\end{proof}
\begin{proposition}
For generic $t^{\frac{1}{4}}$, we have 
\[
\pi^{t^{-\frac{1}{2}},t^{\frac{3}{2}}}(T_i)E_{\lambda^{(n)}}(\cdot;-t^{-\frac{1}{2}},t^{\frac{3}{2}})=
t^{\frac{1}{2}}E_{\lambda^{(n)}}(\cdot;-t^{-\frac{1}{2}},t^{\frac{3}{2}})\qquad \forall \, i\in I^{(n)}.
\]
In particular, there exists a unique $\kappa_n\in\mathbb{C}^*$ such that
\[
g^{(n)}:=\kappa_n\widetilde{\textup{CM}}_n\bigl(E_{\lambda^{(n)}}(\cdot;-t^{-\frac{1}{2}},t^{\frac{3}{2}})
\bigr)\in
\textup{Sol}_n\bigl(V_n; t^{\frac{3}{2}}, (-t^{-\frac{3}{4}})^{n-1}\bigr)
\]
has fully nested component
\begin{equation}
g^{(n)}_{L_\cap}(\mathbf{z})=\prod_{1\leq i<j\leq n}\left(t^{\frac{1}{2}}z_j-t^{-\frac{1}{2}}z_i\right). \label{fnccm}
\end{equation}
\end{proposition}
\begin{proof}
For the first statement, fix $i\in I^{(n)}$. 
In fact, for $n=2k$ even the $(2,3)$-wheels in $s_i\lambda^{(2k)}$ are
$(i,i+1,k+i+1)$ and $(k+i+1,i,i+1)$ if $1\leq i<k$ and $(i,i+1,i-k)$, $(i-k,i,i+1)$ if $k+1\leq i<2k$.
For $n=2k-1$, the $(2,3)$-wheels in
$s_i\lambda^{(2k-1)}$ are $(i,i+1,k+i)$, $(k+i,i,i+1)$ if $1\leq i<k$ and $(i,i+1,i-k+1)$, 
$(i-k+1,i,i+1)$ if $k+1\leq i<2k-1$. Hence, $\#^{(2,3)}(s_i\lambda^{(n)})=1$.

By \cite[Lemma 4.13]{Kasatani:2005aa}, it follows from $\#^{(2,3)}(s_i\lambda^{(n)})=1$
that $E_{s_i\lambda^{(n)}}(\mathbf{z};-t^{-\frac{1}{2}},q)$ can be specialized to $q=t^{\frac{3}{2}}$
for generic $t^{\frac{1}{4}}$. By \eqref{Brelation}, we then obtain
\[
\pi^{t^{-\frac{1}{2}},t^{\frac{3}{2}}}(B_i)E_{\lambda^{(n)}}(\cdot;-t^{-\frac{1}{2}},t^{\frac{3}{2}})=0
\]
since $s_{\lambda^{(n)},i+1}/s_{\lambda^{(n)},i}=t^{-1}$ and $\lambda_{i}^{(n)}>\lambda_{i+1}^{(n)}$.
Substituting the explicit expression of $B_i$ then gives
\begin{equation*}
\begin{split}
0&=\pi^{t^{-\frac{1}{2}},t^{\frac{3}{2}}}(B_i)E_{\lambda^{(n)}}(\cdot;-t^{-\frac{1}{2}},t^{\frac{3}{2}})\\
&=\left((s_{\lambda^{(n)},i+1}s_{\lambda^{(n)},i}^{-1}-1)\pi_n^{t^{-\frac{1}{2}},t^{\frac{3}{2}}}(T_i)
+(t^{\frac{1}{2}}-t^{-\frac{1}{2}})\right)E_{\lambda^{(n)}}(\cdot;-t^{-\frac{1}{2}},t^{\frac{3}{2}})\\
&=(t^{-1}-1)\left(\pi^{t^{-\frac{1}{2}},t^{\frac{3}{2}}}(T_i)-t^{\frac{1}{2}}\right)E_{\lambda^{(n)}}(\cdot;-t^{-\frac{1}{2}},t^{\frac{3}{2}}),
\end{split}
\end{equation*}
hence $\pi_n^{t^{-\frac{1}{2}},t^{\frac{3}{2}}}(T_i)E_{\lambda^{(n)}}(\cdot;-t^{-\frac{1}{2}},t^{\frac{3}{2}})=
t^{\frac{1}{2}}E_{\lambda^{(n)}}(\cdot;-t^{-\frac{1}{2}},t^{\frac{3}{2}})$ for $i\in I^{(n)}$.

It follows that 
\[
0\not=\widetilde{g}^{(n)}:=\widetilde{\textup{CM}}_n\bigl(E_{\lambda^{(n)}}(\cdot;-t^{-\frac{1}{2}},t^{\frac{3}{2}})
\bigr)\in\textup{Sol}_n(V_n;t^{\frac{3}{2}},(-t^{-\frac{3}{4}})^{n-1})
\]
is homogeneous of total degree $\frac{1}{2}n(n-1)$. By Corollary \ref{degreeleading},
\[
\widetilde{g}_{L_\cap}^{(n)}(\mathbf{z})=
\kappa_n\prod_{1\leq i<j\leq n}\left(t^{\frac{1}{2}}z_j-t^{-\frac{1}{2}}z_i\right)
\]
for some $\kappa_n\in\mathbb{C}^*$, hence the result. 
\end{proof}
With the last proposition we have completed the proof of Theorem \ref{mainTHM} for generic
$t^{\frac{1}{4}}\in\mathbb{C}^*$ (note that Lemma \ref{unique1} implies uniqueness, and that for
$n=1$ the desired unique solution $g^{(1)}$ is simply given by the constant function
$g^{(1)}\equiv 1$). The remark following Theorem \ref{mainTHM} then completes the proof
of Theorem \ref{mainTHM} for 
all values $t^{\frac{1}{4}}\in\mathbb{C}^*$ for which $(t^{\frac{1}{2}}+1)(t+1)\not=0$.

 \section{The dual braid recursion}\label{dual}
The extended affine Temperley-Lieb algebra $\mathcal{TL}_n(t^{\frac{1}{2}})$ is invariant under the inversion $t^{\frac{1}{4}} \to t^{-\frac{1}{4}}$.
In this last section of the article, we discuss how this symmetry results in a dual braid recursion
for the qKZ towers of solutions $(g^{(n)})_{n\geq 0}$ from Theorem \ref{mainTHM}. We set $v=1$
in this section.

First, we discuss the inversion on the Kauffman skein relation \eqref{kauffman}.
If we invert $t^{\frac{1}{4}}$ and then rotate the diagrams by ninety degrees, we have
\begin{align*}
		&\psset{unit=0.8}\begin{pspicture}[shift=-0.9](-1,-1)(1,1)
			\pscircle[fillstyle=solid,fillcolor=diskin,linewidth=0.8pt,linestyle=dotted](0,0){1}
			\degrees[8]
			\psline[linecolor= line,linewidth=1.2pt](1;3)(1;7)
			\psline[linecolor= blue,linewidth=1.2pt](1;1)(0.2;1)
			\psline[linecolor= blue,linewidth=1.2pt](1;5)(0.2;5)	
		\end{pspicture}
 = t^{\frac{1}{4}}\; 
		\begin{pspicture}[shift=-0.9](-1,-1)(1,1)
			\pscircle[fillstyle=solid,fillcolor=diskin,linewidth=0.8pt,linestyle=dotted ](0,0){1}
			\degrees[8]
			\pscurve[linecolor= line,linewidth=1.2pt](1;3)(0.9;3)(0.4;2)(0.9;1)(1;1)
			\pscurve[linecolor= line,linewidth=1.2pt](1;5)(0.9;5)(0.4;6)(0.9;7)(1;7)
		\end{pspicture}
		+ t^{-\frac{1}{4}} \;
		\begin{pspicture}[shift=-0.9](-1,-1)(1,1)
			\pscircle[fillstyle=solid,fillcolor=diskin,linewidth=0.8pt,linestyle=dotted](0,0){1}
			\degrees[8]
			\pscurve[linecolor= line,linewidth=1.2pt](1;5)(0.9;5)(0.4;4)(0.9;3)(1;3)
			\pscurve[linecolor= line,linewidth=1.2pt](1;1)(0.9;1)(0.4;0)(0.9;7)(1;7)
		\end{pspicture}\;.
	\end{align*}
Comparing this to the original equation, we see that all the over-crossings swap to under-crossings and vice versa. Hence, in the identification of $\mathcal{TL}_n(t^{\frac{1}{2}})$ with 
$\textup{End}_{\mathcal{S}}(n)$, the inversion $t^{\frac{1}{4}}\rightarrow t^{-\frac{1}{4}}$ amounts
to replacing under-crossings by over-crossing and vice versa. 

Let $\mathcal{I}_n^\iota: \mathcal{TL}_n\rightarrow\mathcal{TL}_{n+1}$ be the algebra map $\mathcal{I}_n$ with the role of $t^{\frac{1}{4}}$ replaced by $t^{-\frac{1}{4}}$. Hence, in the skein theoretic 
description, the arc insertion is now done by over-crossing all arcs its meets, instead of under-crossing.
Similarly, we write 
\[
\phi_n^\iota\in \textup{Hom}_{\mathcal{TL}_n}(V_n,V_{n+1}^{\mathcal{I}_n^\iota})
\]
for the intertwiner obtained from replacing in the construction of 
$\phi_n\in\textup{Hom}_{\mathcal{TL}_n}\bigl(V_n,V_{n+1}^{\mathcal{I}_n}\bigr)$
the parameter $t^{\frac{1}{4}}$ by $t^{-\frac{1}{4}}$ and the role of $\mathcal{I}_n$ by 
$\mathcal{I}_n^\iota$. 
\begin{theorem}
With the assumptions as in Theorem \ref{mainTHM},
let $g^{(n)}\in\textup{Sol}_n(V_n,t^{\frac{3}{2}}, (-t^{-\frac{3}{4}})^{n-1})$ for $n\geq 1$ be the homogeneous polynomial solution qKZ of degree $\frac{1}{2}n(n-1)$ with fully nested component
\[
g^{(n)}_{L_\cap}(\mathbf{z})=\prod_{1\leq i<j\leq n}\left(t^{\frac{1}{2}}z_j-t^{-\frac{1}{2}}z_i\right)
\]
and set $g^{(0)}:=1\in \textup{Sol}_0(V_0,t^{\frac{3}{2}},t^{\frac{1}{4}}+t^{-\frac{1}{4}})$. Write
\[
\widetilde{g}^{(n)}(\mathbf{z}):=(z_1z_2\cdots z_n)^{n-1}g^{(n)}(z_1^{-1},z_2^{-1},\ldots,z_n^{-1}).
\]
Then, $\widetilde{g}^{(n)}(\mathbf{z})\in V_n[\mathbf{z}]$ is a $V_n$-valued homogeneous 
polynomial of total degree $\frac{1}{2}n(n-1)$ and 
\[
\widetilde{g}^{(n+1)}(z_1,\ldots,z_n,0)=t^{\frac{1}{4}(2n-\lfloor n/2\rfloor)}z_1\cdots z_n
\phi_n^\iota\bigl(\widetilde{g}^{(n)}(z_1,\ldots,z_n)\bigr),\qquad n\geq 0.
\]
\end{theorem}
\begin{proof}
Let $R_i^\iota(x)\in\mathcal{TL}_n$ be the $R$-operator $R_i(x)$ with $t^{\frac{1}{4}}$ replaced by 
$t^{-\frac{1}{4}}$,
\[
R_i^\iota(x)=\left(\frac{x-1}{t^{-\frac{1}{2}}-t^{\frac{1}{2}}x}\right)e_i+
\left(\frac{xt^{-\frac{1}{2}}-t^{\frac{1}{2}}}{t^{-\frac{1}{2}}-t^{\frac{1}{2}}x}\right).
\] 
Note that $R_i^\iota(x)=R_i(x^{-1})$, from which it follows that
\[
R_i^\iota(z_{i+1}/z_i)\widetilde{g}^{(n)}(\ldots,z_{i+1},z_i,\ldots)=\widetilde{g}^{(n)}(\mathbf{z})
\]
for $1\leq i<n$. Furthermore, with $q=t^{\frac{3}{2}}$,
\begin{equation*}
\begin{split}
\rho \widetilde{g}^{(n)}(z_2,\ldots,z_n,qz_1)&=q^{n-1}(z_1\cdots z_n)^{n-1}\rho g^{(n)}(z_2^{-1},
\ldots,z_n^{-1},q^{-1}z_1^{-1})\\
&=q^{n-1}(-t^{-\frac{3}{4}})^{n-1}(z_1\cdots z_n)^{n-1}g^{(n)}(z_1^{-1},\ldots,z_n^{-1})\\
&=(-t^{\frac{3}{4}})^{n-1}\widetilde{g}^{(n)}(\mathbf{z}).
\end{split}
\end{equation*}
Hence, $\widetilde{g}^{(n)}(\mathbf{z})$ is a twisted symmetric solution of the qKZ equation
with respect to the action $\nabla^\iota$ of $W_n$ obtained from $\nabla$ by inverting $t^{\frac{1}{4}}$ and setting $q=t^{\frac{3}{2}}$. 
Furthermore, the fully nested component of $\widetilde{g}^{(n)}(\mathbf{z})$ is
\[
\widetilde{g}^{(n)}_{L_\cap}(\mathbf{z})=\prod_{1\leq i<j\leq n}(t^{\frac{1}{2}}z_i-t^{-\frac{1}{2}}z_j).
\]
As in Lemma \ref{unique1}, it follows that $\widetilde{g}^{(n)}(\mathbf{z})$, as symmetric solution
of these qKZ equations, is determined by the fully nested component and that all coefficients 
$\widetilde{g}^{(n)}_L(\mathbf{z})$ are homogeneous polynomials in $z_1,\ldots, z_n$ of total
degree $\frac{1}{2}n(n-1)$. In particular,
\[
\widetilde{g}^{(n)}(\mathbf{z})\in\textup{Sol}_n^\iota\bigl(V_n;t^{\frac{3}{2}},(-t^{\frac{3}{4}})^{n-1}\bigr)
\]
with $\textup{Sol}_n^\iota\bigl(V_n;q,d_n)$ being the polynomial $V_n$-valued functions
$f(\mathbf{z})\in V_n[\mathbf{z}]$ satisfying
\begin{equation*}
\begin{split}
R_i^\iota(z_{i+1}/z_i)f(\ldots,z_{i+1},z_i,\ldots)&=f(\mathbf{z}),\qquad 1\leq i<n,\\
\rho f(z_2,\ldots,z_n,qz_1)&=d_nf(\mathbf{z}).
\end{split}
\end{equation*}
Using slightly modified versions of 
Lemmas \ref{liftedsol} and \ref{downlemma} one now shows that 
\begin{equation*}
\begin{split}
(z_1\cdots z_n)\phi_n^\iota\bigl(\widetilde{g}^{(n)}(z_1,\ldots,z_n)\bigr)&\in
\textup{Sol}_n^\iota\bigl(V_{n+1}^{\mathcal{I}_n^\iota};
t^{\frac{3}{2}}, (-t^{\frac{3}{4}})^{n+1}\bigr),\\
\widetilde{g}^{(n+1)}(z_1,\ldots,z_n,0)&\in \textup{Sol}_n^\iota\bigl(V_{n+1}^{\mathcal{I}_n^\iota};
t^{\frac{3}{2}}, (-t^{\frac{3}{4}})^{n+1}\bigr).
\end{split}
\end{equation*}
Furthermore, a modified version of Lemma \ref{intertwinerfullynested} yields
\[
\phi_n^\iota(L)=\sum_{L^\prime\in\mathcal{L}_{n+1}}d_{L^\prime,L}L^\prime
\]
with $d_{L,L_\cap^{(n+1)}}=\delta_{L,L_\cap^{(n)}}t^{\frac{1}{4}\lfloor n/2\rfloor}$. Hence
$\widetilde{g}^{(n+1)}(z_1,\ldots,z_n,0)$ and $t^{\frac{1}{4}(2n-\lfloor n/2\rfloor)}z_1\cdots z_n
\phi_n^\iota\bigl(\widetilde{g}^{(n)}(\mathbf{z})\bigr)$ have the same fully nested component.
The properly modified version of Lemma \ref{unique2} then shows that they are equal.
\end{proof}

\appendix
\section{Recursion relations for link pattern components of qKZ solutions} 
\label{uniqueproof}
In this appendix, we prove Lemmas \ref{unique1}{\bf (a)} and \ref{unique2}{\bf (b)}, i.e.,
we show that twisted symmetric solutions of qKZ equations with values in link pattern modules
are determined by the fully nested component. 
This proof is done for three different representations of $\mathcal{H}_n$; link patterns $\C[LP_{2k}]$ (when $n=2k$), punctured link patterns $\C[\mathcal{L}_{n}]\simeq V_n$ (Lemma \ref{unique1})  and the restricted modules $V^{\nu_n}_{n+1}$ (Lemma \ref{unique2}). The first two cases have been considered before in the literature, see \cite{Di-Francesco:2006aa, Kasatani:2007aa, De-Gier:2010aa}. We recall these here in detail, since the technicalities play an important role in proving the most delicate third case.

A link pattern of size $2k$ is a diagram with $2k$ equally spaced points on the boundary of the 
disc $\mathbb{D}$ that are connected by $k$ non-intersecting curves lying within the disc. 
To establish convention the points are numbered $1$ to $2k$ going counter-clockwise around the disc.
We denote the set of link patterns of size $2k$ by $LP_{2k}$.
As an example, $LP_6$ consists of the following link patterns:

$$\psset{unit=0.6}
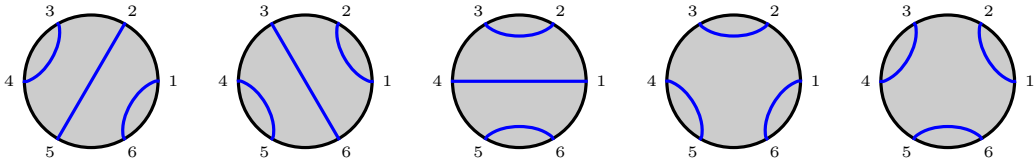
\begin{pspicture}[shift=-1.8](-2.25,-1.8)(2.25,1.8)
    	\pscircle[fillstyle=solid, fillcolor=diskin,linewidth=1.2pt](0,0){1.5}
    	\degrees[6]
	\rput{0}(1.8;0){\tiny$1$}
	\rput{0}(1.8;1){\tiny$2$}
	\rput{0}(1.8;2){\tiny$3$}
	\rput{0}(1.8;3){\tiny$4$}
	\rput{0}(1.8;4){\tiny$5$}
	\rput{0}(1.8;5){\tiny$6$}
	\pscurve[linecolor=line,linewidth=1.2pt](1.5;5)(1.4;5)(1;5.5)(1.4;0)(1.5;0)
	\pscurve[linecolor=line,linewidth=1.2pt](1.5;4)(1.4;4)(0;0)(1.4;1)(1.5;1)
	\pscurve[linecolor=line,linewidth=1.2pt](1.5;3)(1.4;3)(1;2.5)(1.4;2)(1.5;2)
   \end{pspicture} 
 ~
   \begin{pspicture}[shift=-1.8](-2.25,-1.8)(2.25,1.8)
    	\pscircle[fillstyle=solid, fillcolor=diskin,linewidth=1.2pt](0,0){1.5}
    	\degrees[6]
	\rput{0}(1.8;0){\tiny$1$}
	\rput{0}(1.8;1){\tiny$2$}
	\rput{0}(1.8;2){\tiny$3$}
	\rput{0}(1.8;3){\tiny$4$}
	\rput{0}(1.8;4){\tiny$5$}
	\rput{0}(1.8;5){\tiny$6$}
	\pscurve[linecolor=line,linewidth=1.2pt](1.5;0)(1.4;0)(1;0.5)(1.4;1)(1.5;1)
	\pscurve[linecolor=line,linewidth=1.2pt](1.5;5)(1.4;5)(0;0)(1.4;2)(1.5;2)
	\pscurve[linecolor=line,linewidth=1.2pt](1.5;4)(1.4;4)(1;3.5)(1.4;3)(1.5;3)
   \end{pspicture}  
    ~
   \begin{pspicture}[shift=-1.8](-2.25,-1.8)(2.25,1.8)
    	\pscircle[fillstyle=solid, fillcolor=diskin,linewidth=1.2pt](0,0){1.5}
    	\degrees[6]
	\rput{0}(1.8;0){\tiny$1$}
	\rput{0}(1.8;1){\tiny$2$}
	\rput{0}(1.8;2){\tiny$3$}
	\rput{0}(1.8;3){\tiny$4$}
	\rput{0}(1.8;4){\tiny$5$}
	\rput{0}(1.8;5){\tiny$6$}
	\pscurve[linecolor=line,linewidth=1.2pt](1.5;1)(1.4;1)(1;1.5)(1.4;2)(1.5;2)
	\pscurve[linecolor=line,linewidth=1.2pt](1.5;0)(1.4;0)(0;0)(1.4;3)(1.5;3)
	\pscurve[linecolor=line,linewidth=1.2pt](1.5;5)(1.4;5)(1;4.5)(1.4;4)(1.5;4)
   \end{pspicture}  
    ~
   \begin{pspicture}[shift=-1.8](-2.25,-1.8)(2.25,1.8)
    	\pscircle[fillstyle=solid, fillcolor=diskin,linewidth=1.2pt](0,0){1.5}
    	\degrees[6]
	\rput{0}(1.8;0){\tiny$1$}
	\rput{0}(1.8;1){\tiny$2$}
	\rput{0}(1.8;2){\tiny$3$}
	\rput{0}(1.8;3){\tiny$4$}
	\rput{0}(1.8;4){\tiny$5$}
	\rput{0}(1.8;5){\tiny$6$}
	\pscurve[linecolor=line,linewidth=1.2pt](1.5;1)(1.4;1)(1;1.5)(1.4;2)(1.5;2)
	\pscurve[linecolor=line,linewidth=1.2pt](1.5;3)(1.4;3)(1;3.5)(1.4;4)(1.5;4)
	\pscurve[linecolor=line,linewidth=1.2pt](1.5;5)(1.4;5)(1;5.5)(1.4;0)(1.5;0)
   \end{pspicture}  
    ~
   \begin{pspicture}[shift=-1.8](-2.25,-1.8)(2.25,1.8)
    	\pscircle[fillstyle=solid, fillcolor=diskin,linewidth=1.2pt](0,0){1.5}
    	\degrees[6]
	\rput{0}(1.8;0){\tiny$1$}
	\rput{0}(1.8;1){\tiny$2$}
	\rput{0}(1.8;2){\tiny$3$}
	\rput{0}(1.8;3){\tiny$4$}
	\rput{0}(1.8;4){\tiny$5$}
	\rput{0}(1.8;5){\tiny$6$}
	\pscurve[linecolor=line,linewidth=1.2pt](1.5;2)(1.4;2)(1;2.5)(1.4;3)(1.5;3)
	\pscurve[linecolor=line,linewidth=1.2pt](1.5;4)(1.4;4)(1;4.5)(1.4;5)(1.5;5)
	\pscurve[linecolor=line,linewidth=1.2pt](1.5;0)(1.4;0)(1;0.5)(1.4;1)(1.5;1)
   \end{pspicture} 
 $$  
Link patterns can also be drawn by placing the endpoints on a horizontal line such that the $k$ non-intersecting curves lie above it.
To establish convention, the points are numbered in increasing order from left to right.
As an example the link patterns of  $LP_6$ can be drawn as
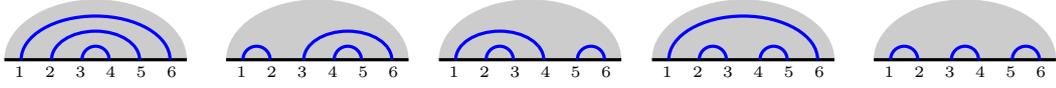
\begin{align*} \psset{unit=0.4}
 \begin{pspicture}(0,-0.5)(7,2)
 	\psellipticarc*[linecolor=diskin,linewidth=1.2pt](3.5,0)(3,2){0}{180}
	\psline[linewidth=1.2pt](0.5,0)(6.5,0)
	\rput{0}(1,-0.4){\tiny$1$}
	\rput{0}(2,-0.4){\tiny$2$}		
	\rput{0}(3,-0.4){\tiny$3$}
	\rput{0}(4,-0.4){\tiny$4$}
	\rput{0}(5,-0.4){\tiny$5$}
	\rput{0}(6,-0.4){\tiny$6$}
	\psellipticarc[linecolor=line,linewidth=1.2pt](3.5,0)(2.5,1.5){0}{180}
	\psellipticarc[linecolor=line,linewidth=1.2pt](3.5,0)(1.5,1){0}{180}
	\psellipticarc[linecolor=line,linewidth=1.2pt](3.5,0)(0.5,0.5){0}{180}
 \end{pspicture} 
 ~
 \begin{pspicture}(0,-0.5)(7,2)
 	\psellipticarc*[linecolor=diskin,linewidth=1.2pt](3.5,0)(3,2){0}{180}
	\psline[linewidth=1.2pt](0.5,0)(6.5,0)
	\rput{0}(1,-0.4){\tiny$1$}
	\rput{0}(2,-0.4){\tiny$2$}		
	\rput{0}(3,-0.4){\tiny$3$}
	\rput{0}(4,-0.4){\tiny$4$}
	\rput{0}(5,-0.4){\tiny$5$}
	\rput{0}(6,-0.4){\tiny$6$}
	\psellipticarc[linecolor=line,linewidth=1.2pt](1.5,0)(0.5,0.5){0}{180}
	\psellipticarc[linecolor=line,linewidth=1.2pt](4.5,0)(1.5,1){0}{180}
	\psellipticarc[linecolor=line,linewidth=1.2pt](4.5,0)(0.5,0.5){0}{180}
 \end{pspicture} 
  \begin{pspicture}(0,-0.5)(7,2)
 	\psellipticarc*[linecolor=diskin,linewidth=1.2pt](3.5,0)(3,2){0}{180}
	\psline[linewidth=1.2pt](0.5,0)(6.5,0)
	\rput{0}(1,-0.4){\tiny$1$}
	\rput{0}(2,-0.4){\tiny$2$}		
	\rput{0}(3,-0.4){\tiny$3$}
	\rput{0}(4,-0.4){\tiny$4$}
	\rput{0}(5,-0.4){\tiny$5$}
	\rput{0}(6,-0.4){\tiny$6$}
	\psellipticarc[linecolor=line,linewidth=1.2pt](5.5,0)(0.5,0.5){0}{180}
	\psellipticarc[linecolor=line,linewidth=1.2pt](2.5,0)(1.5,1){0}{180}
	\psellipticarc[linecolor=line,linewidth=1.2pt](2.5,0)(0.5,0.5){0}{180}
 \end{pspicture} 
 \begin{pspicture}(0,-0.5)(7,2)
 	\psellipticarc*[linecolor=diskin,linewidth=1.2pt](3.5,0)(3,2){0}{180}
	\psline[linewidth=1.2pt](0.5,0)(6.5,0)
	\rput{0}(1,-0.4){\tiny$1$}
	\rput{0}(2,-0.4){\tiny$2$}		
	\rput{0}(3,-0.4){\tiny$3$}
	\rput{0}(4,-0.4){\tiny$4$}
	\rput{0}(5,-0.4){\tiny$5$}
	\rput{0}(6,-0.4){\tiny$6$}
	\psellipticarc[linecolor=line,linewidth=1.2pt](2.5,0)(0.5,0.5){0}{180}
	\psellipticarc[linecolor=line,linewidth=1.2pt](4.5,0)(0.5,0.5){0}{180}
	\psellipticarc[linecolor=line,linewidth=1.2pt](3.5,0)(2.5,1.5){0}{180}
 \end{pspicture} 
 ~
  \begin{pspicture}(0,-0.5)(7,2)
 	\psellipticarc*[linecolor=diskin,linewidth=1.2pt](3.5,0)(3,2){0}{180}
	\psline[linewidth=1.2pt](0.5,0)(6.5,0)
	\rput{0}(1,-0.4){\tiny$1$}
	\rput{0}(2,-0.4){\tiny$2$}		
	\rput{0}(3,-0.4){\tiny$3$}
	\rput{0}(4,-0.4){\tiny$4$}
	\rput{0}(5,-0.4){\tiny$5$}
	\rput{0}(6,-0.4){\tiny$6$}
	\psellipticarc[linecolor=line,linewidth=1.2pt](1.5,0)(0.5,0.5){0}{180}
	\psellipticarc[linecolor=line,linewidth=1.2pt](3.5,0)(0.5,0.5){0}{180}
	\psellipticarc[linecolor=line,linewidth=1.2pt](5.5,0)(0.5,0.5){0}{180}
 \end{pspicture} 
\end{align*}
 respectively.
Due to this form the curves are sometime referred to as arches and a \emph{little arch} is one that connects two consecutive points.  

A Dyck path of length $2k$ is a lattice path from $(0,0)$ to $(2k,0)$ with steps $(1,1)$ called a \emph{rise} and $(1,-1)$ called a \emph{fall}, which never falls below the $x$-axis. 
We denote the set of Dyck paths of length $2k$ by $DP_{2k}$.
As an example $DP_{6}$ consists of the following Dyck paths:
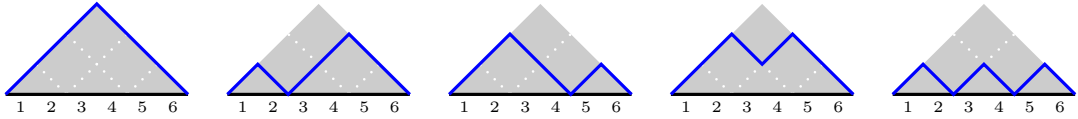
\begin{align*}
 \psset{unit=0.4}
 \begin{pspicture}(-0.5,-0.5)(6.5,3)
\psline*[linecolor=diskin,linewidth=1pt](0,0)(3,3)(6,0)
\psline[linewidth=1pt,linestyle=dotted,linecolor=white](2,0)(1,1)
\psline[linewidth=1pt,linestyle=dotted,linecolor=white](4,0)(2,2)
\psline[linewidth=1pt,linestyle=dotted,linecolor=white](2,0)(4,2)
\psline[linewidth=1pt,linestyle=dotted,linecolor=white](4,0)(5,1)
\psline[linewidth=1.2pt](0,0)(6,0)
\psline[linewidth=1.2pt,linecolor=line](0,0)(1,1)(2,2)(3,3)(4,2)(5,1)(6,0)
	\rput{0}(0.5,-0.4){\tiny$1$}
	\rput{0}(1.5,-0.4){\tiny$2$}		
	\rput{0}(2.5,-0.4){\tiny$3$}
	\rput{0}(3.5,-0.4){\tiny$4$}
	\rput{0}(4.5,-0.4){\tiny$5$}
	\rput{0}(5.5,-0.4){\tiny$6$}
 \end{pspicture} 
 ~
 \begin{pspicture}(-0.5,-0.5)(6.5,3)
\psline*[linecolor=diskin,linewidth=1pt](0,0)(3,3)(6,0)
\psline[linewidth=1pt,linestyle=dotted,linecolor=white](2,0)(1,1)
\psline[linewidth=1pt,linestyle=dotted,linecolor=white](4,0)(2,2)
\psline[linewidth=1pt,linestyle=dotted,linecolor=white](2,0)(4,2)
\psline[linewidth=1pt,linestyle=dotted,linecolor=white](4,0)(5,1)
\psline[linewidth=1.2pt](0,0)(6,0)
\psline[linewidth=1.2pt,linecolor=line](0,0)(1,1)(2,0)(3,1)(4,2)(5,1)(6,0)
	\rput{0}(0.5,-0.4){\tiny$1$}
	\rput{0}(1.5,-0.4){\tiny$2$}		
	\rput{0}(2.5,-0.4){\tiny$3$}
	\rput{0}(3.5,-0.4){\tiny$4$}
	\rput{0}(4.5,-0.4){\tiny$5$}
	\rput{0}(5.5,-0.4){\tiny$6$}
 \end{pspicture} 
 ~
 \begin{pspicture}(-0.5,-0.5)(6.5,3)
\psline*[linecolor=diskin,linewidth=1pt](0,0)(3,3)(6,0)
\psline[linewidth=1pt,linestyle=dotted,linecolor=white](2,0)(1,1)
\psline[linewidth=1pt,linestyle=dotted,linecolor=white](4,0)(2,2)
\psline[linewidth=1pt,linestyle=dotted,linecolor=white](2,0)(4,2)
\psline[linewidth=1pt,linestyle=dotted,linecolor=white](4,0)(5,1)
\psline[linewidth=1.2pt](0,0)(6,0)
\psline[linewidth=1.2pt,linecolor=line](0,0)(1,1)(2,2)(3,1)(4,0)(5,1)(6,0)
	\rput{0}(0.5,-0.4){\tiny$1$}
	\rput{0}(1.5,-0.4){\tiny$2$}		
	\rput{0}(2.5,-0.4){\tiny$3$}
	\rput{0}(3.5,-0.4){\tiny$4$}
	\rput{0}(4.5,-0.4){\tiny$5$}
	\rput{0}(5.5,-0.4){\tiny$6$}
 \end{pspicture} 
  ~
\begin{pspicture}(-0.5,-0.5)(6.5,3)
\psline*[linecolor=diskin,linewidth=1pt](0,0)(3,3)(6,0)
\psline[linewidth=1pt,linestyle=dotted,linecolor=white](2,0)(1,1)
\psline[linewidth=1pt,linestyle=dotted,linecolor=white](4,0)(2,2)
\psline[linewidth=1pt,linestyle=dotted,linecolor=white](2,0)(4,2)
\psline[linewidth=1pt,linestyle=dotted,linecolor=white](4,0)(5,1)
\psline[linewidth=1.2pt](0,0)(6,0)
\psline[linewidth=1.2pt,linecolor=line](0,0)(1,1)(2,2)(3,1)(4,2)(5,1)(6,0)
	\rput{0}(0.5,-0.4){\tiny$1$}
	\rput{0}(1.5,-0.4){\tiny$2$}		
	\rput{0}(2.5,-0.4){\tiny$3$}
	\rput{0}(3.5,-0.4){\tiny$4$}
	\rput{0}(4.5,-0.4){\tiny$5$}
	\rput{0}(5.5,-0.4){\tiny$6$}
 \end{pspicture} 
 ~
 \begin{pspicture}(-0.5,-0.5)(6.5,3)
\psline*[linecolor=diskin,linewidth=1pt](0,0)(3,3)(6,0)
\psline[linewidth=1pt,linestyle=dotted,linecolor=white](2,0)(1,1)
\psline[linewidth=1pt,linestyle=dotted,linecolor=white](4,0)(2,2)
\psline[linewidth=1pt,linestyle=dotted,linecolor=white](2,0)(4,2)
\psline[linewidth=1pt,linestyle=dotted,linecolor=white](4,0)(5,1)
\psline[linewidth=1.2pt](0,0)(6,0)
\psline[linewidth=1.2pt,linecolor=line](0,0)(1,1)(2,0)(3,1)(4,0)(5,1)(6,0)
	\rput{0}(0.5,-0.4){\tiny$1$}
	\rput{0}(1.5,-0.4){\tiny$2$}		
	\rput{0}(2.5,-0.4){\tiny$3$}
	\rput{0}(3.5,-0.4){\tiny$4$}
	\rput{0}(4.5,-0.4){\tiny$5$}
	\rput{0}(5.5,-0.4){\tiny$6$}
 \end{pspicture} 
\end{align*}
A Dyck path can also be encoded by a string of $2k$ numbers $(a_1,\ldots, a_{2k})$ where $a_j$ for $1\leq j \leq 2k$ is the height of the path after step $j$.
Furthermore, for a Dyck path $L$ we define its content $|L|$ to be the number of boxes within the gray triangle that lie above the path. 
For example, for the last Dyck path $L \in DP_{6}$ in the example above, $|L|=3$.

There exists a bijection between $LP_{2k}$ and $DP_{2k}$.
To go from link patterns to Dyck paths, consider the link pattern drawn on a horizontal line and traverse along the line from left to right.
Each point $i$ that is the beginning/end of an arch corresponds to a rise/fall at step $i$ in the Dyck path.
To go from Dyck paths to link patterns, for each rise draw the start of an arch and for each fall an end, and then complete the diagram by connecting a start with an end such that the arches do not intersect.
For example, the ordered lists of elements in $LP_6$ and $DP_6$ are the same under the bijective correspondence.

The bijection allows us to establish a containment ordering on link patterns.
For two link patterns $L,L' \in LP_{2k}$, we say that $L$ contains $L'$ if the entire corresponding Dyck path of $L'$ can be drawn along or below the Dyck path of $L$.
More formally, let $L$ and $L'$ correspond to the Dyck paths $(a_1,\ldots a_{2k})$ and $(b_1,\ldots b_{2k})$, respectively.
Then, $L$ contains $L'$ if $a_j \geq b_j$ for all $1\leq j \leq 2k$.
As an example, in the list of Dyck paths in $DP_6$ the first path contains all  other paths.
Furthermore, note that if $L$ contains $L'$ then $|L| \leq |L'|$.

Using the disc diagrams, one can define the action of $\mathcal{TL}_{2k}$ on $\C[LP_{2k}]$ similarly to that on $\C[\mathcal{L}_{2k}]$, one just ignores the puncture so we do not have the loop removal rule for non-contractible loops.
Using the horizontal line diagrams is more suitable when discussing the action of the finite Temperley-Lieb algebra $\mathcal{TL}^f_{2k}$.
The induced action of $\mathcal{TL}^f_{2k}$ on $\C[DP_{2k}]$ can then be described as follows.

At step $i$ for  $1 \leq i \leq 2k-1$, a Dyck path can have one of three different local situations;
\begin{enumerate}
\item Steps $i,i+1$ form a local maximum, i.e. a rise followed by a fall;
\item Steps $i,i+1$ form a local minimum, i.e. a fall followed by a rise;
\item Steps $i,i+1$ form a slope, i.e. two consecutive rises or falls.
\end{enumerate}
If steps $i,i+1$ form a local maximum, then the action of $e_i$ acts as a scalar, leaving the path unchanged and multiplying by a factor $-(t^{\frac{1}{2}}+t^{-\frac{1}{2}})$ (line 1, Figure \ref{dyckaction}).
If steps $i,i+1$ form a local minimum, then $e_i$ changes it into a local maximum (line 2, Figure \ref{dyckaction}).
For a slope, if it is two consecutive rises, say with heights $a_i=m$ and $a_{i+1}=m+1$, then let $j>i+1$ be the first step that is a fall with $a_j=m$.
The action of $e_i$ then changes step $i+1$ into fall and $j$ into a rise, creating a local maximum at $i,i+1$ and decreasing the height of the path between $i$ and $j$ by two (line 3, Figure \ref{dyckaction}). 
This decrease in height shifts the internal path down, and we refer to it as a \emph{collapse}.
If the slope is downwards with height $a_i=m,a_{i+1}=m-1$,  let $j<i$ be the last rise with $a_j=m+1$.
Then, the action of $e_i$ changes step $i$ and $j$ to a rise and fall, respectively.
This creates a local maximum at $i,i+1$ and causes a collapse decreasing $a_l$ by two for $j\leq l<i$. 
Note that a collapse leads to a smaller Dyck path in the inclusion order.

Figure \ref{dyckaction} gives a diagrammatic definition of the action on Dyck paths.
The dotted frame indicates the section of the paths where they differ and the dotted line in the third mapping represents a Dyck path of length $j-i-2$. 
The case for two consecutive falls is the same as the third line but with the diagrams reflected across a vertical line in the middle of the diagrams.

\begin{figure}[t]
    \centering
    \begin{align*}
\psset{unit=0.7}
\begin{pspicture}[shift=-1.5](-0.5,-1.5)(6.5,1.5)
\psline*[linecolor=diskin,linewidth=1pt](0,-1.5)(0,1.5)(6,1.5)(6,-1.5)
\psline[linestyle=dotted,linewidth=1pt](0,-1.5)(0,1.5)(6,1.5)(6,-1.5)(0,-1.5)
\psline[linewidth=1.2pt,linecolor=line,linestyle=dashed](0,0)(2,0)
\psline[linewidth=1.2pt,linecolor=line,linestyle=dashed](4,0)(6,0)
\psline[linewidth=1.2pt,linecolor=line](2,0)(3,1)(4,0)
\psline[linewidth=1pt,linestyle=dotted,linecolor=white](2,0)(3,-1)(4,0)
	\rput{0}(2.5,0){\tiny$i$}
	\rput{0}(3.4,0){\tiny$i\!+\!1$}		
 \end{pspicture} 
 &\overset{e_i}{\longmapsto }&
 -(t^{\frac{1}{2}}+t^{-\frac{1}{2}}) 
& \psset{unit=0.7}\begin{pspicture}[shift=-1.5](-0.5,-1.5)(6.5,1.5)
\psline*[linecolor=diskin,linewidth=1pt](0,-1.5)(0,1.5)(6,1.5)(6,-1.5)
\psline[linestyle=dotted,linewidth=1pt](0,-1.5)(0,1.5)(6,1.5)(6,-1.5)(0,-1.5)
\psline[linewidth=1.2pt,linecolor=line,linestyle=dashed](0,0)(2,0)
\psline[linewidth=1.2pt,linecolor=line,linestyle=dashed](4,0)(6,0)
\psline[linewidth=1.2pt,linecolor=line](2,0)(3,1)(4,0)
\psline[linewidth=1pt,linestyle=dotted,linecolor=white](2,0)(3,-1)(4,0)
	\rput{0}(2.5,0){\tiny$i$}
	\rput{0}(3.4,0){\tiny$i\!+\!1$}		
 \end{pspicture}  
 \\ 
 \psset{unit=0.7}
\begin{pspicture}[shift=-1.5](-0.5,-1.5)(6.5,1.5)
\psline*[linecolor=diskin,linewidth=1pt](0,-1.5)(0,1.5)(6,1.5)(6,-1.5)
\psline[linestyle=dotted,linewidth=1pt](0,-1.5)(0,1.5)(6,1.5)(6,-1.5)(0,-1.5)
\psline[linewidth=1.2pt,linecolor=line,linestyle=dashed](0,0)(2,0)
\psline[linewidth=1.2pt,linecolor=line,linestyle=dashed](4,0)(6,0)
\psline[linewidth=1.2pt,linecolor=line](2,0)(3,-1)(4,0)
\psline[linewidth=1pt,linestyle=dotted,linecolor=white](2,0)(3,1)(4,0)
	\rput{0}(2.5,0){\tiny$i$}
	\rput{0}(3.4,0){\tiny$i\!+\!1$}
 \end{pspicture} 
&\overset{e_i}{\longmapsto }
  &&\psset{unit=0.7}
\begin{pspicture}[shift=-1.5](-0.5,-1.5)(6.5,1.5)
\psline*[linecolor=diskin,linewidth=1pt](0,-1.5)(0,1.5)(6,1.5)(6,-1.5)
\psline[linestyle=dotted,linewidth=1pt](0,-1.5)(0,1.5)(6,1.5)(6,-1.5)(0,-1.5)
\psline[linewidth=1.2pt,linecolor=line,linestyle=dashed](0,0)(2,0)
\psline[linewidth=1.2pt,linecolor=line,linestyle=dashed](4,0)(6,0)
\psline[linewidth=1.2pt,linecolor=line](2,0)(3,1)(4,0)
\psline[linewidth=1pt,linestyle=dotted,linecolor=white](2,0)(3,-1)(4,0)
	\rput{0}(2.5,0){\tiny$i$}
	\rput{0}(3.4,0){\tiny$i\!+\!1$}		
 \end{pspicture}  
 \\ 
 \psset{unit=0.7}
\begin{pspicture}[shift=-1.5](-0.5,-1.5)(7.5,1.5)
\psline*[linecolor=diskin,linewidth=1pt](0,-1.5)(0,1.5)(7,1.5)(7,-1.5)
\psline[linestyle=dotted,linewidth=1pt](0,-1.5)(0,1.5)(7,1.5)(7,-1.5)(0,-1.5)
\psline[linewidth=1.2pt,linecolor=line,linestyle=dashed](0,-1)(1,-1)
\psline[linewidth=1.2pt,linecolor=line,linestyle=dashed](6,0)(7,0)
\psline[linewidth=1.2pt,linecolor=line](1,-1)(2,0)(3,1)
\psline[linewidth=1.2pt,linecolor=line](5,1)(6,0)
\psline[linewidth=1.2pt,linecolor=line,linestyle=dotted](3,1)(5,1)
\psline[linewidth=1pt,linestyle=dotted,linecolor=white](2,0)(3,-1)
\psline[linewidth=1pt,linestyle=dotted,linecolor=white](5,-1)(6,0)
	\rput{0}(1.5,0){\tiny$i$}
	\rput{0}(2.6,0){\tiny$i\!+\!1$}
	\rput{0}(5.5,0){\tiny$j$}			
 \end{pspicture}  
&\overset{e_i}{\longmapsto }
  &&\psset{unit=0.7}
  \begin{pspicture}[shift=-1.5](-0.5,-1.5)(7.5,1.5)
\psline*[linecolor=diskin,linewidth=1pt](0,-1.5)(0,1.5)(7,1.5)(7,-1.5)
\psline[linestyle=dotted,linewidth=1pt](0,-1.5)(0,1.5)(7,1.5)(7,-1.5)(0,-1.5)
\psline[linewidth=1.2pt,linecolor=line,linestyle=dashed](0,-1)(1,-1)
\psline[linewidth=1.2pt,linecolor=line,linestyle=dashed](6,0)(7,0)
\psline[linewidth=1.2pt,linecolor=line](1,-1)(2,0)(3,-1)
\psline[linewidth=1.2pt,linecolor=line](5,-1)(6,0)
\psline[linewidth=1.2pt,linecolor=line,linestyle=dotted](3,-1)(5,-1)
\psline[linewidth=1pt,linestyle=dotted,linecolor=white](2,0)(3,1)
\psline[linewidth=1pt,linestyle=dotted,linecolor=white](5,1)(6,0)
	\rput{0}(1.5,0){\tiny$i$}
	\rput{0}(2.6,0){\tiny$i\!+\!1$}
	\rput{0}(5.5,0){\tiny$j$}		
 \end{pspicture}  
\end{align*}
    \caption{The action of $e_i$ on Dyck paths}\label{dyckaction}
\end{figure}
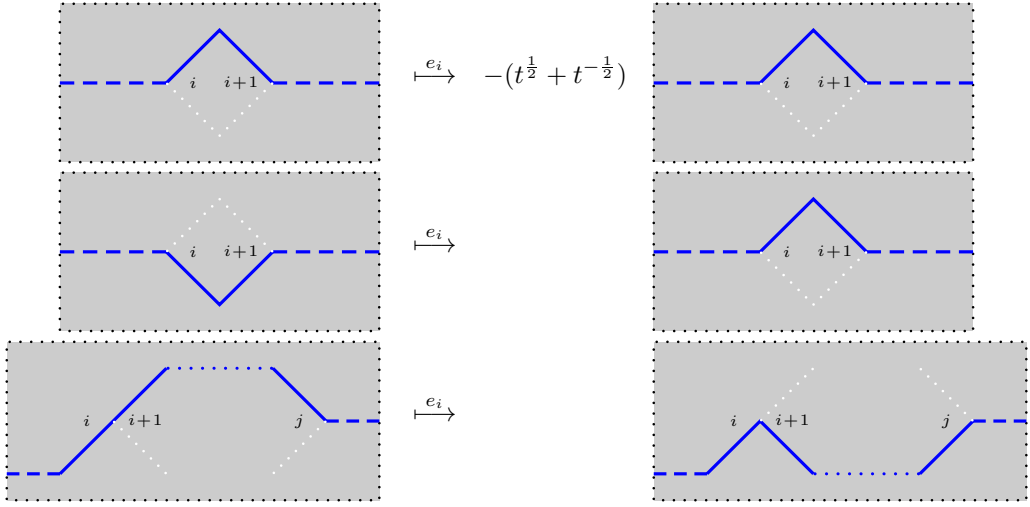
 
\subsection{Link patterns}
Let $L_{0}\in DP_{2k}$ denote the Dyck path with $k$ rises followed by $k$ falls i.e. $(1,2,\ldots,k,k-1,\ldots,0)$.
Note that $L_{0}$ contains all Dyck paths in $DP_{2k}$ and $|L_0|=0$. We show that 
the solution $g^{(2k)}(\mathbf{z}) \in \textup{Sol}_{2k}(\C[LP_{2k}];q,c_{2k})$ is determined by its base component $g^{(2k)}_{L_{0}}(\mathbf{z})$.

We prove this by showing that if $g^{(2k)}_{L_{0}}(\mathbf{z}) \equiv 0 $, then we have $  g^{(2k)}_L(\mathbf{z}) \equiv 0$ for all $L \in DP_{2k}$.
This is done using the qKZ equations \eqref{recursion}, which give for $1\leq j <2k$, 
\begin{align}\label{dpcomponents}
 	g_L^{(n)}(\mathbf{z})&-b(z_{j+1}/z_j)g_L^{(n)}(s_j\mathbf{z})=\sum_{L^\prime\in LP_{2k}:\, e_jL^\prime\sim L}\gamma_{L^\prime,L}^{(j)}a(z_{j+1}/z_j)g_{L^\prime}^{(n)}(s_j\mathbf{z}),
\end{align}
where $e_jL^\prime\sim L$ means that $L$ is obtained from $e_jL^\prime$ by removing the loops in $e_jL^\prime$ (there is in fact at most one loop). The coefficient $\gamma_{L^\prime,L}^{(j)}$ is defined by
$$  \gamma^{(j)}_{L^\prime,L} = \begin{cases} 
						-(t^{\frac{1}{2}} +t^{-\frac{1}{2}})  &\text{ if } e_jL^\prime \text{ has a null-homotopic loop,}  \\
						1& \text{ otherwise}.
						\end{cases}
$$

We begin with the inductive hypothesis,
$$  g^{(2k)}_L(\mathbf{z}) \equiv 0 \text{ if } |L|\leq m,$$
where $m \in \Z_{\geq 0}$.
Now consider a Dyck path $L$ such that $|L| =m$ with a local maximum at, say, step $i$ with $a_i>1$
(if such a local maximum does not exist, then $L$ is the unique Dykh path with maximal 
content $|L|$, and hence there is nothing to prove).
We use equation \eqref{dpcomponents} for $j=i$ and examine the pre-images  $L'$ in the sum on the right-hand side.
We find that besides $L$ there is only one pre-image that is contained by $L$.
This is the pre-image that has a local minimum turned into a local maximum by the action of $e_i$.
Let us denote this particular Dyck path by $N$.
Switching a local minimum to a local maximum is equivalent to removing a box, so we have $|N| = m+1$.
Furthermore, all other pre-images $L' \neq N$ contain $L$ so $|L'| \leq m$ and $g^{(2k)}_{L'}(\mathbf{z}) \equiv 0$.
Thus, $g^{(2k)}_{N}(\mathbf{z}) \equiv 0$ with $|N|=m+1$ by \eqref{dpcomponents}.
Since $|L_{0}|=0$, it provides the base case of the induction and determines all other components.

\begin{remark} \label{collapsing}
The algorithm of this proof can be viewed as collapsing the local maxima till we end up at the last component, which has $k$ local maxima with height $1$. 
Since a Dyck path cannot fall below the $x$-axis we cannot collapse a local maximum with a height $1$.
Thus, the algorithm never uses \eqref{dpcomponents} at $j$ if the height $a_j$ is one.
This is an important remark for the proofs that follow. 
\end{remark}

\begin{remark}
The Dyck path $L_{0}$ corresponds to the link pattern that connects point $i$ with $2k-i+1$.
The same arguments used in this proof can be found in \cite{Di-Francesco:2007aa} where they prove the unique solution for the model with reflecting boundaries. 
\end{remark}

\subsection{Punctured link patterns} \label{plp}
Here we present the proof to Lemma \ref{unique1}.
 Let $L_\cap$  denote the link pattern 
$$
  \psset{unit=0.8cm}
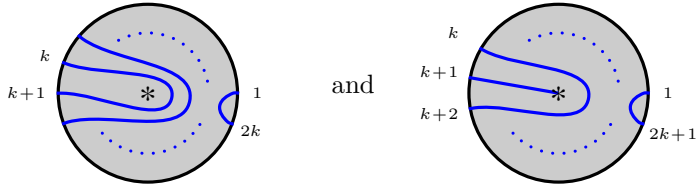
  \begin{pspicture}[shift=-1.8](-2.25,-1.8)(2.25,1.8)
    	\pscircle[fillstyle=solid, fillcolor=diskin,linewidth=1.2pt](0,0){1.5}
    	\degrees[360]
	\psdot[dotstyle=asterisk,dotscale=1.5](0,0)
	\rput{0}(1.8;0){\tiny$1$}
	\rput{0}(1.8;-20){\tiny$2k$}
	\rput{0}(1.8;160){\tiny$k$}
	\rput{0}(2;180){\tiny$k\!+\!1$}
	\pscurve[linecolor=line,linewidth=1.2pt](1.5;0)(1.4;0)(1.2;-10)(1.4;-20)(1.5;-20)
	\pscurve[linecolor=line,linewidth=1.2pt](1.5;160)(1.4;160)(-0.4;170)(1.4;180)(1.5;180)
	\pscurve[linecolor=line,linewidth=1.2pt](1.5;140)(1.4;140)(-0.7;170)(1.4;200)(1.5;200)
	\psarc[linestyle=dotted,linecolor=line,linewidth=1.2pt](0,0){1}{10}{120}
	\psarc[linestyle=dotted,linecolor=line,linewidth=1.2pt](0,0){1}{220}{-30}
   \end{pspicture}   
   \hspace{0.5cm} \text{ and }\hspace{0.5cm} 
  \begin{pspicture}[shift=-1.8](-2.25,-1.8)(2.25,1.8)
    	\pscircle[fillstyle=solid, fillcolor=diskin,linewidth=1.2pt](0,0){1.5}
    	\degrees[360]
	\rput{0}(1.8;0){\tiny$1$}
	\rput{0}(2;-20){\tiny$2k\!+\!1$}
	\rput{0}(2;150){\tiny$k$}
	\rput{0}(2;170){\tiny$k\!+\!1$}
	\rput{0}(2;190){\tiny$k\!+\!2$}
	\pscurve[linecolor=line,linewidth=1.2pt](1.5;0)(1.4;0)(1.2;-10)(1.4;-20)(1.5;-20)
	\pscurve[linecolor=line,linewidth=1.2pt](1.5;150)(1.4;150)(-0.5;170)(1.4;190)(1.5;190)
	\psline[linecolor=line,linewidth=1.2pt](1.5;170)(0;0)
	\psarc[linestyle=dotted,linecolor=line,linewidth=1.2pt](0,0){1}{10}{120}
	\psarc[linestyle=dotted,linecolor=line,linewidth=1.2pt](0,0){1}{220}{-30}
	\psdot[dotstyle=asterisk,dotscale=1.5](0,0)
   \end{pspicture}   
   $$
in $\mathcal{L}_{2k}$ and $\mathcal{L}_{2k+1}$, respectively. 
A little arch in a punctured link patterns is an arch connecting points $j,j+1$ that does not contain the puncture.
For example $L_\cap$ only has one little arch, it connects points $2k+1 $ to 1.

We show that the solution $g^{(n)}(\mathbf{z}) \in \textup{Sol}_{n}(\C[\mathcal{L}_{n}];q,c_n)$ is determined by its base component $g^{(n)}_{L_\cap}(\mathbf{z})$.
The qKZ equations written component-wise are 
\begin{align}
g_L^{(n)}(\mathbf{z})&-b(z_{i+1}/z_i)g_L^{(n)}(s_i\mathbf{z})=\sum_{L^\prime\in\mathcal{L}_n:\, e_iL^\prime\sim L}\gamma_{L^\prime,L}^{(i)}a(z_{i+1}/z_i)g_{L^\prime}^{(n)}(s_i\mathbf{z}),\label{qkzp1}\\
g_{L}^{(n)}(\mathbf{z})&=c^{-1}g_{\rho^{-1}L}^{(n)}(z_2,\ldots,z_n,q^{-1}z_1),\label{qkzp2}
\end{align}
where $e_iL^\prime\sim L$ means that $L$ is obtained from $e_iL^\prime$ by removing the loops in $e_iL^\prime$ (there is in fact at most one loop), see \eqref{recursion}. The coefficient $\gamma_{L^\prime,L}^{(i)}$ is 
$$  \gamma^{(i)}_{L^\prime,L} = \begin{cases} 
						-\bigl(t^{\frac{1}{2}} +t^{-\frac{1}{2}}\bigr)  &\text{ if } e_iL^\prime \text{ has a null-homotopic loop,}  \\
						t^{\frac{1}{4}} +t^{-\frac{1}{4}}  &\text{ if } e_iL^\prime \text{ has a non null-homotopic loop,} \\
						1& \text{ otherwise}.
						\end{cases}
$$
We treat the even and odd case separately.

\subsubsection{The case $n=2k$}

Let $LP^{(*,j)}_{2k}$ denote the set of punctured link patterns in $\mathcal{L}_{2k}$ such that the puncture could be connected to a point on the boundary between points $j$ and $j+1$ (modulo $2k$) without crossing a line. 
Then, $\mathcal{L}_{2k} = \bigcup_{j=1}^{2k} LP^{(*,j)}_{2k}$ (not necessarily disjoint) and  $\rho : LP^{(*,j)}_{2k} \rightarrow  LP^{(*,j+1)}_{2k}$.
Note that $L_\cap$ is in $LP^{(*,k)}_{2k}$ and if we define $L_{2k}:= \rho^k \cdot L_\cap$ then $L_{2k} \in  LP^{(*,2k)}_{2k}$.
Define a bijection from $LP^{(*,2k)}_{2k}$ to $LP_{2k}$ by simply removing the puncture.
This mapping preserves the action of $\mathcal{TL}^f_{2k}$.
Furthermore, it maps $L_{2k} \in LP^{(*,2k)}_{2k}$ to $L_0 \in LP_{2k}$.

To prove that $g^{(2k)}_{L_\cap}(\mathbf{z})$ determines the solution we have the following steps.
First, if $g^{(2k)}_{L_\cap}(\mathbf{z}) \equiv 0$ then by using equation \eqref{qkzp2} $k$ times we have $g^{(2k)}_{L_{2k}}(\mathbf{z}) \equiv 0$.
Second, by the mapping from $LP^{(*,2k)}_{2k}$ to $LP_{2k}$ and the proof on $LP_{2k}$ we have $g^{(2k)}_{L}(\mathbf{z}) \equiv 0$ for all $L \in LP^{(*,2k)}_{2k}$.
Last, we use equation \eqref{qkzp2} to show that if $g^{(2k)}_{L}(\mathbf{z}) \equiv 0$ for all $L \in LP^{(*,i)}_{2k}$ then $g^{(2k)}_{L'}(\mathbf{z}) \equiv 0$ for all $L' \in LP^{(*,i+1)}_{2k}$.

There is one key subtlety that we have ignored, which we point out and address.
There is a difference between equations \eqref{dpcomponents} and \eqref{qkzp1}; in the latter equation the pre-images are in $\mathcal{L}_{2k}$ and not just $LP_{2k}$. 
So when determining all the components for $L \in LP^{(*,2k)}_{2k}$ we must check that all the pre-images $L'$ are also in $LP^{(*,2k)}_{2k}$.
Recall Remark \ref{collapsing};
each step of the algorithm is on a local maximum, which corresponds to a little arch in a link pattern.
For the equation $e_i\cdot L'=L$, the only case where we have $L' \not \in LP^{(*,2k)}_{2k}$ is if the little arch in $L$ is on the boundary of the domain that contains the puncture. 
Such a little arch corresponds to a local maximum of height 1.
Recalling again Remark \ref{collapsing}, we do not collapse such local maxima, and therefore we do not have this case, and can conclude all pre-images are in $LP^{(*,2k)}_{2k}$.
If  $L \in LP^{(*,2k)}_{2k}$ has a little arch $(i,i+1)$ on the boundary of the domain containing the puncture, then $L' \not \in LP^{(*,2k)}_{2k}$  is the link pattern identical to $L$ but with the puncture inside the little arch.
\subsubsection{The case $n=2k+1$}
Let $LP^{(*,j)}_{2k+1}$ denote the set of punctured link patterns in $\mathcal{L}_{2k+1}$ such that the defect line connects the puncture  to point $j$ on the boundary. 
Then, we have $\mathcal{L}_{2k+1} = \bigsqcup_{j=1}^{2k+1} LP^{(*,j)}_{2k+1}$ and  $\rho : LP^{(*,j)}_{2k+1} \rightarrow  LP^{(*,j+1)}_{2k+1}$.
Note that $L_\cap$ is in $LP^{(*,k+1)}_{2k}$ and if we define $L_{2k+1}:= \rho^k \cdot L_\cap$ then $L_{2k+1} \in  LP^{(*,2k+1)}_{2k+1}$.
Define a bijection from $LP^{(*,2k+1)}_{2k+1}$ to $LP_{2k}$ by simply removing the defect line, puncture and boundary point $2k+1$.
This mapping preserves the action of $\mathcal{TL}^f_{2k}$.
Furthermore, it maps $L_{2k+1} \in LP^{(*,2k+1)}_{2k+1}$ to $L_0 \in LP_{2k}$.

Now to prove that $g^{(2k+1)}_{L_\cap}(\mathbf{z})$ determines the solution we have the following steps.
First, if $g^{(2k+1)}_{L_\cap}(\mathbf{z}) \equiv 0$ then by using equation \eqref{qkzp2} $k$ times we have $g^{(2k+1)}_{L_{2k+1}}(\mathbf{z}) \equiv 0$.
Second, by the mapping from $LP^{(*,2k+1)}_{2k+1}$ to $LP_{2k}$  and the proof on $LP_{2k}$ we have $g^{(2k)}_{L}(\mathbf{z}) \equiv 0$ for all $L \in LP^{(*,2k+1)}_{2k+1}$.
Last, by equation \eqref{qkzp2},  if $g^{(2k+1)}_{L}(\mathbf{z}) \equiv 0$ for all $L \in LP^{(*,i)}_{2k+1}$ then $g^{(2k+1)}_{L'}(\mathbf{z}) \equiv 0$ for all $L' \in LP^{(*,i+1)}_{2k+1}$.

The same subtle issue occurs in this case, and the argument is identical.
The only case a pre-image $L'$ is not in $LP^{(*,2k+1)}_{2k+1}$ is when the little arch in $L$ is on the boundary of the domain that contains the puncture.
Such a little arch corresponds to a local maximum of height 1.
Again, recalling Remark \ref{collapsing}, we do not collapse such local maxima, therefore we do not have this case, and can conclude all pre-images are in $LP^{(*,2k+1)}_{2k+1}$.
If  $L \in LP^{(*,2k+1)}_{2k+1}$ has a little arch $(i,i+1)$ on the boundary of the domain containing the puncture, then $L' \not \in LP^{(*,2k+1)}_{2k+1}$  is the link pattern with either point $i$ or $i+1$ connected to the puncture and the other to the point $2k+1$.

\subsection{The restricted module $V^{\nu_n}_{n+1}$}
Here we present the proof to Lemma \ref{unique2}.
Let $g^{(n)}(\mathbf{z}) \in \textup{Sol}_n(V^{\nu_n}_{n+1};q,c_n)$.
The qKZ equations associated with the representation $V^{\nu_n}_{n+1}$ written component-wise are,
\begin{align}
g_L^{(n)}(\mathbf{z})-b(z_{j+1}/z_j)g_L^{(n)}(s_j\mathbf{z})=\sum_{L^\prime\in\mathcal{L}_{n+1}:\, e_jL^\prime\sim L}\gamma_{L^\prime,L}^{(j)}a(z_{j+1}/z_j)g_{L^\prime}^{(n)}(s_j\mathbf{z}), \label{rqkz1}\\
 t^{-\frac{1}{4}} \sum_{L^\prime\in\mathcal{L}_{n+1}:\, e_nL^\prime\sim L} 
 \gamma_{L^\prime,L}^{(n)}g^{(n)}_{L'}(z_2 ,\ldots,z_{n},q^{-1} z_1) + t^{\frac{1}{4}} g^{(n)}_L(z_2 ,\ldots,z_{n},q^{-1} z_1) =c_n^{-1} g_{\rho \cdot L}^{(n)}(\mathbf{z})
\label{rqkz2}
\end{align}
for $1\leq j<n$.
It is important to note that the link patterns are in $\mathcal{L}_{n+1}$ but the first equation is only for $1\leq j < n$; there is one equation less than in the previous cases. Note furthermore that \eqref{rqkz2} follows from the fact that $\mathcal{I}_n(\rho)=\rho(t^{-\frac{1}{4}}e_n+t^{\frac{1}{4}})$.

The proof for the even and odd case are treated separately.

\subsubsection{The case $n=2k$}
Note that for the case $n=2k$ the link patterns are in $\mathcal{L}_{2k+1}$.
Recall from subsection \ref{plp} the definitions for $L_\cap,L_{2k+1} \in \mathcal{L}_{2k+1}$ and $LP^{(*,j)}_{2k+1}$.
We show the solution $g^{(2k)}(\mathbf{z})$ is determined by its base component $g_{L_\cap}^{(2k)}(\mathbf{z})$ in three steps.

First, consider equation $\eqref{rqkz2}$ for $L=L_\cap$.
Since $L_\cap$ does not have a little arch connecting $(2k,2k+1)$ there is no pre-image $L'$ such that $e_{2k}\cdot L'\sim L_{\cap}$.
Therefore, there are no terms in the sum (over $L'$)  on the left-hand side of the equation, and if $g_{L_\cap}^{(2k)}(\mathbf{z}) \equiv 0$ then $g_{\rho \cdot L_\cap}^{(2k)}(\mathbf{z}) \equiv 0$.
Continuing this way, we conclude that $g_{\rho^iL_\cap}^{(2k)}(\mathbf{z})\equiv 0$ for all $0\leq i\leq 2k$. But then we have covered all the possible rotations of $L_\cap$, so $g_{\rho^iL_\cap}^{(2k)}(\mathbf{z})\equiv 0$ for all $i\in\mathbb{Z}$. In particular,
$g_{L_{2k+1}^{(2k)}}(\mathbf{z}) \equiv 0$.

The second step is identical to subsection \ref{plp}. 
By the mapping from $LP^{(*,2k+1)}_{2k+1}$ to $LP_{2k}$  and the proof on $LP_{2k}$ we have $g^{(2k)}_{L}(\mathbf{z}) \equiv 0$ for all $L \in LP^{(*,2k+1)}_{2k+1}$. 
The fact that we have one equation less does not play a role here as the mapping from $LP^{(*,2k+1)}_{2k+1}$ to $LP_{2k}$ decreases the size of the link patterns by one. 

The last step is to use equation \eqref{rqkz2}.
However, this is not as simple as subsection \ref{plp} because equation  \eqref{rqkz2}  has an extra term on the left-hand side.
It is a sum over pre-images for the action of $e_n$. We will refer to it as the \emph{pre-image sum}.
Consider equation \eqref{rqkz2} for $L\in LP^{(*,2k)}_{2k+1}$.
Since $L$ has the defect line connected to point $2k$ there are no pre-images $L'$ such that $e_{2k}\cdot L'\sim L$.
Therefore the pre-image-sum of $\eqref{rqkz2}$ does not give a contribution and $g^{(2k)}_{L} (\mathbf{z})  \equiv 0$ because $g^{(2k)}_{\rho \cdot L}(\mathbf{z})  \equiv 0$ as $\rho \cdot L \in LP^{(*,2k+1)}_{2k+1}$.
Now consider \eqref{rqkz2}  for $ L \in LP^{(*,2k+1)}_{2k+1}$.
 By the same argument, the pre-image-sum gives no contribution and $g^{(2k)}_{\rho \cdot L} (\mathbf{z})  \equiv 0$.
 Since each $L' \in LP^{(*,1)}_{2k+1}$ is of the form $\rho \cdot L$ for some $L \in LP^{(*,2k+1)}_{2k+1}$ we have $g^{(2k)}_{L} (\mathbf{z}) \equiv 0$ for all $L \in LP^{(*,1)}_{2k+1}$.

Having shown that $g^{(2k)}_{L}(\mathbf{z})  \equiv 0$ for $L \in LP^{(*,1)}_{2k+1} \sqcup LP^{(*,2k)}_{2k+1}\sqcup LP^{(*,2k+1)}_{2k+1}$ we now use an inductive argument to complete the proof.
The induction hypothesis is 
$$ g^{(2k)}_{L} (\mathbf{z})  \equiv 0 \text{ if } L \in  LP^{(*,2k)}_{2k+1}\sqcup LP^{(*,2k+1)}_{2k+1} \bigsqcup_{1\leq j \leq i} LP^{(*,j)}_{2k+1}$$
for some $1 \leq i <2k$.
Consider equation \eqref{rqkz2} for $L \in LP^{(*,i)}_{2k+1}$.
If $L$ does not have a little arch connecting $(2k,2k+1)$ then we have the same argument used before: there are no pre-images and the pre-image sum does not give a contribution, hence $g^{(2k)}_{L}(\mathbf{z})  \equiv 0$ for $L \in LP^{(*,i+1)}_{2k+1}$.
If $L$ does have a little arch connecting $(2k,2k+1)$ then the pre-image $L'$ must have the defect line connected to points $i$, $2k$ or $2k+1$.
The case that the link pattern $L'$ in the pre-image lies in  $LP^{(*,i)}_{2k+1}$ is obvious. For the other two cases the parts of the pre-images connected to $i, 2k, 2k+1$ are drawn in Figure \ref{fig-preimage2kp1}. It follows that the pre-image $L'$ is in $ LP^{(*,2k)}_{2k+1}\sqcup LP^{(*,2k+1)}_{2k+1} \sqcup  LP^{(*,i)}_{2k+1}$ and the pre-image-sum is equivalently zero.
Hence, in both cases $g^{(2k)}_{\rho\cdot L}(\mathbf{z}) \equiv 0 $ and since $\rho: LP^{(*,i)}_{2k+1} \rightarrow LP^{(*,i+1)}_{2k+1}$ we have $g^{(2k)}_{ L}(\mathbf{z}) \equiv 0$ for all $L \in LP^{(*,i+1)}_{2k+1}$.
This completes the induction and the base case is $i=1$ which was discussed in the previous paragraph.

\begin{figure}[ht]
	\begin{align*}
	 \psset{unit=0.8cm}
  \begin{pspicture}[shift=-1.8](-2.25,-1.8)(2.25,1.8)
    	\pscircle[fillstyle=solid, fillcolor=diskin,linewidth=1.2pt](0,0){1.5}
    	\degrees[360]
	\psdot[dotstyle=asterisk,dotscale=1.5](0,0)
	\rput{0}(1.9;-20){\tiny$2k\!+\!1$}
	\rput{0}(1.9;-40){\tiny$2k$}
	\rput{0}(1.8;70){\tiny$i$}
	\pscurve[linecolor=line,linewidth=1.2pt](1.5;-40)(1.4;-40)(0.3;200)(1.4;70)(1.5;70)
	\psline[linecolor=line,linewidth=1.2pt](0;0)(1.5;-20)
   \end{pspicture}  
    & \in LP^{(*,2k+1)}_{2k+1}, \\
   	 \psset{unit=0.8cm}
  \begin{pspicture}[shift=-1.8](-2.25,-1.8)(2.25,1.8)
    	\pscircle[fillstyle=solid, fillcolor=diskin,linewidth=1.2pt](0,0){1.5}
    	\degrees[360]
	\psdot[dotstyle=asterisk,dotscale=1.5](0,0)
	\rput{0}(1.9;-20){\tiny$2k\!+\!1$}
	\rput{0}(1.9;-40){\tiny$2k$}
	\rput{0}(1.8;70){\tiny$i$}			
	\pscurve[linecolor=line,linewidth=1.2pt](1.5;-20)(1.4;-20)(0.4;45)(1.4;70)(1.5;70)
	\psline[linecolor=line,linewidth=1.2pt](0;0)(1.5;-40)	
   \end{pspicture}  
    & \in LP^{(*,2k)}_{2k+1}
	\end{align*}
	    \caption{Pre-images of $L\in LP^{(*,i)}_{2k+1}$ for the action of $e_{2k}$.}\label{fig-preimage2kp1}
\end{figure}
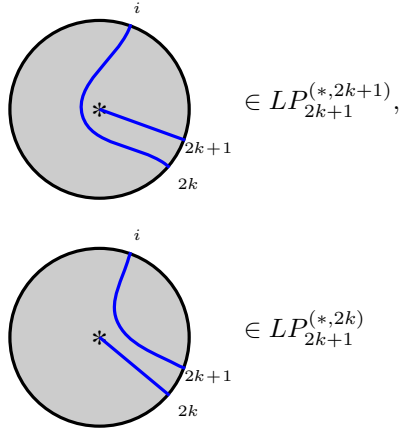

\subsubsection{The case $n=2k-1$ ($k\geq 1$)}
Note that for the case $n=2k-1$ the link patterns are in $\mathcal{L}_{2k}$
Recall from subsection \ref{plp} the definitions for $L_\cap,L_{2k} \in \mathcal{L}_{2k}$ and $LP^{(*,j)}_{2k}$.
We show the solution $g^{(2k-1)}(\mathbf{z})$ is determined by its base component $g_{L_\cap}^{(2k-1)}(\mathbf{z})$ in three steps.

The first step is identical to the case $n=2k$.
The link pattern $L_\cap$ does not have a little arch connecting $(2k-1,2k)$ and neither do the link patterns $\rho^j \cdot L_\cap$ for $1\leq j \leq k$.
So if $g_{L_\cap}^{(2k-1)}(\mathbf{z}) \equiv 0$ then $g_{L_{2k}}^{(2k-1)}(\mathbf{z}) \equiv 0$.

The second step is similar to subsection \ref{plp}; however, there is a new subtle issue to note; we have one equation less.
By the mapping from $LP^{(*,2k)}_{2k}$ to $LP_{2k}$  and the proof on $LP_{2k}$, we have $g^{(2k-1)}_{L}(\mathbf{z}) \equiv 0$ for all $L \in LP^{(*,2k)}_{2k}$.  Indeed, the equations \eqref{rqkz1} with $1\leq j<2k-1$ suffice since a local maximum at step $2k-1$ must have height one,
cf. Remark \ref{collapsing}.

For the final step, consider equation \eqref{rqkz2} for $L\in LP^{(*,2k-1)}_{2k}$.
Since $L$ does not have a little arch connecting points $(2k-1,2k)$, there is no pre-image $L'$ and hence the pre-image sum does not give a contribution.
Therefore, we have $g^{(2k-1)}_{L}(\mathbf{z}) \equiv 0$ for all $L \in LP^{(*,2k-1)}_{2k}$ since $\rho \cdot L \in LP^{(*,2k)}_{2k}$.
Next consider the equation for $L\in LP^{(*,2k-2)}_{2k}$.
If $L$ has a little arch connecting points $(2k-1,2k)$, then $L$ is also in $LP^{(*,2k)}_{2k}$ so $g^{(2k-1)}_L \equiv 0$.
If $L$ does not have a little arch connecting points $(2k-1,2k)$, then there are no pre-images $L'$ and the pre-image sum does not give a contribution, so $g^{(2k-2)}_L (\mathbf{z})\equiv 0$ because $\rho \cdot L \in LP^{(*,2k-1)}_{2k}$.

Now consider equation \eqref{rqkz2} for $L \in LP^{(*,2k)}_{2k}$. 
All the possible pre-images  of $L$ with respect to the action of $e_{2k-1}$ are in $LP^{(*,j)}_{2k}$ for $j=2k-2,2k-1,2k$. 
Therefore, the pre-image sum gives no contribution and $g^{(2k-1)}_L(\mathbf{z}) \equiv 0$ for $L \in LP^{(*,1)}_{2k}$.
We now use an induction argument. 
The hypothesis is 
$$ g^{(2k-1)}_L(\mathbf{z}) \equiv 0 \text{ if } L \in LP^{(*,j)}_{2k} \text{ for } j = 2k-2,2k-1,2k,1,\ldots ,i
$$
for some $1\leq i <2k-2$.
Consider equation \eqref{rqkz2} for $L \in LP^{(*,i)}_{2k}$.
If $L$ does not have a little arch connecting points $(2k-1,2k)$,  then we have the same argument as before: There are no pre-images and the pre-image sum gives no contribution, hence $g^{(2k-1)}_{\rho L} \equiv 0$ and $g_{L^{\prime\prime}}^{(2k-1)}(\mathbf{z})\equiv 0$ for all 
$L^{\prime\prime}\in LP_{2k}^{(*,i+1)}$.
If $L$ does have a little arch connecting points $(2k-1,2k)$ then we examine the pre-images of $L$ with respect to action of $e_{2k-1}$.
We find that 
\[
\{ L' \in \mathcal{L}_n | e_{2k-1}L'\sim L\} \subseteq  LP^{(*,2k-2)}_{2k} \cup LP^{(*,2k-1)}_{2k} \cup LP^{(*,2k)}_{2k} \cup LP^{(*,i)}_{2k}
\] 
(see  Figure \ref{fig-preimage2k}). 
Therefore, the pre-image sum is equivalently zero and since $g^{(2k-1)}_L(\mathbf{z}) \equiv 0$ and $\rho \cdot L \in LP^{(*,i+1)}_{2k}$, we have $g^{(2k-1)}_L(\mathbf{z}) \equiv 0$ for $L \in LP^{(*,i+1)}_{2k}$.
This completes the inductive step. The base case is $i=1$ which was discussed in the beginning of this paragraph.
\begin{figure}[t]
	\begin{align*}
	 \psset{unit=0.8cm}
  \begin{pspicture}[shift=-1.8](-2.25,-1.8)(2.25,1.8)
    	\pscircle[fillstyle=solid, fillcolor=diskin,linewidth=1.2pt](0,0){1.5}
    	\degrees[360]
	\psdot[dotstyle=asterisk,dotscale=1.5](0,0)
	\rput{0}(1.8;-20){\tiny$2k$}
	\rput{0}(1.9;-40){\tiny$2k\!-\!1$}
	\rput{0}(1.8;40){\tiny$j$}	
	\rput{0}(1.8;70){\tiny$i$}
	\rput{0}(1.8;90){\tiny$i\!+\!1$}	
	\rput{0}(1.8;140){\tiny$p$}
	\rput{0}(1.8;180){\tiny$l$}			
	\rput{0}(2;200){\tiny$m$}
	\pscurve[linecolor=line,linewidth=1.2pt](1.5;-40)(1.4;-40)(0.3;80)(1.4;200)(1.5;200)
	\pscurve[linecolor=line,linewidth=1.2pt](1.5;-20)(1.4;-20)(0.5;80)(1.4;180)(1.5;180)
	\pscurve[linecolor=line,linewidth=1.2pt](1.5;70)(1.4;70)(1.2;55)(1.4;40)(1.5;40)
	\pscurve[linecolor=line,linewidth=1.2pt](1.5;90)(1.4;90)(1.2;115)(1.4;140)(1.5;140)	
   \end{pspicture}  
   \text{ or } 
    \psset{unit=0.8cm}
  \begin{pspicture}[shift=-1.8](-2.25,-1.8)(2.25,1.8)
    	\pscircle[fillstyle=solid, fillcolor=diskin,linewidth=1.2pt](0,0){1.5}
    	\degrees[360]
	\psdot[dotstyle=asterisk,dotscale=1.5](0,0)
	\rput{0}(1.8;-20){\tiny$2k$}
	\rput{0}(1.9;-40){\tiny$2k\!-\!1$}
	\rput{0}(1.8;40){\tiny$j$}	
	\rput{0}(1.8;70){\tiny$i$}
	\rput{0}(1.8;90){\tiny$i\!+\!1$}			
	\rput{0}(2;200){\tiny$m$}
	\pscurve[linecolor=line,linewidth=1.2pt](1.5;-40)(1.4;-40)(0.3;80)(1.4;200)(1.5;200)
	\pscurve[linecolor=line,linewidth=1.2pt](1.5;-20)(1.4;-20)(0.8;35)(1.4;90)(1.5;90)
	\pscurve[linecolor=line,linewidth=1.2pt](1.5;70)(1.4;70)(1.2;55)(1.4;40)(1.5;40)	
   \end{pspicture}   
   & & \in LP^{(*,2k-2)}_{2k}, \\
   	 \psset{unit=0.8cm}
  \begin{pspicture}[shift=-1.8](-2.25,-1.8)(2.25,1.8)
    	\pscircle[fillstyle=solid, fillcolor=diskin,linewidth=1.2pt](0,0){1.5}
    	\degrees[360]
	\psdot[dotstyle=asterisk,dotscale=1.5](0,0)
	\rput{0}(1.8;-20){\tiny$2k$}
	\rput{0}(1.9;-40){\tiny$2k\!-\!1$}
	\rput{0}(1.8;40){\tiny$j$}	
	\rput{0}(1.8;70){\tiny$i$}
	\rput{0}(1.8;90){\tiny$i\!+\!1$}	
	\rput{0}(1.8;140){\tiny$p$}			
	\pscurve[linecolor=line,linewidth=1.2pt](1.5;-40)(1.4;-40)(-0.3;-30)(1.4;-20)(1.5;-20)
	\pscurve[linecolor=line,linewidth=1.2pt](1.5;70)(1.4;70)(1.2;55)(1.4;40)(1.5;40)
	\pscurve[linecolor=line,linewidth=1.2pt](1.5;90)(1.4;90)(1.2;115)(1.4;140)(1.5;140)	
   \end{pspicture}  
   & & \in LP^{(*,2k-1)}_{2k}, \\
    \psset{unit=0.8cm}
  \begin{pspicture}[shift=-1.8](-2.25,-1.8)(2.25,1.8)
    	\pscircle[fillstyle=solid, fillcolor=diskin,linewidth=1.2pt](0,0){1.5}
    	\degrees[360]
	\psdot[dotstyle=asterisk,dotscale=1.5](0,0)
	\rput{0}(1.8;-20){\tiny$2k$}
	\rput{0}(1.9;-40){\tiny$2k\!-\!1$}
	\rput{0}(1.8;20){\tiny$j$}	
	\rput{0}(1.8;70){\tiny$i$}
	\rput{0}(1.8;90){\tiny$i\!+\!1$}	
	\rput{0}(1.8;40){\tiny$p$}
	\rput{0}(1.8;60){\tiny$l$}			
	\rput{0}(2;200){\tiny$m$}
	\pscurve[linecolor=line,linewidth=1.2pt](1.5;-40)(1.4;-40)(-0.5;0)(1.4;40)(1.5;40)
	\pscurve[linecolor=line,linewidth=1.2pt](1.5;-20)(1.4;-20)(-0.3;0)(1.4;20)(1.5;20)
	\pscurve[linecolor=line,linewidth=1.2pt](1.5;70)(1.4;70)(1.2;65)(1.4;60)(1.5;60)
	\pscurve[linecolor=line,linewidth=1.2pt](1.5;90)(1.4;90)(1.2;145)(1.4;200)(1.5;200)	
   \end{pspicture}  
   \text{ or } 
     \begin{pspicture}[shift=-1.8](-2.25,-1.8)(2.25,1.8)
    	\pscircle[fillstyle=solid, fillcolor=diskin,linewidth=1.2pt](0,0){1.5}
    	\degrees[360]
	\psdot[dotstyle=asterisk,dotscale=1.5](0,0)
	\rput{0}(1.8;-20){\tiny$2k$}
	\rput{0}(1.9;-40){\tiny$2k\!-\!1$}
	\rput{0}(1.8;20){\tiny$j$}	
	\rput{0}(1.8;70){\tiny$i$}
	\rput{0}(1.8;90){\tiny$i\!+\!1$}			
	\rput{0}(2;200){\tiny$m$}
	\pscurve[linecolor=line,linewidth=1.2pt](1.5;-40)(1.4;-40)(-0.5;15)(1.4;70)(1.5;70)
	\pscurve[linecolor=line,linewidth=1.2pt](1.5;-20)(1.4;-20)(-0.3;0)(1.4;20)(1.5;20)
	\pscurve[linecolor=line,linewidth=1.2pt](1.5;90)(1.4;90)(1.2;145)(1.4;200)(1.5;200)	
   \end{pspicture}  
   & & \in LP^{(*,2k)}_{2k}.
	\end{align*}
	    \caption{Pre-images of $L\in LP^{(*,i)}_{2k}$ for the action of $e_{2k-2}$.}\label{fig-preimage2k}
\end{figure}
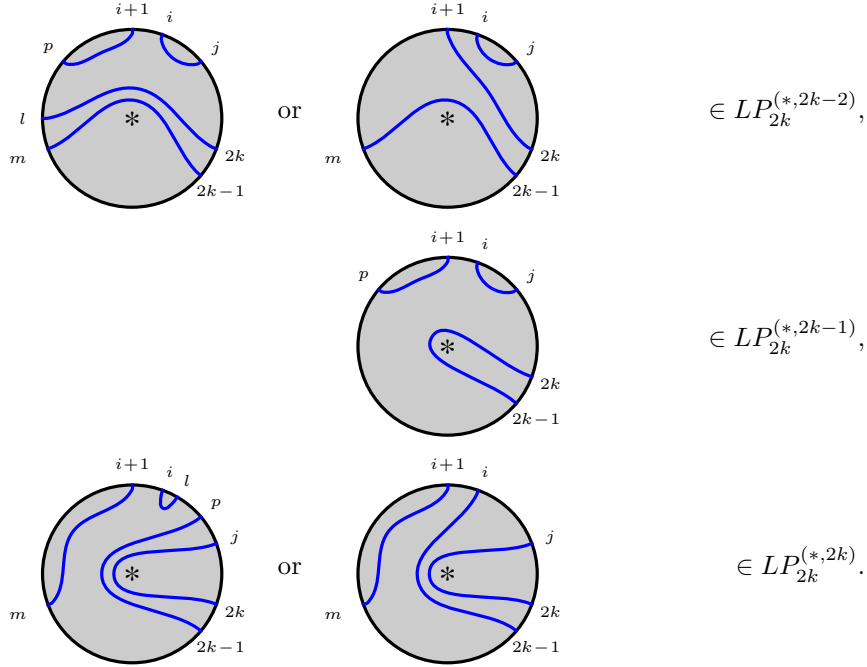

\newpage

\end{document}